\documentclass[lineo,english, 11pt, a4paper, twoside, openright, oldfontcommands]{article}

\usepackage{import}
\usepackage{style}
\usepackage{makeidx}
\usepackage{natbib}
\usepackage{bbm}
\usepackage{pdfpages}
\usepackage{amsmath}

\makeindex
\title{Quasi-parametric rates for Sparse Multivariate Functional Principal Components Analysis}
\author{Ryad Belhakem}
\date{\today}

\begin{document}
\maketitle
\begin{abstract}
This work aims to give non-asymptotic results for estimating the first principal component of a multivariate random process.
We first define the covariance function and the covariance operator in the multivariate case. We then define a projection operator. This operator can be seen as a reconstruction step from the raw data in the functional data analysis context. Next, we show that the eigenelements can be expressed as the solution to an optimization problem, and we introduce the LASSO variant of this optimization problem and the associated plugin estimator. Finally, we assess the estimator's accuracy.
We establish a minimax lower bound on the mean square reconstruction error of the eigenelement, which proves that the procedure has an optimal variance in the minimax sense.

\end{abstract}

\section{Introduction}
    We consider the $D-$multivariate setting and our object of interest is functional data observed on a fixed discretization grid defined as $\{t_h=\frac{h}{p-1};h=0,\dots, p-1\}$. For $D\in\mathbb{N}^*$, we assume the observations to be corrupted by random noise. Hence we consider the following statistical model: We have for all $i\in\{1,\dots,n\}$ and all $h\in\{0,\dots,p-1\}$:
    
\begin{align}\label{model:multi}
     \textbf{Y}_i(t_h)&=(Y_{i,1}(t_h),\dots,Y_{i,D}(t_h))\nonumber \\&= (Z_{i,1}(t_h)+\epsilon_{i,1,h},\dots,Z_{i,D}(t_h)+\epsilon_{i,D,h})\nonumber\\&=\textbf{Z}_i(t_h)+\bf{\epsilon}_{i,h},
    \end{align}
    meaning for each component $d\in\{1,\dots,D\}$ and for all $i\in\{1,\hdots,n\}$ we have:
      \begin{equation*}
    	Y_{i,d}(t_h)=Z_{i,d}(t_h)+\varepsilon_{i,d,h}, 
    \end{equation*}
     with the $\bf{\varepsilon}_{i,h}$'s being independent centered Gaussian vectors of errors with covariance $\sigma^2I_D$. The errors are assumed to be independent of the $\textbf{Z}_i$'s and the $\textbf{Z}_i$'s are independent and identically distributed with the same distribution as $\textbf{Z}$. Here, we assume that the grid is fixed and regular with $p$ points, note that our approach could be generalized to the case where $|t_h-h/p|\leq 1/p$. With the modern ability to record multiple data streams at a very fine timescale, the need for new tools to perform analysis on all these streams at once has arisen. From a functional data analysis (FDA) perspective, these data can be considered as realizations of a multivariate random process on a dense grid. In this context, Multivariate functional principal components analysis (MfPCA) provides an interesting analysis framework.
As MfPCA is an extension of functional principal component analysis (fPCA) to the multivariate case, it inherits most of its properties
regarding representation. Thus as fPCA, it is used as a preprocessing tool. For example, we can cite the modeling and forecasting of multi-population mortality (see \cite{kakin}) and identification of biomarkers for accurate diagnosis of Alzheimer's disease in an early stage (see \cite{happ}). The notable advantage of such an approach is its capacity to handle all the data presented at once, i.e., accounting for variability through time (intra-functional correlations) and space (correlations between functions).
    
    \subsection{Motivation and model}

To the best of our knowledge, only a few articles investigate convergence rates. In \cite{happ} a truncation approach was adopted; this can be formulated as a two-step procedure. The authors first fix $({\phi}_\lambda)_{\lambda\in\Lambda}$ an orthonormal basis of the space considered, then they choose $M$ components, i.e. $\{\phi_{\lambda_1},\dots,\phi_{\lambda_M}\}$, and they project the curves on those $M$ components. Finally they compute the eigenfunction after the projection step. The obtained rates using this approach are at best of the order of $O(\frac{M^3}{n})$.
In (\cite{Hsing} and \cite{Chiou14}) a nonparametric approach was adopted, based on kernel estimation, assuming that the covariance is twice differentiable. This results in rates depending on the size of the bandwidth $h$ of the order of $O_p(h^4 +\frac{\log(n)}{nh})$.
In \cite{BPRR21} we showed that in the univariate case, assuming $Z$ to be a $\alpha-$Hölder continous function, the non-penalised estimator converges at minimax rates ($O(n^{-1})+O(p^{-2\alpha})$), however in the $D$-dimensional context, the same approach gives rates of order ($D^2 n^{-1}+D^2p^{-2\alpha}$). So, in a high dimensional context where $D\gg n$, those rates are non-relevant.\\
The motivation of this work is to define an estimation process for the principal component that produces relevant results despite the high dimensional setting, and to investigate the effect of the noise (which is assumed to be Gaussian and i.i.d), $n$ the number of replicates, $p$ the size of the grid, the sparsity, and the regularization on the estimator of the first principal component in the multivariate setting.\\
In this article, we establish a minimax lower bound on the reconstruction error of the first eigenelement. Next, we extend the estimation process seen in \cite{BPRR21} to the multivariate case. Once the estimator of the covariance operator is defined, we focus our research on establishing a new optimization problem that computes the first principal component. Finally, we base our methodology on a Lasso-penalized M-estimator. We show that any stationary point achieves the minimax optimal variance. The primary conditions we impose are: regularity of $\textbf{Z}$ to be at least a $\alpha-$Hölder continuous function, sparsity in the first eigenfunction that cannot be larger than $\frac{\sqrt{n}}{\log(pD)}$, and the dimension $D$ cannot be larger than $p^\alpha$.

\subsection{Organization of the article}
We first define $\mathbb{H}$ and the covariance operator associated with $\textbf{Z}$ in Section \ref{sec:definition}. We then define the regularity class of the curve $\textbf{Z}$ in Section \ref{sec:reg:multi}. In Section \ref{sec:lowerbound:multi} we show a lower bound in the minimax sense on this class. We define the optimization problem and its LASSO variant and the resulting estimator in Section \ref{sec:estim:multi}, we then establish an upper bound in Section \ref{sec:upperspecific:multi} for histograms. Finally proofs are presented in Section \ref{sec:multi:proof}.

\section{Definitions}\label{sec:definition}
   In what follows, we define the space $\mathbb{H}$, the scalar product of $\mathbb{H}$ and the norm associated with it and we explicit the form of the covariance function and the covariance operator.
   \begin{definition}\label{def:space:H}
 Let $D\in\mathbb{N}^*$, we consider $\textbf{f}=(f_1,\dots,f_D)^T, \text{with all $f_d$'s} \in L^2([0,1])$ \\and $\textbf{g}=(g_1,\dots,g_D)^T, \text{with all $g_d$'s}\in L^2([0,1])$, two vectors of functions. We define $\langle\cdot,\cdot \rangle_\mathbb{H}$ such that: 
 \begin{align*}
     \langle\textbf{f},\textbf{g}\rangle_\mathbb{H}&=\sum_{d=1}^D\int_0^1 f_d(t)g_d(t)dt,\\
     \|\textbf{f}\|_\mathbb{H}^2&=\sum_{d=1}^D\|f_d\|_{\L^2}^2= \langle\textbf{f},\textbf{f}\rangle_\mathbb{H},
 \end{align*}
 and $\mathbb{H}$ such that:
 $$\mathbb{H}:=\big\{\textbf{f}:[0,1]\rightarrow \mathbb{R}^D;\quad \|\textbf{f}\|_\mathbb{H}<\infty \big\}. $$

   \end{definition}
In the sequel we assume that $\textbf{Z}\in\mathbb{H}$ and we assume that $\textbf{Z}$ is a centered random continuous function such that $\mathbb{E}[\|\textbf{Z}\|_\mathbb{H}^2]\leq \infty$. Note that as defined $(\mathbb{H},\langle\cdot,\cdot\rangle_\mathbb{H},\|\cdot\|_\mathbb{H})$ is a separable Hilbert space. 
\paragraph{Notations:}We denote $\|\cdot\|$ the $\L^2:=L^2([0,1])$-norm and $\|\cdot\|_{\ell_2}$ the $\ell_2$-norm for a vector. For $P$ a probability measure, we denote
$\mathbb{E}$ the associated expectation. We denote $P_\textbf{Z}$ the distribution of the process $\textbf{Z}$
and $\mathbb{E}_\textbf{Z}$ the associated expectation. Also we define the norms on $\mathbb{H}$ such that $ \|\textbf{h}\|_1 = \sum_{d=1}^D\|h_d\|_1$, $\|\textbf{h}\|_\infty = \max_{d\in\{1,\dots,D\}}\|h_d\|_\infty$ and $\|\textbf{h}\|_0=\sum_{d=1}^D \1_{h_d\neq 0}$ for any $\textbf{h}\in\mathbb{H}$. We define $sign(u)=1_{\{u\geq 0\}}-1_{\{u<0\}}$ for any $u\in\mathbb{R}$.
   
   \begin{definition}\label{def:operator}
   We define $K$ the covariance function of $\textbf{Z}$ such that for all $(s,t) \in [0,1]^2$ and for all $(d,d')\in\{1,\dots,D\}^2$:
\begin{eqnarray*}
 K_{d,d'}(s,t): = (K(s,t))_{d,d'} = \mathbb{E}(Z_d(s)Z_{d'}(t)),
\end{eqnarray*}
and 
\[ K(s,t) = \mathbb{E}(\textbf{Z}(s)\textbf{Z}(t)^T)\]
with the associated integral operator $\Gamma$, such that for all $\textbf{f} \in \mathbb{H}$,
\begin{eqnarray*}
\quad \Gamma(\textbf{f})(\bullet) &=& \int_0^1 K(s,\bullet)\textbf{f}(s)ds,
\end{eqnarray*}
 We also define the Hilbert-Schmidt norm and the associated scalar product:
let $\Gamma,\Gamma' \in L(\mathbb{H})$ the space of linear operators on $\mathbb{H}$, when it makes sense, $$\langle \Gamma,\Gamma'\rangle_{HS}=\sum_{i\in \mathbb{N}^*}\langle \Gamma(\textbf{e}_i),\Gamma'(\textbf{e}_i)\rangle_{\mathbb{H}},$$
denoting by $(\textbf{e}_i)_{i\in \mathbb{N}^*}$ an orthonormal basis of $\mathbb{H}$. Note that the scalar product $\langle\cdot,\cdot\rangle_{HS}$ is independent of the choice of the basis,
and $\|\cdot\|_{HS}=\sqrt{\langle\cdot,\cdot\rangle}_{HS}$ is the associated norm.  
   \end{definition}
   The operator $\Gamma$ is well defined since $\mathbb{E}[\|\textbf{Z}\|_\mathbb{H}^2]<\infty$. Also note that, as defined, $K$ is symmetric, in the sense that we have:\[\forall (s,t)\in[0,1]^2\quad K(s,t)=K(s,t)^T,\]
   which implies that the operator $\Gamma$ is a linear self-adjoint operator.
   Note that in this setting, Mercer's theorem is still valid (see \cite{Chiou14}), i.e., there exists a sequence of orthonormal functions $(\textbf{f}_\ell)_{\ell\in\mathbb{N}^*}$ and positive numbers $(\mu_\ell)_{\ell\in\mathbb{N}^*}$ (eigenfunctions and the associated eigenvalues) such that
\begin{align*}
    \Gamma &=\sum_{\ell\in\mathbb{N}^*}\mu_\ell \textbf{f}_\ell\otimes\textbf{f}_\ell\\&=\sum_{\ell\in\mathbb{N}^*} \textbf{g}_\ell\otimes\textbf{g}_\ell,
\end{align*}
where $\textbf{g}_\ell:=\sqrt{\mu_\ell}\textbf{f}_\ell$ for all $\ell\in\mathbb{N}^*$ and $\otimes$ is the tensorial product defined such that for any $\textbf{h},\textbf{f},\textbf{g}\in\mathbb{H}$:
\begin{equation}\label{eq:def:tensorial}
    \quad\textbf{h}\otimes \textbf{f} (\textbf{g})=\langle \textbf{h},\textbf{g}\rangle_\mathbb{H} \textbf{f}.
\end{equation}
 The eigenfunctions $(\textbf{f}_\ell)_{\ell\in\mathbb N}$ are ordered such that the associated eigenvalues sequence is non-increasing. We also suppose that all the eigenvalues are distinct, i.e., for all $\ell\in\mathbb{N}^* \quad\mu_{\ell}>\mu_{\ell+1}$.
\section{Definition of the smoothness class for the functional curve $\textbf{Z}$}\label{sec:reg:multi}
       The convergence rates depend on the underlying smoothness of the process of interest.
       \begin{definition}Let $\alpha\in(0,1)$ and $L>0$. We set
\begin{eqnarray}
\mathcal R_\alpha^{(D)}(L)=&&\left\{\mathbb{P} \text{ probability measure on }\mathbb{H} \text{ (with $\mathbb{P}_d$ its marginal on the d-th direction)} \right.\nonumber\\
&&\text{such that}\left.\forall s,t\in[0,1],\int_{\mathbb{H}}\max_{d\in\{1,\dots,D\}}(z_d(t)-z_d(s))^2d\mathbb{P}_d(z)\leq L|t-s|^{2\alpha}\right\}.\nonumber\\\label{hyp:regZ}
\end{eqnarray}

       \end{definition}
The use of these regularity sets is natural. Indeed, we can, for instance, remark that $P_\textbf{Z}$, the distribution of $\textbf{Z}$, satisfies.
\[
\mathbb{P}_\textbf{Z}\in  \mathcal R_\alpha^{(D)}(L)\Rightarrow \max_{d\in\{1,\dots,D\}}\mathbb E_\textbf{Z}[(Z_d(t)-Z_d(s))^2]\leq L|t-s|^{2\alpha}\quad\forall s,t\in[0,1].
\]
This regularity condition implies that the kernel $K$ is bounded and it is a $\alpha-$Hölder continous function, for any $(s,s',t,t')\in[0,1]^4$:

\begin{equation}\label{alpha:holder:inequality}
    \mathbb{P}_\textbf{Z}\in  \mathcal R_\alpha^{(D)}(L)\Rightarrow \max_{d,d'\in\{1,\dots,D\}} |K_{d,d'}(s,t)-K_{d,d'}(s',t')|\leq (L\|K\|_\infty)^{1/2}(|t'-t|^\alpha+|s'-s|)^{\alpha},\end{equation}
where $\|K\|_\infty=\max_{d,d'\in\{1,\dots,D\}}\|K_{d,d'}\|_\infty$ and $\|K_{d,d'}\|_\infty=\sup_{(s,t)\in[0,1]^2}|K_{d,d'}(s,t)|$.
\begin{remark}

The class $R_\alpha^{(D)}(L)$ can be seen as extension of $R_\alpha(L)$ defined in \cite{BPRR21} to the multivariate case. Thus it inherits all of the properties of $R_\alpha(L)$, i.e., Gaussian processes such as Brownian motion and fractional Brownian motion belong to $R_\alpha^{(D)}(L)$ for $\alpha$ and $L$ well chosen.
\end{remark}
   
\section{Lower bound}\label{sec:lowerbound:multi}
The lower bound of the risk for estimating eigenfunctions can be viewed as a benchmark to achieve. We focus on the first eigenfunction, but a similar result, though
more technical, could be obtained for the other eigenfunctions.

\begin{theorem}\label{thm:multi:borne_inf_fp}
Let $s\in\{1,\dots,D\}$ such that $s\leq n$ and $\|\textbf{f}_1\|_0=s$ ,$\alpha\in(0,1)$ and $p\geq\max(3,s^{\frac{1}{2\alpha}})$ and $L>0$ . There exists $n_0$ only depending on $L$ and $\alpha$ such that, for all $n\geq n_0$,
\[
\inf_{\widehat{\textbf{f}}_1}\sup_{P_Z\in\mathcal R_\alpha^{(D)}(L)}\mathbb E[\|\widehat{\textbf{f}}_1-\textbf{f}_1\|_\mathbb{H}^2]\geq c(\sigma) s (p^{-2\alpha}+n^{-1}),
\] 
   where $c(\sigma)$ is a positive constant depending on $\sigma$ and the infimum is taken over all estimators i.e. all measurable function of the observations $\{\textbf{Y}_i(t_h), h=0,\hdots,p-1, i=1,\hdots,n\}$.
\end{theorem}
 The proof is provided in Subsection \ref{proof:multi:borne_inf} page \pageref{proof:multi:borne_inf}.
This result can be seen as an extension of the result obtained in \cite{BPRR21} since the same techniques are used to obtain this lower bound. Thus it also suffers from the same limitations (the lower bound assumes a Gaussian distribution and is only shown for the first eigenfunction) but the lower bound of the minimax risk by $sp^{-2\alpha}$ is valid whatever the distribution of $\textbf{Z}$. 

\section{Definition of the estimators of the eigenelements}\label{sec:estim:multi}

Following the approach adopted in \cite{BPRR21}, we introduce the estimator of the covariance operator $\Gamma$ from the data. Let $(\phi_\lambda)_{\lambda \in \Lambda_M}$ be an orthonormal system of $\L^2$ where $\Lambda_M$ is a finite set of cardinal $M$. In the following, we will consider histograms only. Then, in the setting of model \eqref{model:multi}, we first reconstruct the observed curve on the entire interval $[0,1]$, and we define, for $i=1,\dots,n$ and $d=1,\dots,D$,
\[
\widetilde Y_{i,d}(t) := \sum_{\lambda\in\Lambda_M}\widetilde y_{i,d,\lambda} \phi_\lambda(t)
,\quad
\widetilde y_{i,d,\lambda}:=\frac1p\sum_{h=0}^{p-1}Y_{i,d}(t_h)\phi_\lambda(t_h), \quad t\in[0,1],
\]
where $\widetilde{y}_{i,d,\lambda}$ is an approximation of $\langle Z_{i,d},\phi_\lambda\rangle$. We then define the functions $\widetilde{\textbf{Y}}_i\in\mathbb{H}$ and the vectors $\widetilde{\textbf{y}}_i\in\mathbb{R}^{MD}$ as follows: 
\[(\widetilde{\textbf{Y}}_i)_d:=\widetilde{Y}_{i,d}\quad,(\widetilde{{y}}_{i,d})_{\lambda}:=\widetilde{y}_{i,d,\lambda}\quad\text{and } (\widetilde{\textbf{y}}_i)_{d}:=\widetilde{y}_{i,d}.\]
Thus $\widetilde{\textbf{Y}}_1\hdots,\widetilde{\textbf{Y}}_n$ is a smoothed version of the raw data. Then we define a natural estimator of the covariance kernel $K$ as follows:

\begin{equation}\label{eq:Kestim}
   \widehat{ K}_\phi(s,t)=\frac1n\sum_{i=1}^n  \widetilde{\textbf{Y}}_i(t)^T \widetilde{\textbf{Y}}_i(s),\quad (s,t) \in [0,1]^2,
\end{equation}
and derive from it an estimator of the covariance operator $\widehat{{\Gamma}}_\phi$
\begin{eqnarray*}
 \widehat{{\Gamma}}_\phi(\textbf{f})(\cdot) &=& \int_0^1 \widehat{ K}_\phi(s,\cdot)\textbf{f}(s)ds,\quad \textbf{f} \in \mathbb{H}.
\end{eqnarray*}
Since the kernel $\widehat{K}_\phi$ is symmetric, the operator $\widehat{\Gamma}_\phi$ is self-adjoint and it is also a finite-rank hence compact operator since ${\rm Im}(\widehat{ \Gamma}_\phi)\subset\text{span}\{\widetilde{\textbf{Y}}_1,\hdots,\widetilde{\textbf{Y}}_n\}$. From the spectral theorem, we know that there exists a basis $(\widehat{\textbf{f}}_{\phi,\ell})_{\ell\geq 1}$ of $\mathbb L^2([0,1])$ of eigenfunctions of $\widehat{\Gamma}_\phi$. We denote by $(\widehat{\mu}_{\phi,\ell})_{\ell\geq 1}$ the associated eigenvalues sorted in non-increasing order. Finally, with a slight abuse of notation, we denote the vectors containing all the observations by $(\textbf{Y}_i)_{i\in\{1,\dots,n\}}$, i.e., let $j\in\{1,\dots,pD\}$ we denote by $q$ and $r$ the qotient and the rest of the euclidean division of $j$ by $D$, for any $i\in\{1,\dots,n\}$ $$ (Y_{i})_{j}:=Y_{i,r}(t_q).$$ 
\paragraph{Optimization problem}\label{subsec:optimization}
Our approach mimics the one adopted by \cite{VAN}, which relies heavily on the Taylor expansion in $\mathbb{R}^p$ of the function:
\[\beta\rightarrow\|\Sigma-\beta^T\beta\|_F^2,\quad\beta\in\mathbb{R}^p,\]
where $\Sigma$ is a covariance matrix and $\|\bullet\|_F$ is the Frobenius norm. We first defines a functional counterpart of this function, for this purpose, recall the tensor product notation $\otimes$, given in \ref{eq:def:tensorial}. Note that $\textbf{h}\otimes \textbf{f}$ defines a linear operator verifying $\|\textbf{h}\otimes \textbf{f}\|_{HS} = \|\textbf{h}\|_{\mathbb{H}}\|\textbf{f}\|_{\mathbb{H}}< \infty$. We now introduce the objective function $R$ such that  
\begin{equation}\label{eq:def:R}
    \quad R(\textbf{h}):=\|\Gamma-\textbf{h} \otimes \textbf{h}\|_{HS}^2,\quad\textbf{h}\in\mathbb{H}.
\end{equation}
Extending this approach depends on our ability to use the Taylor expansion and the Taylor-Lagrange remainder in $\mathbb{H}$, however in
\cite{Zeidler85} Chapter 4 (theorem 4.5) an extension to Banach spaces of the Taylor expansion is defined, and since $R$ is a real-valued function, we know that the mean value theorem is valid; thus, the existence of the Taylor-Lagrange remainder is guaranteed.

\begin{theorem}
Let $R:U\subset \mathbb{H}\rightarrow \mathbb{R}$ be a map defined on an open convex set
$U$ of a Banach space $\mathbb{H}$. We suppose that $R$ is $n-$times differentiable. Then the Taylor formula
holds true for all $(\textbf{u},\textbf{h})\in (U\times \mathbb{H})$ such that $\textbf{h}+\textbf{u}\in U$ :
 
$$ R(\textbf{u}+\textbf{h})=R(\textbf{u})+\sum_{k=1}^{n-1}\frac{dR_\textbf{u}^{(k)}{(\textbf{h}^k)}}{k!}+\Xi_n,$$ 
where we set $dR_\textbf{u}^{(k)}(\textbf{h}^k):=dR_{\textbf{u}}^{(k)}(\textbf{u})\underbrace{(\textbf{h},\dots,\textbf{h})}_{k-\text{times}}$, and the remainder has the following form:

$$\Xi_n:=\int_0^1\frac{(1-\tau)^{n-1}}{(n-1)!}dR_{\textbf{u}+\tau \textbf{h}}^{(n)}(\textbf{h}^n) d\tau.$$

\end{theorem}
Since $R$ is a real-valued function, we can apply the mean value theorem on $\Xi_n$ which means that there exists $\tau^*\in[0,1]$ such that
\begin{align*}
     \int_0^1\frac{(1-\tau)^{n-1}}{(n-1)!}dR_{\textbf{u}+\tau \textbf{h}}^{(n)}(\textbf{h}^n) d\tau&=dR_{\textbf{u}+\tau^* \textbf{h}}^{(n)}(\textbf{h}^n) \int_0^1\frac{(1-\tau)^{n-1}}{(n-1)!} d\tau\\&=\frac{dR_{\textbf{u}+\tau^* \textbf{h}}^{(n)}(\textbf{h}^n) }{n!}.
\end{align*}
This allows us to deduce the following corollary.

\begin{corollary}\label{cor:remainder}
Let $R:U\subset \mathbb{H}\rightarrow \mathbb{R}$ be a map defined on an open convex set
$U$ of a Banach space $\mathbb{H}$. We suppose that $R$ is $n-$times differentiable. Then the Taylor formula
holds true for all $(\textbf{u},\textbf{h})\in (U\times \mathbb{H})$ such that $\textbf{h}+\textbf{u}\in U$. There exists $\tau^*\in[0,1]$ such that:
 
 \begin{equation}\label{eq:taylor-lag}
      R(\textbf{u}+\textbf{h})=R(\textbf{u})+\sum_{k=1}^{n-1}\frac{dR_\textbf{u}^{(k)}(\textbf{h}^k)}{k!}+\Xi_n,
 \end{equation}
where we set $dR_\textbf{u}^{(k)}(\textbf{h}^k):=dR_\textbf{u}^{(k)}(\textbf{u})\underbrace{(\textbf{h},\dots,\textbf{h})}_{k-\text{times}}$, and the remainder has the following form

$$\Xi_n:=\frac{dR_{\textbf{u}+\tau^* \textbf{h}}^{(n)}(\textbf{h}^n) }{n!}.$$
\end{corollary}

\begin{remark}
In what follows we will use Corollary \ref{cor:remainder} with $n=2$ and $U=\mathcal{B}(\eta)$ (a ball defined later). So there exists $ \tau^*\in]0,1[$ such that for all $(\textbf{u},\textbf{h})\in (U\times \mathbb{H});\textbf{h}+\textbf{u}\in U$:
\begin{equation}\label{eq:diff:R:taylor}
    R(\textbf{u}+\textbf{h})=R(\textbf{u})+dR_{\textbf{u}}(\textbf{h})+\frac{dR_{\textbf{u}+\tau^* \textbf{h}}(\textbf{h})}{2}.
\end{equation}

\end{remark}	
One key element of the proof is the sub-differential of the $1$-norm introduced earlier. First, we recall the definition of a sub-differential, and we will show its existence and compute its value.
	\begin{definition}
	      Let $\textbf{F}:\mathbb{H}\rightarrow\mathbb{R}$ be a convex function, a sub-differential at point $\textbf{h}_0$ is a function denoted $\partial \textbf{F}(\textbf{h}_0)$ such that;
	      
	      \[\forall \textbf{h}\in\mathbb{H}\quad \textbf{F}(\textbf{h})-\textbf{F}(\textbf{h}_0)\geq \langle\partial \textbf{F}(\textbf{h}_0),\textbf{h}-\textbf{h}_0\rangle_\mathbb{H}.\]
	\end{definition}
In what follows we compute $\partial\|\bullet\|_1$. Let $\textbf{g},\textbf{h}\in\mathbb{H}$. First we define the $Sign$ function as :
\begin{align*}
    Sign(\textbf{g})&=(sign(g_1),...,sign(g_D)).
\end{align*}
Note that we have the following equality for any $g\in L^2([0,1])$:
\begin{align*}
    \|g\|_1 =\max_{\|h\|_\infty\leq 1}\langle g,h\rangle_{L^2([0,1])}=\langle sign(g),g\rangle_{L^2([0,1])}.
\end{align*}
This implies that :
\begin{align}
    \|\textbf{g}\|_1 &= \sum_{d=1}^D \|g_d\|_1\nonumber\\&=\sum_{d=1}^D \max_{\|h_d\|_\infty\leq 1}\langle g_d,h_d\rangle_{L^2([0,1])}\nonumber\\&=\max_{\|\textbf{h}\|_\infty\leq 1}\langle \textbf{g},\textbf{h}\rangle_{\mathbb{H}}\leq\langle Sign(\textbf{g}),\textbf{g}\rangle_{\mathbb{H}}.\nonumber
\end{align}
However, since $\langle Sign(\textbf{g}),\textbf{g}\rangle_{\mathbb{H}}=\|\textbf{g}\|_1$ and $\|Sign(\textbf{g})\|_\infty=1$ it implies that:
\begin{equation}
   \|\textbf{g}\|_1 =\max_{\|\textbf{h}\|_\infty\leq 1}\langle \textbf{g},\textbf{h}\rangle_{\mathbb{H}}=\langle Sign(\textbf{g}),\textbf{g}\rangle_{\mathbb{H}} \label{norm-1}.
\end{equation}
Thus $Sign(\textbf{h})$ is a good candidate for the subdifferential of $\|\textbf{h}\|_1$. Let $\textbf{g},\textbf{h}\in\mathbb{H}$. We show that 
\begin{align*}
    \forall \textbf{h}\in \mathbb{H}\quad  \|\textbf{h}\|_1-\|\textbf{g}\|_1\geq \langle Sign(\textbf{g}),\textbf{h}-\textbf{g}\rangle_\mathbb{H}.
\end{align*}
However since $ \|\textbf{g}\|_1 =\langle Sign(\textbf{g}),\textbf{g}\rangle_{L^2([0,1])}$, it is sufficient to show that :

\begin{align*}
    \forall \textbf{h}\in \mathbb{H}\quad \|\textbf{h}\|_1\geq \langle Sign(\textbf{g})),\textbf{h}\rangle_\mathbb{H},
\end{align*}
which is a consequence of Equation \eqref{norm-1}. Thus for all $\textbf{h}\in\mathbb{H}$ we have:

\begin{equation}\label{eq:subdiif:norme:1}
    \partial\|\textbf{h}\|_1=Sign(\textbf{h}).
\end{equation}
We now focus on $R$ given in \eqref{eq:def:R} and explicit its successive derivatives.
\begin{proposition}\label{prop:differential:multi}Let $\textbf{h}\in\mathbb{H}$. Denoting by $\dot{R}$ and $\ddot{R}$, the first and the second differentials, we have:
 \begin{align*}
    \dot{R}(\textbf{h})&=4(\|\textbf{h}\|_\mathbb{H}^2\textbf{h}-\Gamma(\textbf{h}) )\\\ddot{R}(\textbf{h})&=4(\|\textbf{h}\|_\mathbb{H}^2I+2\textbf{h}\otimes\textbf{h}-\Gamma).
\end{align*}
\end{proposition} 
The proof is provided in Subsection \ref{proof:prop:differential:multi} page \pageref{proof:prop:differential:multi}. To conclude this section, we now present the optimization problem and its empirical LASSO counterpart.

\begin{proposition}(Reformulations of the MfPCA problem)
\label{prop:equivalent_formulations}
The first element of the FPC basis $\textbf{g}_1$ is solution of the following constrained optimization problem:

\begin{eqnarray}
\label{Eq:function optimization on gamma}
\textbf{g}_1 &\in& \underset{\textbf{h}\in\mathbb{H}}{\argmin}\quad R(\textbf{h}). 
\end{eqnarray}    

\end{proposition}Note in addition that the solution is unique up to a sign change for $\textbf{g}_1$.  The proof is provided in Subsection \ref{proof:optime} page \pageref{proof:optime}. 
We denote by $\widehat{R}_\phi$ and ${R}_\phi$ the functions defined for any $\textbf{h}\in\mathbb{H}$ by

\begin{equation}\label{eq:def:hatRphi}
    \widehat{R}_\phi(\textbf{h}):=\|\widehat{\Gamma}_\phi-\textbf{h}\otimes\textbf{h}\|_{HS}^2,
\end{equation} and 
\begin{equation}\label{eq:def:Rphi}
    {R}_\phi(\textbf{h}):=\|{\Gamma}_\phi-\textbf{h}\otimes\textbf{h}\|_{HS}^2,
\end{equation}where $\Gamma_\phi=\mathbb{E}[\widehat{\Gamma}_\phi]$. We investigate the statistical properties of a LASSO variant of this optimization problem, i.e.:

\begin{eqnarray}
\label{Eq:function optimization on gamma lasso}\widehat{{\textbf{g}}} &\in& \underset{\textbf{h}\in\mathcal{B}(\eta),\|\textbf{h}\|_1\leq T}{\argmin} \widehat{R}_\phi(\textbf{h})+\lambda\|\textbf{h}\|_1,
\end{eqnarray}    
where $(T,\lambda)$ are tuning parameters and for some $\eta\geq 0$ we define $\mathcal{B}(\eta):=\{\textbf{h}\in\mathbb{H}, \|\textbf{h}-\textbf{g}_1\|\leq \eta\}$ to be a small ball containing $\textbf{g}_1$. The constraint $\|\textbf{h}\|_1\leq T$ might seem unnecessary or at least redundant with the penalty $\lambda\|\textbf{h}\|_1$, we impose it due to the non-convexity of $\widehat{R}_\phi$. From a theoretical point of view, it is essential to derive guarantees on $\widehat{\textbf{g}}$. In the sequel we will show that $T=O((\frac{n}{\log(pD)})^{1/4})$ is convenient. Thus, this constraint is not restrictive. Note that similar constraints were found in \cite{VAN} and \cite{loh2015regularized}.

\begin{remark}
This optimization problem can be seen as the functional counterpart of the multivariate PCA problem (see \cite{VAN} section 3.1). The Hilbert-Schmidt norm is considered in the literature as an extension of the Frobenius norm to the compact operators. 
\end{remark}

\begin{remark}
The existence and the form of $\mathcal{B}(\eta)$ might seem problematic regarding the feasibility of this approach in practice. To overcome this difficulty we compute a pre-estimate $\textbf{g}_{init}$ using an other algorithm (The one provided by \cite{Hall} was shown to provide good estimate regardless of the dimension $D$). Finally, to create $\mathcal{B}(\eta)$, one would replace $\textbf{g}_{1}$ by $\textbf{g}_{init}$ in $\mathcal{B}(\eta)$.
\end{remark}
    \section{Upper bound for histograms}\label{sec:upperspecific:multi}
    In this paragraph, we specify our results for the case of histograms.
    \begin{definition}
          Let $M$ be an integer such that $M$ divides $p$ and for any $\lambda\in\Lambda_M=\{0,\ldots,M-1\}$
\begin{equation*}
\phi_\lambda(t)=\sqrt{M}1_{I_\lambda}(t),\quad t \in [0,1],
\end{equation*}
with $I_\lambda=[\lambda/M,(\lambda+1)/M)$.
    \end{definition}
Next, we introduce the oracle condition, and we assume it is valid for subsequent Theorem \ref{theo:multi}.\\
 {\bf[Oracle] Oracle condition:}
We assume that $p,M$ and $T$ are such that:
\begin{equation*}
 4\Big( \frac{8\sqrt{L\|K\|_\infty}D}{M^\alpha}+\frac{\sigma^2}{p}+108\big(\widetilde{\mu}_1+\frac{\sigma^2}{p}\big)C_T(C_T+\sqrt{2})\Big)\leq \sqrt{\mu_1}(\rho-8\eta),\label{Oracle condition}
\end{equation*}
where $C_T:=(\|\textbf{g}_1\|_1+T)^2\sqrt{\frac{3\log(pD)}{n}}$, ${\mu}_1$ (resp $\widetilde{\mu}_1$) is the largest eigenvalue of $\Gamma$ (resp $\Gamma_\phi$), $\rho=\sqrt{\mu_1}-\sqrt{\mu_2}$ and $\eta$ is a constant assumed to be smaller then $\frac{\rho}{8}$.
\begin{remark}
Note that, the oracle condition implies that $C_T\leq \frac{\sqrt{\mu_1}(\rho-8\eta)}{432(\widetilde{\mu}_1+\frac{\sigma^2}{p})}$, which means that $C_T$ should be in the worst case of the order of the constant. Since, $C_T$ is of the order of $\frac{T^2\sqrt{\log(pD)}}{\sqrt{n}}$ it means that $T=O((\frac{n}{\log(pD)})^{1/4})$.
\end{remark}

\begin{remark}
Note that our oracle condition imposes a condition on the ratios $\frac{D}{M^{\alpha}}$ and $\frac{\|\textbf{g}\|_1^2\sqrt{\log(pD)}}{\sqrt{n}}$. For our approach to be valid, these two quantities need to be at worst of the order of the constant. This has two implications. The first one is that the grid needs to be dense enough ($M^\alpha\asymp D$) to balance the effect of the dimension $D$. The second concerns the level of sparsity $s$ since $\|\textbf{g}_1\|_1^2\leq\|\textbf{g}_1\|_0\|\textbf{g}_1\|_\mathbb{H}^2 =s\mu_1$. A sufficient condition is to have $\frac{\|\textbf{g}\|_1^2\sqrt{\log(pD)}}{\sqrt{n}}$ at worst of the order of the constant, it means that $s$ can't be larger than $ \sqrt{\frac{n}{\log(pD)}}$. Similar limitation regarding the values of $s$ can be found in \cite{VAN}.
\end{remark}
Let $\lambda_1$ be defined such that:
\begin{equation}\label{definition lambda}
    \lambda_1=4\sqrt{(\widetilde{\mu}_1+\frac{\sigma^2}{p})(\|K\|_\infty+\frac{\sigma^2}{p})}\big(4\sqrt{\frac{\log(DM)}{n}}+\frac{\log(DM)}{n}\big),
\end{equation}

where $\widetilde{\mu}_1+\frac{\sigma^2}{p}$ is the largest eigenvalue of $\frac{\mathbb{E}[\textbf{Y}_1\textbf{Y}_1^T]}{p}.$
\begin{theorem}
\label{theo:multi}Let $\mathbb{P}_\textbf{Z}\in R_\alpha^{(D)}(L)$, $s\in\{1,\dots,D\}$, $\|\textbf{g}\|_1\leq T$ and $\|\textbf{g}_1\|_0=s$. For all $\lambda\geq 4\Big(\|\textbf{g}_1\|_\mathbb{H}(\lambda_1+\frac{8\sqrt{L\|K\|_\infty s}}{M^\alpha})+\frac{\|\textbf{g}_1\|_\infty\sigma^2}{p}+\lambda_1\Big)$ with probability at least $1-\frac{2\log(T)}{pD}-\frac{2}{MD}$ we have :

$$\|\widehat{{\textbf{g}}}-\textbf{g}_1\|_\mathbb{H}^2\leq \frac{4s\lambda^2}{\mu_1(\rho-8\eta)^2},$$

and
$$\|\widehat{{\textbf{g}}}-\textbf{g}_1\|_\mathbb{H}^2\leq \frac{256 s}{\mu_1(\rho-8\eta)^2}\Big(\|\textbf{g}_1\|_\mathbb{H}^2(\lambda_1^2+\frac{64L\|K\|_\infty s}{M^{2\alpha}})+\frac{\|\textbf{g}_1\|_\infty^2\sigma^4}{p^2}+\lambda_1^2\Big).$$
\end{theorem}
The proof is provided in Subsection \ref{proof:theo:multi} page \pageref{proof:theo:multi}.
\begin{remark}
Note that $\lambda_1$ depends on the value of $\|\textbf{g}_1\|_\mathbb{H}$, which is unknown in practice. However, in practice one would replace $\|\textbf{g}_1\|_\mathbb{H}$ by $\|\textbf{g}_{init}\|_\mathbb{H}$ (The pre-estimate). Similar limitations regarding the values of $\lambda_1$ can be found in \cite{VAN}.
\end{remark}
\begin{remark}
We can deduce the following relation between $\|\widehat{{\textbf{f}}}-\textbf{f}_1\|_\mathbb{H}^2$ and $\|\widehat{{\textbf{g}}}-\textbf{g}_1\|_\mathbb{H}^2$:
\begin{align*}
    \|\widehat{{\textbf{g}}}-\textbf{g}_1\|_\mathbb{H}^2&=\|\widehat{{\textbf{g}}}\|_\mathbb{H}^2+\|{{\textbf{g}}}_1\|_\mathbb{H}^2-2\langle\widehat{{\textbf{g}}},{{\textbf{g}}}_1\rangle\\&=\|\widehat{{\textbf{g}}}\|_\mathbb{H}^2+\|{{\textbf{g}}}_1\|_\mathbb{H}^2-2\|\widehat{{\textbf{g}}}\|_\mathbb{H}\|{{\textbf{g}}}_1\|_\mathbb{H}\langle\widehat{{\textbf{f}}},{{\textbf{f}}}_1\rangle\\&=(\|\widehat{{\textbf{g}}}\|_\mathbb{H}-\|{{\textbf{g}}}_1\|_\mathbb{H})^2+\|\widehat{{\textbf{g}}}\|_\mathbb{H}\|{{\textbf{g}}}_1\|_\mathbb{H}(2-2\langle\widehat{{\textbf{f}}},{{\textbf{f}}}_1\rangle)\\&\geq\|\widehat{{\textbf{g}}}\|_\mathbb{H}\|{{\textbf{g}}}_1\|_\mathbb{H}(2-2\langle\widehat{{\textbf{f}}},{{\textbf{f}}}_1\rangle)\\&=\|\widehat{{\textbf{g}}}\|_\mathbb{H}\|{{\textbf{g}}}_1\|_\mathbb{H}\|\widehat{{\textbf{f}}}-\textbf{f}_1\|_\mathbb{H}^2\\&\geq\big(\|{{\textbf{g}}}_1\|_\mathbb{H}-\|\widehat{{\textbf{g}}}-\textbf{g}_1\|_\mathbb{H}\big)\|{{\textbf{g}}}_1\|_\mathbb{H}\|\widehat{{\textbf{f}}}-\textbf{f}_1\|_\mathbb{H}^2.
\end{align*}
Assuming $\|\widehat{{\textbf{g}}}-\textbf{g}_1\|_\mathbb{H}\leq\frac{\|{{\textbf{g}}}_1\|_\mathbb{H}}{2}$ implies that:
\begin{align*}
    \|\widehat{{\textbf{g}}}-\textbf{g}_1\|_\mathbb{H}^2&\geq\frac{\|{{\textbf{g}}}_1\|_\mathbb{H}^2}{2}\|\widehat{{\textbf{f}}}-\textbf{f}_1\|_\mathbb{H}^2,
\end{align*}
and since $\|{{\textbf{g}}}_1\|_\mathbb{H}^2=\mu_1\|{{\textbf{f}}}_1\|_\mathbb{H}^2=\mu_1$ we have:
\begin{align*}
 \|\widehat{{\textbf{f}}}-\textbf{f}_1\|_\mathbb{H}^2\leq \frac{2  \|\widehat{{\textbf{g}}}-\textbf{g}_1\|_\mathbb{H}^2}{\mu_1},
\end{align*}
which allows us to deduce the following corollary. 
\end{remark}
\begin{corollary}\label{cor:multi:erreur} Under the assumptions of Theorem \ref{theo:multi}, taking $M=p$ and $$\lambda=4\Big(\|\textbf{g}_1\|_\mathbb{H}(\lambda_1+\frac{8\sqrt{L\|K\|_\infty s}}{p^\alpha})+\frac{\|\textbf{g}_1\|_\infty\sigma^2}{p}+\lambda_1\Big),$$ we have with probability at least $1-2\frac{\log(T)+1}{pD}$
\begin{equation}\label{eq:cor:2:multi}\|\widehat{{\textbf{f}}}-\textbf{f}_1\|_\mathbb{H}^2\leq \frac{512}{\mu_1^2(\rho-8\eta)^2}\Big(\|\textbf{g}_1\|_\mathbb{H}^2(s\lambda_1^2+\frac{64L\|K\|_\infty s^2}{p^{2\alpha}})+\frac{\|\textbf{g}_1\|_\infty^2s\sigma^4}{p^2}+s\lambda_1^2\Big).\end{equation}
\end{corollary}

\begin{remark} Since $\alpha\leq 1$, the term $s\sigma^4/p^2$ is not larger than the term $s^2L\|K\|_\infty/p^{2\alpha}$ (up to a constant). Denoting by $A$ the subset of $\Omega$ under which Equation \eqref{eq:cor:2:multi} is valid, under the assumptions of Theorem \ref{theo:multi} we have: 
\begin{align*}
    \mathbb{E}\big(\|\widehat{\textbf{f}}-\textbf{f}_1\|^2_\mathbb{H}\big)&=\mathbb{E}\big(\|\widehat{\textbf{f}}-\textbf{f}_1\|^2_\mathbb{H}\1_A\big)+\mathbb{E}\big(\|\widehat{\textbf{f}}-\textbf{f}_1\|^2_\mathbb{H}\1_{A^c}\big)\\&\leq C(K,\widetilde{\mu}_1,\mu_1,\rho,\eta)\log(pD)(\frac{s^2}{p^{2\alpha}}+\frac{s}{n})+2P(A^c)\\&\leq C(K,\widetilde{\mu}_1,\mu_1,\rho,\eta)\log(pD)(\frac{s^2}{p^{2\alpha}}+\frac{s}{n})+4\frac{\log(T)+1}{pD},
\end{align*}
for $C(K,\widetilde{\mu}_1,\mu_1,\rho,\eta)$ a constant that depends on $\|K\|_\infty,\widetilde{\mu}_1,\mu_1,\rho$ and $\eta$. Note that if $T\leq n\leq pD$ \[\frac{\log(T)}{pD}\leq \frac{\log(n)}{n}\leq\frac{\log(pD)}{n},\]
and by Inequality \eqref{eq:se:debarasser:de mu_1} we have that :
\[\widetilde{\mu}_1\leq \frac{8D\sqrt{\|K\|_\infty L}}{(\alpha+1)p^{\alpha}}+\mu_1.\]
Recall that due to the oracle condition $\frac{D}{p^{\alpha}}$ cannot be larger than a constant. Thus :
\[{\sup}_{\substack{\mathbb{P}_\textbf{Z}\in R_\alpha^{(D)}(L),\mu_1\leq \mu,\\\|K\|_\infty\leq L_K,\rho-8\eta\geq c}}\mathbb{E}\big(\|\widehat{\textbf{f}}-\textbf{f}_1\|^2_\mathbb{H}\big)\leq C(L_K,\mu,c)\log(p)(\frac{s^2}{p^{2\alpha}}+\frac{s}{n}),\]
for $C(L_K,\mu,c)$ a constant that depends on $L_K,\mu$ and $c$. This upper bound matches the lower bound up to a $s$ term. Note that the variance term is optimal (up to a $\log$ term). However, the lower bound does not match the upper bound regarding the bias. We only managed to narrow down the minimax rate to be between $\frac{s}{p^{\alpha}}$ and $\frac{s^2}{p\alpha}$. In the continuity of the univariate case, parameters $n,p$, and $s$ strongly influence the rates. When $p$ is large enough ($p\geq (sn)^{\frac{1}{2\alpha}})$, then our procedure achieves the rate $\frac{s\log(p)}{n}$ which is often encountered in the sparse parametric setting. Similarly, as in \cite{BPRR21}, the impact of the noise is negligible. Finally, to the best of our knowledge, these rates are new.
\end{remark}

\section{Proofs}\label{sec:multi:proof}
\subsection{Proof of Theorem~\ref{thm:multi:borne_inf_fp} }\label{proof:multi:borne_inf}
To establish Theorem~\ref{thm:multi:borne_inf_fp}, we prove following Propositions~\ref{prop:multi:borne_inf_fpn} and ~\ref{prop:multi:borne_inf_fpp}.
\begin{proposition}\label{prop:multi:borne_inf_fpn}
Under Assumptions of Theorem~\ref{thm:multi:borne_inf_fp}, \[
\inf_{\widehat{\textbf{f}}_1}\sup_{P_Z\in\mathcal R_\alpha^{(D)}(L)}\mathbb E[\|\widehat{\textbf{f}}_1-\textbf{f}_1\|_\mathbb{H}^2]\geq c_1sn^{-1},
\] 
where $c_1>0$ is a constant depending on $\sigma$.
\end{proposition}
\begin{proposition}\label{prop:multi:borne_inf_fpp}
Under Assumptions of Theorem~\ref{thm:multi:borne_inf_fp}, \[
\inf_{\widehat{\textbf{f}}_1}\sup_{P_Z\in\mathcal R_\alpha^{(D)}(L)}\mathbb E[\|\widehat{\textbf{f}}_1-\textbf{f}_1\|_\mathbb{H}^2]\geq c_2sp^{-2\alpha},
\] 
where $c_2>0$ is a universal constant.
\end{proposition}
The result of Theorem~\ref{thm:multi:borne_inf_fp} is deduced from Propositions~\ref{prop:multi:borne_inf_fpn} and ~\ref{prop:multi:borne_inf_fpp}, by taking $c=\frac{1}{2}\min(c_1; c_2)>0.$
   \subsubsection{Proof of Proposition~\ref{prop:multi:borne_inf_fpn}} 

The proof of Proposition~\ref{prop:multi:borne_inf_fpn} follows the general scheme described in \cite{tsybakov_introduction_2009} Section 2. Let 
\[
\phi(t)=e^{-\frac1{1-t^2}}1_{(-1,1)}(t),\quad t\in\mathbb{R} .
\]
We then define 
\[
\varphi(t)=\left\{
\begin{array}{ll}
\phi(4t-3) &\text{ if }t\in[1/2,1),\\
-\phi(4t-1) &\text{ if }t\in(0,1/2),\\  0&\text{ if }t\notin (0,1).
\end{array}
\right.
\]
Both functions $\phi$ and $\varphi$ are $C^\infty$ on $\mathbb R$ with bounded support, then are $\alpha$-H\"older continuous, for all $\alpha>0$. Moreover $\int_0^1\varphi(t)dt=0$. We note $ L_\alpha$ such that, for all $t,u\in\mathbb R$, 
\[
|\varphi(t)-\varphi(u)|\leq L_\alpha |t-u|^{\alpha}.
\]
Let us now define two components for the test eigenfunctions, let $s\in\mathbb{N}^*,S\leq D d$ the level of sparsity and $S$ a subset of $\{1,\dots,D\}$ such that $|S|=s$, we set

\[
\eta_{1,0}^*(t)=\frac{1}{\sqrt{s}}1_{[0,1]}(t),\quad t\in\mathbb{R}
\]
and with \[\varphi_{a,x}(t)=a\varphi(xt),\]
for $a>0$ and $x\geq 1$,
\[
\eta_{1,1}^*(t)=C\left(\eta_{1,0}^*(t)+\sqrt{\frac{1}{n}}\varphi_{a,x}(t)\right),\quad t\in\mathbb{R}
\]
The eigenfunctions are defined as follows, let $m\in S$:

\[    (\textbf{f}_{1,0})_{m=0,\dots,D}= (\eta_{1,0}^*\times 1_{m\in S})_{m=0,\dots,D},\]
note that $\|\textbf{f}_{1,0}\|_\mathbb{H}=\sum_{m\in S}\|\eta_{1,0}^*\|^2=1$, and
\[    (\textbf{f}_{1,1})_{m=0,\dots,D} = ({\eta_{1,1}^*\times 1_{m\in S}})_{m=0,\dots,D},\]
with $C$ such that $\|\textbf{f}_{1,1}\|_\mathbb{H}=1$.  We first determine $C$:
\[
\|\eta_{1,1}^*\|^2=C^2\left(\|\eta_{1,0}^*\|^2+\sqrt{\frac{4}{n}}\int_0^1\varphi_{a,x}(t)dt+\frac{1}{n}\|\varphi_{a,x}\|^2\right)=C^2\left(\frac{1}{s}+\frac{1}{n}\|\varphi_{a,x}\|^2\right). 
\]
Now, since $x\geq 1$,
\[
\|\varphi_{a,x}\|^2=a^2\int_0^1\varphi^2(xt)dt=\frac{a^2}x\int_0^x\varphi^2(t)dt=\frac{a^2}x\int_0^1\varphi^2(t)dt=\frac{a^2}x\|\varphi\|^2.
\]
Thus :

\[
\|\textbf{f}_{1,1}\|_\mathbb{H}^2=\sum_{m\in S}\|\eta_{1,1}^*\|^2=sC^2\left(\frac{1}{s}+\frac{a^2\|\varphi\|^2}{xn}\right). 
\]
Then, we set
\begin{equation}\label{defC:multi}
C:=\left(1+\frac{a^2s}{xn}\|\varphi\|^2\right)^{-1/2}\leq 1,
\end{equation}
so that $\|\textbf{f}_{1,1}\|_\mathbb{H}=1$. Now, for $\xi\sim\mathcal N(0,1)$ and $\mu_{1}^*>0$, we introduce
\[
\textbf{Z}^j(t)=\sqrt{\mu_{1}^*}\xi\textbf{f}_{1,j}(t),\quad j=0,1
\]
and we consider Model~\eqref{model:multi} such that $\textbf{Z}_1^j, \ldots,\textbf{Z}_n^j$ are i.i.d copies of $\textbf{Z}^j$. Let, for $j=0,1$, $P_{j}^{\textbf{Z}}$ the law of $\textbf{Z}^j$. We have for any $(t,u)\in [0,1]^2$,
\begin{align*}
\max_{d\in\{1,\dots,D\}}\int_{C([0,1])}(z_d(t)-z_d(u))^2dP_{j,d}^\textbf{Z}(z)&=\max_{d\in\{1,\dots,D\}}\mathbb E[(Z_d^j(t)-Z_d^j(u))^2]\\
&=\max_{d\in\{1,\dots,D\}}\mu_1^*\mathbb E[\xi^2]\left(f_{1,d,j}^*(t)-f_{1,d,j}^*(u)\right)^2.\\&=\mu_1^*\mathbb E[\xi^2]\left(\eta_{1,j}^*(t)-\eta_{1,j}^*(u)\right)^2.
\end{align*}
We have easily that $P_{0}^\textbf{Z}\in\mathcal R_\alpha^{(D)}(L)$ since $\eta_{1,0}^*$ is constant on $[0,1]$, implying
\[
\max_{d\in\{1,\dots,D\}}\int_{C([0,1])}(z_d(t)-z_d(u))^2dP_{0,d}^Z(z)=0.
\]
We have

\begin{align*}
    \max_{d\in\{1,\dots,D\}}\int_{C([0,1])}(z_d(t)-z_d(u))^2dP_{1,d}^Z(z)&=\mu_1^*\mathbb E[\xi^2]\left(\eta_{1,1}^*(t)-\eta_{1,1}^*(u)\right)^2\\
&=\frac{C^2\mu_1^*}{n}\left(\varphi_{a,x}(t)-\varphi_{a,x}(u)\right)^2\\&=\frac{C^2a^2\mu_1^*}{n}\left(\varphi(xt)-\varphi(xu)\right)^2\\&\leq\frac{C^2 L_\alpha^2 a^2\mu_1^*x^{2\alpha}}{n}\left|t-u\right|^{2\alpha},
\end{align*}
and since $C\leq 1$, $P_{1}^Z\in\mathcal R_\alpha^{(D)}(L)$  if 
\begin{equation}\label{Cond1:multi}
\frac{ L_\alpha^2 a^2\mu_1^*x^{2\alpha}}{n}\leq L.
\end{equation}
This allows to deduce that 
\[
\inf_{\widehat{\textbf{f}}_1}\sup_{P_Z\in\mathcal R_\alpha^{(D)}(L)}\mathbb E[\|\widehat{\textbf{f}}_1-\textbf{f}_1^*\|_\mathbb{H}^2]\geq\inf_{\widehat{\textbf{f}}_1}\sup_{j=0,1}\mathbb E[\|\widehat{\textbf{f}}_1-\textbf{f}_{1,j}^*\|_\mathbb{H}^2],
\] 
and the aim of what follows is to prove a lower bound for $\mathbb E[\|\widehat{\textbf{f}}_1-\textbf{f}_{1,j}^*\|_\mathbb{H}^2]$. 

Let $\widehat{\textbf{f}}_1$ an estimator and $\hat\psi:\mathbf Z\in\mathcal M_{n\times pD}\to\{0,1\}$ the minimum distance test defined by
\[
\hat\psi= {\argmin}_{j=0,1} \|\widehat{\textbf{f}}_1-\textbf{f}_{1,j}^*\|_\mathbb{H}^2,
\]
we have for $j=0,1$,
\[
\|\widehat{\textbf{f}}_1-\textbf{f}_{1,j}^*\|_\mathbb{H}\geq\frac12\|\textbf{f}_{1,\hat \psi}^*-\textbf{f}_{1,j}^*\|_\mathbb{H}. 
\]
Now, since $\int_0^1\eta_{10}^*(t)\varphi_{a,x}(t)dt=0$, if 
\begin{equation}\label{Cond2:multi}
\frac{sa^2}{xn}\leq 1,
\end{equation}
we have $C\geq (1+\|\varphi\|^2)^{-1/2}$, we fix $a=\frac{\sqrt{x}}{\|\varphi\|}$ in what follows, and we have:
\begin{eqnarray*}
\|\textbf{f}_{1,\hat \psi}^*-\textbf{f}_{1,j}^*\|_\mathbb{H}^2&=&\mathbf 1_{\{\hat \psi\neq j\}}\sum_{m\in S}\|\eta_{1,0}^*-\eta_{1,1}^*\|^2=\mathbf{1}_{\{\hat \psi\neq j\}}\sum_{m\in S}\left\|(1-C)\eta_{1,0}^*-\frac{C}{\sqrt n}\varphi_{a,x}\right\|^2\\
&=&\sum_{m\in S}\mathbf 1_{\{\hat \psi\neq j\}}\left(\frac{1}{s}(1-C)^2+\frac{C^2}n\|\varphi_{a,x}\|^2\right)\\
&\geq&\sum_{m\in S}\mathbf 1_{\{\hat \psi\neq j\}}\frac{C^2a^2}{xn}\|\varphi\|^2\geq \mathbf 1_{\{\hat \psi\neq j\}}\frac{a^2s}{xn}\frac{\|\varphi\|^2}{\|\varphi\|^2+1}\\&\geq& \mathbf 1_{\{\hat \psi\neq j\}}\frac{s}{n}\frac{1}{\|\varphi\|^2+1}. 
\end{eqnarray*}
Then, 
\begin{eqnarray}\label{Borneinf:multi}
\inf_{\widehat{\textbf{f}}_1}\sup_{P_Z\in\mathcal R_\alpha^{(D)}(L)}\mathbb E[\|\widehat{\textbf{f}}_1-\textbf{f}_{1,j}^*\|_\mathbb{H}^2]&\geq&\frac{s}{4n(\|\varphi\|^2+1)}\times\inf_{\hat\psi}\max_{j=0,1}\mathbb P(\widehat\psi\neq j).\nonumber\\
\end{eqnarray}
We now prove that the quantity $\inf_{\hat\psi}\max_{j=0,1}\mathbb P(\widehat\psi\neq j)$ can be bounded from below by an absolute positive constant. For this purpose, we control the Hellinger distance between the data generated by the two models. More precisely, we have to prove that for some constant $H^2_{\max}< 2$, we have
\[
H^2((P_{0}^{obs})^{\otimes n},(P_{1}^{obs})^{\otimes n})\leq H^2_{\max}
\]
where $P_{j}^{obs}$ is the law of the random vector 
\[\mathbf Y^{j,obs}:=(Y_1^{j}(t_0),\hdots,Y_1^{j}(t_{p-1}),\dots,Y_D^{j}(t_0),\hdots,Y_D^{j}(t_{p-1}))\] such that $$Y_d^{j}(t_k)=Z_d^j(t_k)+\varepsilon_{k,d}^j$$ with $\varepsilon_{0,1}^0,\hdots,\varepsilon_{p-1,1}^0,\varepsilon_{0,D}^1,\hdots,\varepsilon_{p-1,D}^1\sim_{i.i.d.}\mathcal N(0,\sigma^2)$. 
First remark that 
\[
\mathbf Y^{j,obs}\sim\mathcal N(0,G_{j}),
\]
where $ G_{j}$ is a symmetric block matrix of size $pD\times pD$ with each block matrix being of size $p\times p$. When $j=0$ we have,
\begin{align*}
     G_0 &= \frac{\mu_1^*}{s} \mathbbm{1}_{s\times s}+\sigma^2 I_{pD}
\end{align*}
where $ \mathbbm{1}_{s\times s}=\mathbbm{1}_s\mathbbm{1}_s^T$ and $\mathbbm{1}_s=(\mathbf{1} \times 1_{m\in S})_{m=1,\dots,D}$ and $\1=(1,\dots,1)^T\in\mathbb{R}^p$. When $j=1$ we have,

\begin{align*}
     G_{1}=&\mu_1^*C^2\left(\frac{1}{s}\mathbbm{1}_{s\times s}+\frac{1}{\sqrt{sn}} \left( \mathbbm{1}_s\boldsymbol\varphi_{a,x}^T+\boldsymbol\varphi_{a,x} \mathbbm{1}_s^T\right)+\frac{1}{n}\boldsymbol\varphi_{a,x}\boldsymbol\varphi_{a,x}^T\right)+\sigma^2 I_{pD}
\end{align*}
where $\boldsymbol\varphi_{a,x}=((\varphi_{a,x}(t_0),\hdots,\varphi_{a,x}(t_{p-1}))^T\times 1_{m\in S})_{m=1,\dots,D}$. Taking $\nu$ to be the Lebesgue measure on $\mathbb R^{pD}$ we get 
\[
H^2((P_{0}^{obs})^{\otimes n},(P_{1}^{obs})^{\otimes n})=2-2A(P_{0}^{obs},P_{1}^{obs})^n,
\]
where, in our case where the variables are Gaussian with equal mean vectors, the Hellinger affinity writes (see e.g. \citealt[pp. 45, 46 and 51]{Pardo06}),
\begin{equation}\label{eq:hellinger_affinity:multi}
A(P_{0}^{obs},P_{1}^{obs})=\frac{\det( G_{0}G_{1})^{1/4}}{\det((G_{0}+ G_{1})/2)^{1/2}}. 
\end{equation}
Matrices $G_0$ and $G_1$ can be analyzed in terms of eigenvalues and eigenfunctions, and assuming that $p\geq 3$, we take $x\geq 1$ such that
\[x=\left\{
\begin{array}{cccc}
 1 &\text{ if }&(p-1)/2&\text{ is an integer} \\
\frac{p-1}{p-2} & \text{ if }  &   p/2&\text{ is an integer} 
\end{array}
\right.\]
and
\begin{equation*}
q:=\frac{p-1}{2x}
\end{equation*}
is an integer such that $q\leq (p-1)/2\leq p-1$. In this case,
\begin{eqnarray*}
\mathbbm{1}_s^T\boldsymbol\varphi_{a,x}&=&\boldsymbol\varphi_{a,x}^T \mathbbm{1}_s=\sum_{m\in S}\sum_{k=0}^{p-1}\varphi_{a,x}(t_k)=a\sum_{m\in S}\sum_{k=0}^{p-1}\varphi(xt_k)=a\sum_{m\in S}\sum_{k=0}^{p-1}\varphi\left(\frac{xk}{p-1}\right),\end{eqnarray*}
and
\begin{eqnarray*}\sum_{k=0}^{p-1}\varphi\left(\frac{xk}{p-1}\right)
&=&\sum_{k=0}^{p-1}\phi\left(\frac{4xk}{p-1}-3\right)\mathbf 1_{\left\{\frac{xk}{p-1}\in[1/2,1[\right\}}\\&-&\sum_{k=0}^{p-1}\phi\left(\frac{4xk}{p-1}-1\right)\mathbf 1_{\left\{\frac{xk}{p-1}\in]0,1/2[\right\}}\\
&=&\sum_{k=0}^{p-1}\phi\left(\frac{4xk}{p-1}-3\right)\mathbf 1_{\left\{k\in[(p-1)/(2x),(p-1)/x[\right\}}\\&-&\sum_{k=0}^{p-1}\phi\left(\frac{4xk}{p-1}-1\right)\mathbf 1_{\left\{k\in]0,(p-1)/(2x)[\right\}}\\
&=&\sum_{\ell=-q}^{p-1-q}\phi\left(\frac{4x\ell}{p-1}-1\right)\mathbf 1_{\left\{\ell\in[0,(p-1)/(2x)[\right\}}\\&-&\sum_{k=0}^{p-1}\phi\left(\frac{4xk}{p-1}-1\right)\mathbf 1_{\left\{k\in]0,(p-1)/(2x)[\right\}}\end{eqnarray*}
replacing the variable $k$ in the first sum by $\ell=k-q$ and remarking that $q\leq p-1-q$. We also have
\begin{eqnarray*}
\sum_{k=0}^{p-1}\varphi\left(\frac{xk}{p-1}\right)&=&\sum_{\ell=-q}^{p-1-q}\phi\left(\frac{2\ell}{q}-1\right)\mathbf 1_{\left\{\ell\in[0,q[\right\}}-\sum_{k=0}^{p-1}\phi\left(\frac{2k}{q}-1\right)\mathbf 1_{\left\{k\in]0,q[\right\}}\\
&=&\sum_{\ell=0}^{q-1}\phi\left(\frac{2\ell}{q}-1\right)-\sum_{k=1}^{q-1}\phi\left(\frac{2k}{q}-1\right)\\
&=&\phi(-1)=0,
\end{eqnarray*}
thus 

\[\mathbbm{1}_s^T\boldsymbol\varphi_{a,x}=\boldsymbol\varphi_{a,x}^T \mathbbm{1}_s=0.\]
Note that we also have
$$ \mathbbm{1}_{s\times s}\mathbbm{1}_s=sp\mathbbm{1}_s,\quad \mathbbm{1}_{s\times s}^2=sp\mathbbm{1}_{s\times s}.$$
We set
\begin{align*}
    v_1:&=\frac{1}{\sqrt{sp}}\mathbbm{1}_s, \quad v_2:=\|\boldsymbol\varphi_{a,x}\|_{\ell_2}^{-1}\boldsymbol\varphi_{a,x}=a^{-1}s^{-1/2}\Big(\sum_{k=0}^{p-1}\varphi^2(xt_k)\Big)^{-1/2}\boldsymbol\varphi_{a,x},
\end{align*}
so that $\|v_1\|_{\ell_2}=\|v_2\|_{\ell_2}=1$ and $v_3,\ldots, v_{pD}$ an orthonormal basis of $\text{span}(v_1,v_2)^\bot$. The matrix $$V:=[v_1;v_2;\cdots;v_{pD}],$$
is an orthogonal matrix.
Since $G_0=\frac{\mu_1^*}{s}\mathbbm{1}_{s\times s}+\sigma^2 I_{pD}$, we have for any $k\in\{3,\dots,pD\}$
$$G_0v_k=\sigma^2 v_k \text{ and}\quad G_0 v_1=(p\mu_1^*+\sigma^2)v_1,\quad G_0v_2=\sigma^2 v_2.$$
Similarly, for $G_1$ we have for any $k\in\{3,\dots,pD\}$,
$$G_1 v_k = \sigma^2 v_k,$$and
$$G_1 v_1=(p\mu_1^*C^2+\sigma^2)v_1+ \mu_1^*C^2\sqrt{\frac{p}{n}}\|\boldsymbol\varphi_{a,x}\|_{\ell_2}v_2,$$ $$G_1 v_2= \Big(\frac{\mu_1^*C^2\|\boldsymbol\varphi_{a,x}\|_{\ell_2}^2}{n}+\sigma^2\Big)v_2+ \mu_1^*C^2\sqrt{\frac{p}{n}}\|\boldsymbol\varphi_{a,x}\|_{\ell_2}v_1$$
which means that
$$G_0=V\begin{pmatrix}
p\mu_1^*+\sigma^2 & 0 & \cdots & 0 \\
0 & \sigma^2  & \cdots & 0 \\
\vdots  & \vdots  & \ddots & \vdots  \\
0 & 0& \cdots & \sigma^2
\end{pmatrix} V^T,
$$
and
$$
G_1=V\begin{pmatrix}
p\mu_1^*C^2+\sigma^2 &  \mu_1^*C^2\sqrt{\frac{p}{n}}\|\boldsymbol\varphi_{a,x}\|_{\ell_2}&0& \cdots & 0 \\
\mu_1^*C^2\sqrt{\frac{p}{n}}\|\boldsymbol\varphi_{a,x}\|_{\ell_2}  & \frac{\mu_1^*C^2}{n}\|\boldsymbol\varphi_{a,x}\|_{\ell_2}^2+\sigma^2  &0& \cdots & 0 \\
0&0&\sigma^2&\cdots &0\\
\vdots  & \vdots  & \vdots  & \ddots &  \vdots   \\
0 & 0&0& \cdots & \sigma^2
\end{pmatrix} V^T.
$$
In particular, we have:
$$\frac{G_0+G_1}{2}=V\begin{pmatrix}
\frac{p\mu_1^*}{2}(C^2+1)+\sigma^2 &\frac{\mu_1^*C^2}2\sqrt{\frac{p}{n}}\|\boldsymbol\varphi_{a,x}\|_{\ell_2} &0& \cdots & 0 \\
\frac{\mu_1^*C^2}2\sqrt{\frac{p}{n}} \|\boldsymbol\varphi_{a,x}\|_{\ell_2}&\frac{\mu_1^*C^2}{2n}\|\boldsymbol\varphi_{a,x}\|_{\ell_2}^2+\sigma^2  &0& \cdots & 0 \\
0&0&\sigma^2&\cdots &0\\
\vdots  & \vdots  & \vdots  & \ddots &  \vdots   \\
0 & 0&0& \cdots & \sigma^2
\end{pmatrix} V^T.
$$
We obtain
$$\det(G_0)=(p\mu_1^*+\sigma^2)\sigma^{2(pD-1)}=\sigma^{2pD}(1+\sigma^{-2}p\mu_1^*),$$
\begin{align*}
\det(G_1)&=\sigma^{2(pD-2)}\left((p\mu_1^*C^2+\sigma^2)\Big( \frac{\mu_1^*C^2}{n}\|\boldsymbol\varphi_{a,x}\|_{\ell_2}^2+\sigma^2\Big)-{\mu_1^*}^2C^4\frac{p}{n} \|\boldsymbol\varphi_{a,x}\|_{\ell_2}^2\right)\\
&=\sigma^{2(pD-2)}\left(\sigma^4+p\mu_1^*C^2\sigma^2+\frac{\mu_1^*C^2\sigma^2}{n}\|\boldsymbol\varphi_{a,x}\|_{\ell_2}^2\right)\\
&=\sigma^{2pD}\left(1+p\mu_1^*C^2\sigma^{-2}+\frac{\mu_1^*C^2\sigma^{-2}}{n}\|\boldsymbol\varphi_{a,x}\|_{\ell_2}^2\right).
\end{align*}
and
\begin{align*}
\det((G_0+G_1)/2)&=\sigma^{2(pD-2)}\left(\Big(\frac{p\mu_1^*}{2}(C^2+1)+\sigma^2\Big)\Big(\frac{\mu_1^*C^2}{2n}\|\boldsymbol\varphi_{a,x}\|_{\ell_2}^2+\sigma^2\Big)\right.\\&\left.-\frac{p{\mu_1^*}^2C^4}{4n}\|\boldsymbol\varphi_{a,x}\|_{\ell_2}^2\right)\\
&=\sigma^{2(Dp-2)}\left(\big(\sigma^2+\frac{p\mu_1^*}{2}\big)\sigma^2+\big(\sigma^2+\frac{p\mu_1^*}{2}\big)\frac{\mu_1^*C^2}{2n}\|\boldsymbol\varphi_{a,x}\|_{\ell_2}^2\right.\\&\left.+\frac{p\mu_1^*C^2\sigma^2}{2}\right)\\
&=\sigma^{2pD}\left(\big(1+\frac{p\mu_1^*\sigma^{-2}}{2}\big)+\big(\sigma^{-2}\right.\\&\left.+\frac{p\mu_1^*\sigma^{-4}}{2}\big)\frac{\mu_1^*C^2}{2n}\|\boldsymbol\varphi_{a,x}\|_{\ell_2}^2+\frac{p\mu_1^*C^2\sigma^{-2}}{2}\right).
\end{align*}
We fix the value of the eigenvalue $\mu_1^*$ to be: \[\mu_1^*=\frac{1}{sp},\] 
so that \eqref{Cond1:multi} is satisfied, for $n\geq\frac{8L_\alpha^2}{3L\|\varphi\|^2} \geq \frac{L_\alpha^2x^{2\alpha+1}}{\|\varphi\|^2Lp}$ and $n\geq \|\varphi\|^{-2}$ we have
\[C^2=\left(1+\frac{a^2s}{xn}\|\varphi\|^2\right)^{-1}=\left(1+\frac{s}{n}\right)^{-1}=1-\frac{s}{n}+o\Big(\frac{s}{n}\Big),\]
note that, we assumed $s\leq {n}$ which implies the following
\[C^2=1-\frac{s}{n}+o\Big(\frac{s}{n}\Big),\] ,
\[\frac{C^2}{s}=\frac{1}{s}-\frac{1}{n}+o\Big(\frac{1}{n}\Big)\] and
\[\frac{C^2}{n}=\frac{1}{n}-\frac{s}{n^2}+o\Big(\frac{s}{n^2}\Big)=\frac{1}{n}-\frac{s}{n^2}+o\Big(\frac{1}{n}\Big).\]
Observe that
\begin{equation}\label{eq:limit:up}
    \frac{1}{sp}\|\boldsymbol\varphi_{a,x}\|_{\ell_2}^2=\frac{1}{sp}\sum_{m\in S}\sum_{k=0}^{p-1}a^2\varphi^2(xt_k)=\frac{x}{p\|\varphi\|^2}\sum_{k=0}^{p-1}\varphi^2\Big(\frac{xk}{p-1}\Big) \to 1,
\end{equation}
when $p\to+\infty$, so
\[u_p:=\frac{1}{sp}\|\boldsymbol\varphi_{a,x}\|_{\ell_2}^2\]
is bounded from below and above uniformly in $p$ (and $n$).
Furthermore,
\[\det(G_0)=\sigma^{2pD}(1+\frac{\sigma^{-2}}{s})\]
\begin{align*}
\det(G_1)&=\sigma^{2pD}\left(1+\frac{C^2\sigma^{-2}}{s}+\frac{C^2\sigma^{-2}}{n}u_p\right)\\&=\sigma^{2pD}(1+\frac{\sigma^{-2}}{s})\left(1+\frac{\sigma^{-2}}{n(1+\frac{\sigma^{-2}}{s})}\left(u_p-1\right)+o\Big(\frac{1}{n}\Big)\right)\end{align*}
 \begin{align*}
    	\det((G_0+G_1)/2)&=\sigma^{2pD}\left(1+\frac{\sigma^{-2}}{2s}+\big(\sigma^{-2}+\frac{\sigma^{-4}}{2s}\big)\frac{C^2}{2n}u_p+\frac{C^2\sigma^{-2}}{2s}\right)\\
    	&=\sigma^{2pD}(1+\frac{\sigma^{-2}}{s})\left(1+\frac{\sigma^{-2}}{2n(1+\frac{\sigma^{-2}}{s})}(u_p-1)+ \frac{\sigma^{-4}}{4n(1+\frac{\sigma^{-2}}{s})s}u_p+o\Big(\frac{1}{n}\Big)\right).
    \end{align*}
We fix $\varepsilon>0$. For $p$ large enough \eqref{eq:limit:up} implies that $|u_p-1|\leq\varepsilon$ and using \eqref{eq:hellinger_affinity:multi},
\begin{align*}
A(P_{0}^{obs},P_{1}^{obs})&=\frac{\det( G_{0}G_{1})^{1/4}}{\det((G_{0}+ G_{1})/2)^{1/2}}\\
&\geq\frac{\left(1-\frac{\sigma^{-2}}{n(1+\frac{\sigma^{-2}}{s})}\varepsilon+o\Big(\frac{1}{n}\Big)\right)^{1/4}}{\left(1+\frac{\sigma^{-2}}{2n(1+\frac{\sigma^{-2}}{s})}\varepsilon+  \frac{\sigma^{-4}}{4n(1+\frac{\sigma^{-2}}{s})s}(1+\varepsilon)+o\Big(\frac{1}{n}\Big)\right)^{1/2}}
\end{align*}
implying
\begin{align*}
A(P_{0}^{obs},P_{1}^{obs})^n&\geq\frac{\left(1-\frac{\sigma^{-2}}{n(1+\frac{\sigma^{-2}}{s})}\varepsilon+o\Big(\frac{1}{n}\Big)\right)^{n/4}}{\left(1+\frac{\sigma^{-2}}{2n(1+\frac{\sigma^{-2}}{s})}\varepsilon+  \frac{\sigma^{-4}}{4n(1+\frac{\sigma^{-2}}{s})s}(1+\varepsilon)+o\Big(\frac{1}{n}\Big)\right)^{n/2}}
\end{align*}
so
\begin{align*}
    \liminf_{n\to+\infty}A(P_{0}^{obs},P_{1}^{obs})^n&\geq\exp\Big(-0.5\frac{\sigma^{-2}}{(1+\frac{\sigma^{-2}}{s})}\varepsilon-0.125 \frac{\sigma^{-4}}{(1+\frac{\sigma^{-2}}{s})s}(1+\varepsilon)\Big)\\&\geq \exp\Big(-0.5\frac{\sigma^{-2}}{(1+\frac{\sigma^{-2}}{s})}\varepsilon-0.125 \frac{\sigma^{-4}}{s+\sigma^{-2}}(1+\varepsilon)\Big),
\end{align*}
note that $\exp(-\frac{y^4}{y^2+x})$ is increasing in $x$, which implies that for all $x\geq 1$ we have $\exp(-\frac{y^4}{y^2+x})\geq \exp(-\frac{y^4}{y^2+1})$, hence
\begin{align*}
    \liminf_{n\to+\infty}A(P_{0}^{obs},P_{1}^{obs})^n&\geq \exp\Big(-0.5\frac{\sigma^{-2}}{(1+\frac{\sigma^{-2}}{s})}\varepsilon-0.125 \frac{\sigma^{-4}}{1+\sigma^{-2}}(1+\varepsilon)\Big),
\end{align*}
and the last quantity is positive for any $\varepsilon >0$. This implies that

\[\limsup_{n\to+\infty}H^2((P_{0}^{obs})^{\otimes n},(P_{1}^{obs})^{\otimes n})<2.\]


 \subsubsection{Proof of Proposition~\ref{prop:multi:borne_inf_fpp}} 
The proof is based on Assouad's Lemma and follows the general scheme described in \citep[Section 2]{tsybakov_introduction_2009}. 
Let 
\[
\phi(t)=e^{-\frac1{1-t^2}}1_{(-1,1)}(t), \quad t\in\mathbb{R}.
\]
We then define 
\[
\varphi(t)=\left\{
\begin{array}{ll}
\phi(4t-1) &\text{ if }t\in[0,1/2),\\
-\phi(4t+1) &\text{ if }t\in(-1/2,0),\\
0 &\text{ if }t\notin(-1/2,1/2).
\end{array}
\right.
\]
Both functions $\phi$ and $\varphi$ are $C^\infty$ on $\mathbb R$ with bounded support, then are $\alpha$-H\"older continuous, for all $\alpha>0$. The function $\varphi$ has it support included in $(-1/2,1/2)$ and verifies $\int_{-1/2}^{1/2}\varphi(t)dt=0$. We note $ L_\alpha$ such that, for all $t,u\in\mathbb R$, 
\[
|\varphi(t)-\varphi(u)|\leq L_\alpha |t-u|^{\alpha}.
\]
Let us now define test eigenfunctions. For $\bomega=(w_0,\ldots,w_{p-1})\in\{0,1\}^p$, we set

\begin{align*}
    \eta_{1,\bomega}^*(t)&=C_{\bomega}\left(\gamma+\sum_{k=0}^{p-1}\omega_k \left(p^{-\alpha}\varphi\left(p(t-t_k)-1/2\right)\right)\right),\\ (\textbf{f}_{1,\bomega}(t))_{m\in\{1,\dots,D\}}&=\eta_{1,\bomega}^*(t)\times 1_{m\in S}
 \end{align*}
with $C_{\bomega}$ and $\gamma>0$ two positive constants to be specified later. To be an eigenfunction, $\textbf{f}_{1,\bomega}$ has to be of norm 1, which writes
\begin{eqnarray*}
\|\eta_{1,\bomega}^*\|^2&=&C_{\bomega}^2\int_0^1\left(\gamma+\sum_{k=0}^{p-1}\omega_k\left(p^{-\alpha}\varphi(p(t-t_k)-1/2)\right)\right)^2dt\\
&=&C_{\bomega}^2\left(\gamma^2+2\gamma\sum_{k=0}^{p-1}\omega_k\left(p^{-\alpha}\int_0^1\varphi(p(t-t_k)-1/2)dt\right)\right.\\
&&\left.+\int_0^1\left(\sum_{k=0}^{p-1}\omega_k\left(p^{-\alpha}\varphi(p(t-t_k)-1/2)\right)\right)^2dt\right).
\end{eqnarray*}
Using successively that the support of $\varphi$ is in $(-1/2,1/2)$ and that $\int_{-1/2}^{1/2}\varphi(t)dt=0$, we can see that 
\[
\int_0^1\varphi(p(t-t_k)-1/2)dt=\int_{t_k}^{t_{k+1}}\varphi(p(t-t_k)-1/2)dt=p^{-1}\int_{-1/2}^{1/2}\varphi(t)dt=0,
\]
and 
\[
\int_0^1\left(\sum_{k=0}^{p-1}\omega_k\varphi(p(t-t_k)-1/2)\right)^2dt=\sum_{k=0}^{p-1}\omega_k\int_0^1\varphi^2(p(t-t_k)-1/2)dt=p^{-1}\sum_{k=0}^{p-1}\omega_k\|\varphi\|^2.
\]
This implies that 
\begin{eqnarray*}
\|\eta_{1,\bomega}^*\|^2
&=&C_{\bomega}^2\left(\gamma^2+p^{-2\alpha-1}\|\varphi\|^2\sum_{k=0}^{p-1}\omega_k\right).\\\|\textbf{f}_{1,\bomega}\|_\mathbb{H}^2&=&sC_{\bomega}^2\left(\gamma^2+p^{-2\alpha-1}\|\varphi\|^2\sum_{k=0}^{p-1}\omega_k\right)
\end{eqnarray*}
We then fix the quantity 
\[
C_{\bomega}=\left(s\gamma^2+sp^{-2\alpha-1}\|\varphi\|^2\sum_{k=0}^{p-1}\omega_k\right)^{-1/2},
\]
and $\gamma=\gamma_1 s^{-\frac{1}{2}}$.Remark that $C_{\bomega}$ verifies for $p\geq s^{\frac{1}{2\alpha}}$ (it implies that $sp^{-2\alpha}\leq 1$) we have:
\begin{equation}\label{eq:Comega}
\left(\gamma_1^2+\|\varphi\|^2\right)^{-1/2}\leq\left(s\gamma^2+sp^{-2\alpha}\|\varphi\|^2\right)^{-1/2}\leq C_{\bomega}\leq\gamma^{-1}s^{-\frac{1}{2}}=\gamma_1^{-1}.
\end{equation}
Let us define now the associated law of our observations: for $\xi\sim\mathcal N(0,1)$ and $\mu_{1,\bomega}^*=\frac{L}{2L_\alpha^2C_\omega^2}$, we set
\[
\textbf{Z}_{\bomega}(t)=\sqrt{\mu_{1,\bomega}^*}\xi\textbf{f}_{1,\bomega}(t).
\]
Let $P_{\bomega}^\textbf{Z}$ the law of $\textbf{Z}_\bomega$ we have that $P_{\bomega}^\textbf{Z}\in\mathcal R_\alpha^{(D)}(L)$ since for any $d\in\{1,\dots,D\}$
\begin{eqnarray*}
\int_{C([0,1])}(\textbf{z}(t)-\textbf{z}(s))_d^2dP_{d,\bomega}^\textbf{Z}(\textbf{z})&=&\mathbb E[((\textbf{Z}_{\bomega}(t))_d-(\textbf{Z}_{\bomega}(s))_d)^2]\\&=&\mathbb E[(Z_{d,\bomega}(t)-(Z_{d,\bomega}(s))^2]\\&=&\mu_{1,\bomega}^*(({f}_{1,d,\bomega}(t)-{f}_{1,d,\bomega}(s))^2\mathbb E[\xi^2]\\
&=&\mu_{1,\bomega}^*({f}_{1,d,\bomega}(t)-{f}_{1,d,\bomega}(s))^2\\
&=&\mu_{1,\bomega}^*C_{\bomega}^2\\&&\times\left(\sum_{k=0}^{p-1}\omega_kp^{-\alpha}(\varphi(p(t-t_k)-1/2)-\varphi(p(s-t_k)-1/2))\right)^2. 
\end{eqnarray*}
Then, using the properties of $\varphi$, we have two cases : 
\begin{itemize}
\item If $s,t\in[t_\ell,t_{\ell+1}[$, 
\begin{eqnarray*}
\left(\sum_{k=0}^p\omega_kp^{-\alpha}(\varphi(p(t-t_\ell)-1/2)-\varphi(p(s-t_\ell)-1/2))\right)^2&&\\
&&\hspace{-8cm}=\omega_\ell^2p^{-2\alpha}(\varphi(p(t-t_\ell)-1/2)-\varphi(p(s-t_\ell)-1/2))^2\\
&&\hspace{-8cm}\leq p^{-2\alpha}L_{\alpha}^2(p(t-t_\ell)-p(s-t_\ell))^{2\alpha}=L_{\alpha}^2(t-s)^{2\alpha}.
\end{eqnarray*}
\item If $s\in[t_\ell,t_{\ell+1}[$ and  $t\in[t_{\ell'},t_{\ell'+1}[$ with $\ell\neq\ell'$,
\begin{eqnarray*}
\left(\sum_{k=0}^p\omega_kp^{-\alpha}(\varphi(p(t-t_k)-1/2)-\varphi(p(s-t_k)-1/2))\right)^2&&\\
&&\hspace{-8cm}=\omega_\ell^2p^{-2\alpha}\left(|\varphi(p(t-t_\ell)-1/2)-\varphi(p(s-t_\ell)-1/2)|\right.\\
&&\hspace{-6cm}\left.+\omega_{\ell'}^2p^{-2\alpha}|\varphi(p(t-t_{\ell'})-1/2)-\varphi(p(s-t_{\ell'})-1/2)|\right)^2\\
&&\hspace{-8cm}\leq 2L_{\alpha}^2|t-s|^{2\alpha}.
\end{eqnarray*}
\end{itemize}
Finally
\[
\max_{d\in\{1,\dots,D\}}\int_{C([0,1])}(\textbf{z}(t)-\textbf{z}(s))_d^2dP_{d,\bomega}^\textbf{Z}(\textbf{z})\leq  2\mu_{1,\bomega}^*C_{\bomega}^2 L_\alpha^2|t-s|^{2\alpha}=L|t-s|^{2\alpha}. 
\]
This allows us to deduce that
\[
\inf_{\widehat{\textbf{f}}_1}\sup_{P_Z\in\mathcal R_\alpha^{(D)}(L)}\mathbb E[\|\widehat{\textbf{f}}_1-\textbf{f}_1\|_\mathbb{H}^2]\geq\inf_{\widehat{\textbf{f}}_1}\sup_{\bomega\in\{0,1\}^{p}}\mathbb E[\|\widehat{\textbf{f}}_1-\textbf{f}_{1,\bomega}^*\|_\mathbb{H}^2],
\] 
and the aim of what follows is to prove a lower bound for $\mathbb E[\|\widehat{\textbf{f}}_1-\textbf{f}_{1,\bomega}^*\|_\mathbb{H}^2]$. 

Let $\widehat{\textbf{f}}_1$ an estimator and 
\[
\widehat\bomega\in {\argmin}_{\bomega\in\{0,1\}^p} \|\widehat{\textbf{f}}_1-\textbf{f}_{1,\bomega}^*\|_\mathbb{H}^2,
\]
we have 
\[
\|\widehat{\textbf{f}}_1-\textbf{f}_{1,\widehat\bomega}^*\|_\mathbb{H}\geq\frac12\|\textbf{f}_{1,\widehat\bomega}^*-\textbf{f}_{1,\bomega}^*\|_\mathbb{H}. 
\]
Now, still from the support properties of $\varphi$, and denoting $\omega_k=\widehat\omega_k=0$ if $k>m$,

\begin{eqnarray*}
\|\eta_{1,\widehat\bomega}^*-\eta_{1,\bomega}^*\|^2&&\\
&&\hspace{-3cm}=\sum_{k=0}^{p-1}\int_{t_{k}}^{t_{k+1}}\left(C_{\widehat\bomega}(\gamma+\widehat\omega_{k}p^{-\alpha}\varphi(p(t-t_{k})-1/2))-C_{\bomega}(\gamma+\omega_{k}p^{-\alpha}\varphi(p(t-t_{k})-1/2))\right)^2dt\\
&&\hspace{-3cm}=p^{-1}\sum_{k=0}^{p-1}\int_{-1/2}^{1/2}\left(C_{\widehat\bomega}(\gamma+\widehat\omega_{k}p^{-\alpha}\varphi(u))-C_{\bomega}(\gamma+\omega_{k}p^{-\alpha}\varphi(u))\right)^2du\\
&&\hspace{-3cm}=(C_{\widehat\bomega}-C_{\bomega})^2\gamma^2+\|\varphi\|^2p^{-2\alpha-1}\sum_{k=0}^{p-1}(C_{\widehat\bomega}\widehat\omega_{k}-C_{\bomega}\omega_{k})^2\geq\|\varphi\|^2p^{-2\alpha-1}\sum_{k=0}^{p-1}(C_{\widehat\bomega}\widehat\omega_{k}-C_{\bomega}\omega_{k})^2\\
&&\hspace{-3cm}\geq\|\varphi\|^2p^{-2\alpha-1}\left(\min\{C_{\widehat\bomega}^2,C_\bomega^2\}\sum_{k=0}^{p-1}\mathbf 1_{\{\widehat\omega_k\neq\omega_k\}}^2+(C_{\widehat\bomega}-C_{\bomega})^2\sum_{k=0}^{p-1}\mathbf 1_{\{\widehat\omega_k=\omega_k=1\}}^2\right)\\
&&\hspace{-3cm}\geq\|\varphi\|^2p^{-2\alpha-1}\min\{C_{\widehat\bomega}^2,C_\bomega^2\}\rho(\bomega,\widehat\bomega).
\end{eqnarray*}
where $\rho(\bomega,\bomega')=\sum_{k=0}^{p-1}\mathbf 1_{\omega_k\neq\omega_k'}$ is the Hamming distance on $\{0,1\}^p$.
Thus we have for all values of $p$ such that $p\geq s^\frac{1}{2\alpha}$:
$$\|\textbf{f}_{1,\widehat\bomega}^*-\textbf{f}_{1,\bomega}^*\|_\mathbb{H}^2\geq s \|\varphi\|^2p^{-2\alpha-1}\min\{C_{\widehat\bomega}^2,C_\bomega^2\}\rho(\bomega,\widehat\bomega)\geq s(\gamma_1^2+\|\varphi\|^2)^{-1}\|\varphi\|^2p^{-2\alpha-1}\rho(\bomega,\widehat\bomega)$$

Combining all the inequalities above, we have the existence of a constant $\tilde c=\|\varphi\|^2/(4(\gamma_1^2+\|\varphi\|^2))$ such that
\begin{eqnarray*}
\inf_{\widehat{\textbf{f}}_1}\sup_{P_Z\in\mathcal R_\alpha^{(D)}(L)}\mathbb E[\|\textbf{f}_{1,\widehat\bomega}^*-\textbf{f}_{1,\bomega}^*\|_\mathbb{H}^2]&\geq& \tilde csp^{-2\alpha-1}\inf_{\widehat\bomega}\max_{\bomega\in\{0,1\}^{p}}\mathbb E[\rho(\widehat\bomega,\bomega)]\\
\end{eqnarray*}

By Assouad's Lemma (see e.g. Tsybakov, 2009, Theorem 2.12),  there exists a constant $c>0$ such that 
\[
\inf_{\widehat\bomega}\max_{\bomega\in\{0,1\}^{p}}\mathbb E[\rho(\widehat\bomega,\bomega)]\geq cp,
\]
provided we can prove the following upper bound 
\[
KL((P_{\bomega}^{obs})^{\otimes n},(P_{0}^{obs})^{\otimes n})\leq K_{\max}< +\infty, \text{ for all } \bomega\in\{0,1\}^p,
\]
where $P_{\bomega}^{obs}$ is the law of the random vector $\mathbf Y^{obs}_{\bomega}:=(Y_{1,\bomega}(t_0),\hdots,Y_{D,\bomega}(t_{p-1}))$ such that $\textbf{Y}_{\bomega}(t_j)=\textbf{Z}_{\bomega}(t_j)+\varepsilon(t_j)$ with $\varepsilon(t_0),\hdots,\varepsilon(t_{p-1})\sim_{i.i.d.}\mathcal N(0,I_D\sigma^2)$ and $KL(P,Q)$ is the Kullback-Leibler divergence between two measures $P$ and $Q$. 
However, remark that, for all $\bomega\in\{0,1\}^p$, for all $j=0,\hdots,p-1$ and $d\in\{1,\dots,D\}$,
\[
\textbf{Y}_{\bomega}(t_j)=\textbf{Z}_{\bomega}(t_j)+\varepsilon(t_j)=\sqrt{\mu_{1,\bomega}^*}\xi\textbf{f}_{1,\bomega}(t_j)+\varepsilon(t_j). 
\]
Now
\[
\eta_{1,\bomega}^*(t_j)=C_\bomega\left(\gamma+\sum_{k=0}^{p-1}\omega_k(p^{-\alpha}\varphi(p(t_j-t_k)-1/2))\right)=C_\bomega\gamma,
\]
since $\varphi((p(t_j-t_k)-1/2)=\varphi(-1/2)=0$ if $j=k$ and $\varphi((p(t_j-t_k)-1/2)=0$ if $j\neq k$ by the support properties of $\varphi$ and the fact that $p(t_j-t_k)-1/2=\frac{p}{p-1}(j-k)-1/2\geq \frac{p}{p-1}-1/2\geq 1/2$ if $j>k$ and $p(t_j-t_k)-1/2\leq 1/2$ if $j<k$. 
Hence for all $d\in\{1,\dots,D\}$ and for all $j=0,\dots,p-1$,
\begin{align*}
    {Y}_{d,\bomega}(t_j)&=\sqrt{\mu_{1,\bomega}^*}\xi C_\bomega\gamma+\varepsilon_d(t_j)=\frac{\gamma \sqrt{L}}{\sqrt{2}L_{\alpha}}\xi+\varepsilon_d(t_j)
\end{align*}
the distribution of $\mathbf{Y}^{obs}_{\bomega}$ does not depend on $\bomega$. Therefore, 
$$KL((P_{\bomega}^{obs})^{\otimes n},(P_{0}^{obs})^{\otimes n})=nKL(P_{\bomega}^{obs},P_{0}^{obs})=0.$$

\subsection{Proof of Proposition~\ref{prop:differential:multi}}\label{proof:prop:differential:multi}
In this subsection we will compute the first and the second differential of $R$ defined such that for all $\textbf{h}\in\mathbb{H}$ we have:

\begin{align*}
    R(\textbf{h})&=  \left\|\Gamma-\textbf{h} \otimes \textbf{h}\right\|_{{HS}}^2\\&=\left\|\Gamma\right\|_{{HS}}^2+\left\|\textbf{h} \otimes \textbf{h}\right\|_{HS}^2-2\langle\Gamma,\textbf{h} \otimes \textbf{h}\rangle_{HS}.
\end{align*}
For the sake of simplicity, we first compute the values of $\left\|\textbf{h} \otimes \textbf{h}\right\|_{HS}^2$ and $\langle\Gamma,\textbf{h} \otimes \textbf{h}\rangle_{HS}$ in what follows:
\begin{itemize}
    \item Let $\textbf{h}\in\mathbb{H}\backslash\{0\}$ and $(\textbf{e}_i)_{i\in\mathbb{N}^*}$ be an orthonormal basis of $\mathbb{H}$ such that $\textbf{e}_1=\frac{\textbf{h}}{\|\textbf{h}\|_\mathbb{H}}$, we have:
\begin{align*}
    \|\textbf{h}\otimes \textbf{h}\|_{HS}^2&=\sum_{i\in\mathbb{N}^*}\langle\textbf{h}\otimes\textbf{h}(\textbf{e}_i),\textbf{h}\otimes\textbf{h}(\textbf{e}_i)\rangle_\mathbb{H}\\&=\sum_{i\in\mathbb{N}^*}\langle\textbf{h},\textbf{h}\rangle_\mathbb{H}\langle\textbf{h},\textbf{e}_i\rangle_\mathbb{H}\langle\textbf{h},\textbf{e}_i\rangle_\mathbb{H}\\&=\|\textbf{h}\|_\mathbb{H}^2\sum_{i\in\mathbb{N}^*}\langle\textbf{h},\textbf{e}_i\rangle_\mathbb{H}^2=\|\textbf{h}\|_\mathbb{H}^4.
\end{align*}
\item Similarly for $\langle\Gamma,\textbf{h} \otimes \textbf{h}\rangle_{HS}$ we have:

\begin{align*}
    \langle\Gamma,\textbf{h} \otimes \textbf{h}\rangle_{HS}&=\sum_{i\in\mathbb{N}^*}\langle\Gamma(\textbf{e}_i),\textbf{h}\otimes\textbf{h}(\textbf{e}_i)\rangle_\mathbb{H}\\&=\sum_{i\in\mathbb{N}^*}\langle\Gamma(\textbf{e}_i),\textbf{h}\rangle_\mathbb{H}\langle\textbf{h},\textbf{e}_i\rangle_\mathbb{H}.
\end{align*}
Since $\textbf{e}_1=\frac{\textbf{h}}{\|\textbf{h}\|_\mathbb{H}}$ it implies that for all $i\geq 2$ we have $\langle\textbf{e}_i,\textbf{e}_1\rangle_\mathbb{H}=\langle\textbf{e}_i,\frac{\textbf{h}}{\|\textbf{h}\|_\mathbb{H}}\rangle_\mathbb{H}=0$. Thus we have:

\begin{align*}
    \langle\Gamma,\textbf{h} \otimes \textbf{h}\rangle_{HS}&=\sum_{i\in\mathbb{N}^*}\langle\Gamma(\textbf{e}_i),\textbf{h}\rangle_\mathbb{H}\langle\textbf{h},\textbf{e}_i\rangle_\mathbb{H}\\&=\langle\Gamma(\textbf{e}_1),\textbf{h}\rangle_\mathbb{H}\langle\textbf{h},\textbf{e}_1\rangle_\mathbb{H}\\&=\frac{\langle\textbf{h},\textbf{h}\rangle_\mathbb{H}}{\|\textbf{h}\|_\mathbb{H}^2}\langle\textbf{e}_1,\textbf{h}\rangle_\mathbb{H}\langle\Gamma(\textbf{h}),\textbf{h}\rangle_\mathbb{H}\\&=\langle\Gamma(\textbf{h}),\textbf{h}\rangle_\mathbb{H}.
\end{align*}
Thus, we have for all $\textbf{h}\in\mathbb{H}$:
\begin{align*}
    R(\textbf{h})&=\left\|\Gamma\right\|_{{HS}}^2+\left\| \textbf{h}\right\|_{\mathbb{H}}^4-2\langle\Gamma( \textbf{h}),\textbf{h}\rangle_{\mathbb{H}}.
\end{align*}
\item Note that by definition of $\|\cdot\|_{HS}$ we have:
\begin{align*}
    \|\Gamma\|_{HS}^2& = \sum_{i\in\mathbb{N}^*}\langle \Gamma(\textbf{e}_i),\Gamma(\textbf{e}_i)\rangle_\mathbb{H}\\& = \sum_{i\in\mathbb{N}^*}\|\Gamma(\textbf{e}_i)\|_\mathbb{H}^2\\&\geq\|\Gamma(\textbf{e}_1)\|_\mathbb{H}^2=\frac{\|\Gamma(\textbf{h})\|_\mathbb{H}^2}{\|\textbf{h}\|_\mathbb{H}^2}.
\end{align*}
Thus, we have for any $\textbf{h}\in\mathbb{H}$:
\begin{equation}
\|\Gamma(\textbf{h})\|_\mathbb{H}^2\leq\|\Gamma\|_{HS}^2\|\textbf{h}\|_\mathbb{H}^2.
    \label{operator:norm:inequality:multi}
\end{equation}
\end{itemize}
Now we compute the first differential, let $\textbf{h}$ and $\textbf{a}$ two functions in $\mathbb{H}$, we have :
\begin{align*}
     R(\textbf{h}+\textbf{a})-R(\textbf{h})&=\left\| \textbf{h}+\textbf{a}\right\|_{\mathbb{H}}^4-2\langle\Gamma (\textbf{h}+\textbf{a}),(\textbf{h}+\textbf{a})\rangle_{\mathbb{H}}-\left\| \textbf{h}\right\|_{\mathbb{H}}^4+2\langle\Gamma(\textbf{h}),\textbf{h}\rangle_{\mathbb{H}}\\&=-4\langle\Gamma(\textbf{h}),\textbf{a}\rangle_{\mathbb{H}}-2\langle\Gamma(\textbf{a}),\textbf{a}\rangle_{\mathbb{H}}+(\left\| \textbf{h}\right\|_{\mathbb{H}}^2+2\langle\textbf{h},\textbf{a}\rangle_\mathbb{H}+\left\| \textbf{a}\right\|_{\mathbb{H}}^2)^2-\left\| \textbf{h}\right\|_{\mathbb{H}}^4\\&=4\|\textbf{h}\|_\mathbb{H}^2\langle\textbf{h},\textbf{a}\rangle-4\langle\Gamma(\textbf{h}),\textbf{a}\rangle_{\mathbb{H}}\\&+\left\| \textbf{a}\right\|_{\mathbb{H}}^4+4\langle\textbf{a},\textbf{h}\rangle_\mathbb{H}^2+4\langle\textbf{a},\textbf{h}\rangle_\mathbb{H}\left\| \textbf{a}\right\|_{\mathbb{H}}^2+2\left\| \textbf{a}\right\|_{\mathbb{H}}^2\left\| \textbf{h}\right\|_{\mathbb{H}}^2-2\langle\Gamma( \textbf{a}),\textbf{a}\rangle_{\mathbb{H}},
 \end{align*}
note that $|\langle\Gamma(\textbf{a}),\textbf{a}\rangle_{\mathbb{H}}|\leq \|\Gamma(\textbf{a})\|_\mathbb{H}\|\textbf{a}\|_{\mathbb{H}}\leq\|\Gamma\|_{HS}\|\textbf{a}\|_\mathbb{H}^2$ and $\langle\textbf{a},\textbf{h}\rangle_\mathbb{H}^2\leq \|\textbf{h}\|_\mathbb{H}^2\|\textbf{a}\|_\mathbb{H}^2$, which implies that:
\begin{align*}
    \| \textbf{a}\|_{\mathbb{H}}^4+4\langle\textbf{a},\textbf{h}\rangle_\mathbb{H}^2+4\langle\textbf{a},\textbf{h}\rangle_\mathbb{H}\| \textbf{a}\|_{\mathbb{H}}^2+2\| \textbf{a}\|_{\mathbb{H}}^2\| \textbf{h}\|_{\mathbb{H}}^2-2\langle\Gamma(\textbf{a}),\textbf{a}\rangle_{\mathbb{H}}=o(\|\textbf{a}\|_\mathbb{H}).
\end{align*}
Finally we have:  $$R(\textbf{a}+\textbf{h})-R(\textbf{h})=4\|\textbf{h}\|_\mathbb{H}^2\langle\textbf{h},\textbf{a}\rangle-4\langle\Gamma( \textbf{h}),\textbf{a}\rangle_{\mathbb{H}}+o(\|\textbf{a}\|_\mathbb{H}),$$ meaning for any $\textbf{h}\in\mathbb{H}$:

\[\dot{R}(\textbf{h})=4\|\textbf{h}\|_\mathbb{H}^2\textbf{h}-4\Gamma( \textbf{h}).\]
Similarly for the second differential, let $\textbf{h},\textbf{a}\in\mathbb{H}$ we have:

\begin{align*}
    \dot{R}(\textbf{h}+\textbf{a})-\dot{R}(\textbf{h})&=4\Big(\|\textbf{h}+\textbf{a}\|_\mathbb{H}^2(\textbf{h}+\textbf{a})-\Gamma(\textbf{h}+\textbf{a})-\|\textbf{h}\|_\mathbb{H}^2\textbf{h}+\Gamma(\textbf{h})\Big)\\&=4\Big((\|\textbf{h}\|_\mathbb{H}^2+\|\textbf{a}\|_\mathbb{H}^2+2\langle\textbf{h},\textbf{a}\rangle_\mathbb{H})(\textbf{h}+\textbf{a})-\|\textbf{h}\|_\mathbb{H}^2\textbf{h}-\Gamma(\textbf{a})\Big)\\&=4\Big(\|\textbf{h}\|_\mathbb{H}^2\textbf{a}+2\langle\textbf{h},\textbf{a}\rangle_\mathbb{H}\textbf{h}-\Gamma(\textbf{a})\Big)\\&+4\|\textbf{a}\|_\mathbb{H}^2\textbf{h}+8\langle\textbf{a},\textbf{h}\rangle_\mathbb{H}\textbf{a}+4\|\textbf{a}\|_\mathbb{H}^2\textbf{a},
\end{align*}
note that  $\Big\|4\|\textbf{a}\|_\mathbb{H}^2\textbf{h}-8\langle\textbf{a},\textbf{h}\rangle_\mathbb{H}\textbf{a}\Big\|_\mathbb{H}\leq 12\|\textbf{h}\|_\mathbb{H}\|\textbf{a}\|_\mathbb{H}^2=o(\|\textbf{a}\|_\mathbb{H})$ and $4\Big(\|\textbf{h}\|_\mathbb{H}^2\textbf{a}-2\langle\textbf{h},\textbf{a}\rangle_\mathbb{H}\textbf{h}-\Gamma(\textbf{a})\Big)$ is a linear function in $\textbf{a}$ thus for any $\textbf{h}\in\mathbb{H}$:

\[\ddot{R}(\textbf{h})=4\Big(\|\textbf{h}\|_\mathbb{H}^2I+2\textbf{h}\otimes\textbf{h}-\Gamma\Big),\]
which ends the proof of Proposition \ref{prop:differential:multi}.
\subsection{Proof of Proposition~\ref{prop:equivalent_formulations}}\label{proof:optime}

Let $\Gamma$ be a covariance operator, $(\textbf{f}_\ell)_{\ell \in \mathbb{N}^*}$ the basis of eigenfunctions of $\Gamma$ and $(\mu_\ell)_{\ell \in \mathbb{N}^*}$ the associated eigenvalues sorted in decreasing order. Remarking that the function to minimize for any $\textbf{h}\in\mathbb{H}$
\[
R(\textbf{h})=\left\|\Gamma-\textbf{h}\otimes\textbf{h}\right\|_{HS}^2
\]
can be written
\[
R(\textbf{h})=\left\|\Gamma-\textbf{h}\otimes\textbf{h}\right\|_{{HS}}^2=\|\Gamma\|_{{HS}}^2-2\langle\Gamma(\textbf{h}),\textbf{h}\rangle_{\mathbb{H}}+\|\textbf{h}\otimes\textbf{h}\|_{{HS}}^2,
\]
we then minimize the functional 
\[
 J(\boldsymbol{h})=\|\textbf{h}\otimes\textbf{h}\|_{{HS}}^2-2\langle\Gamma(\textbf{h}),\textbf{h}\rangle_{\mathbb{H}},\]
On the set $\mathcal U=\mathbb{H}$. The Karush-Kuhn-Tucker (KKT) theory for infinite-dimensional optimization problems can be found, e.g., in \citet[Section 47.10]{Zeidler85}. We first remark that all function $J$ is a convex function. Moreover, it’s G\^ateaux-derivative $J'$  exists, which means that for all $\boldsymbol\eta\in\mathcal U$, there exists $J' (\boldsymbol{h})\in\mathcal{U}$ such that for all $\boldsymbol\eta\in\mathcal U$,

\[
\langle J'(\boldsymbol h),\boldsymbol\eta\rangle=\lim_{t\to 0}\frac{J(\boldsymbol h+t\boldsymbol\eta)-J(\boldsymbol h)}t. 
\]
By \citet[Theorem 47.E]{Zeidler85}, we know that $\bar{\boldsymbol h}$ is a solution to 
\[
\inf_{\boldsymbol h} J(\boldsymbol h)
\]
if and only if 
\begin{equation}\label{eq:KKT}
J'(\overline{\boldsymbol h}) = 0
\end{equation}
Thus for all $\boldsymbol{\eta}\in\mathbb{H}$ we have :

\[
4\|\overline{\boldsymbol{h}}\|_\mathbb{H}^2\langle\overline{\boldsymbol{h}},\boldsymbol{\eta}\rangle-4\langle\Gamma(\overline{\boldsymbol{h}}),\boldsymbol{\eta}\rangle_\mathbb{H}=0
\]
which implies that :
\[
\|\overline{\boldsymbol{h}}\|_\mathbb{H}^2\overline{\boldsymbol{h}}=\Gamma(\overline{\boldsymbol{h}})
\]
which implies that the space ${\rm span}\{\overline{\boldsymbol{h}}\}$ is an invariant subspace of $\Gamma$ of dimension at most $1$. In other words there exists an integers $j_0$  such that
\begin{equation}\label{eq:equiv_spaces}
{\rm span}\{\overline{\boldsymbol{h}}\}={\rm span}\{\textbf{f}_{j_0}\}.
\end{equation}
We then have to prove that $j_0=1$ minimizes the functional, in order to do so note that 
\[
J(\overline{\boldsymbol{h}})= \sum_{j\in\mathbb{N}^*}\mu_j-\mu_{j_0}.
\]
Since $(\mu_j)_{j\geq 1}$ is sorted in decreasing order, this implies that $J(\overline{\boldsymbol{h}})$ is minimal when $j_0=1$.

\subsection{Proof of Theorem \ref{theo:multi}}\label{proof:theo:multi}

As defined, the objective function $R$ is not convex, to illustrate this, we fix $\textbf{a}=a\textbf{f}_1$ and $\textbf{h}=\sqrt{\frac{\mu_1}{4}}\textbf{f}_1$, note that :
\begin{align*}
  \langle\textbf{a},\ddot{R}(\textbf{h})(\textbf{a})\rangle_\mathbb{H}&=4\Big(\|\textbf{h}\|_\mathbb{H}^2\langle\textbf{a},\textbf{a}\rangle_\mathbb{H}+2\langle\textbf{h},\textbf{a}\rangle_\mathbb{H}^2+\langle\textbf{a},\Gamma(\textbf{a})\rangle_\mathbb{H}\Big)\\&=4\Big((a\sqrt{\frac{\mu_1}{4}})^2+2(a\sqrt{\frac{\mu_1}{4}})^2-a^2\mu_1\Big)\\&=4a^2\Big(\frac{\mu_1}{4}+2\frac{\mu_1}{4}-\mu_1)=-a^2\mu_1<0.
\end{align*}
Thus, the naive approach that minimizes the empirical risk, namely $\widehat{R}_\phi$, is cursed by non-convexity since even the theoretical risk $R$ is non-convex, finding the global minimizer of a non-convex function is known to be a challenging task. To overcome this issue, we take advantage of the fact that $R$ is locally convex in the neighborhood of $\textbf{g}_1$ and constraint our optimization problem to find a solution in a small ball $\mathcal{B}(\eta)$ that contains the global minimizer.
Provided that $\eta$ (the radius of $\mathcal{B}$) is small enough, Lemma \ref{Lemma hermitian} below shows that $R$ is convex on $\mathcal{B}(\eta)$. The same approach can be found in \cite{VAN}. 
\begin{lemma}\label{Lemma hermitian}
We denote by $\rho=\sqrt{\mu_1}-\sqrt{\mu_2}$ and assume that $8\eta < \rho$. Then for
all $\textbf{g}$
satisfying $\|\textbf{g}-{\textbf{g}}_1\|_\mathbb{H}<\eta$ we have for all $ \textbf{x}\in \mathbb{H}$:
$$\langle \textbf{x}\ddot{{R}}(\textbf{g}),\textbf{x}\rangle_\mathbb{H}> 4\sqrt{\mu}_1(\rho-8\eta)\|\textbf{x}\|_\mathbb{H}^2.$$

\end{lemma}
The proof is provided in Subsection \ref{proof:Lemma hermitian} page \pageref{proof:Lemma hermitian}. Finally, to obtain the stated result, we will need the following lemma, which controls the error terms appearing in the rest of the proof.
\begin{lemma}\label{upper bound on E(g)}
Assuming the oracle condition is valid, for all $\lambda\geq 4\Big(\|\textbf{g}_1\|_\mathbb{H}(\lambda_1+\frac{8\sqrt{L\|K\|_\infty s}}{M^\alpha})+\|\textbf{g}_1\|_\infty\frac{\sigma^2}{p}+\lambda_1\Big)$ we have with probability at least $1-2\frac{\log(T)+1}{MD}$
$$|E(\textbf{g}_1)|\leq \lambda\|\textbf{g}_1-\widehat{{\textbf{g}}}\|_1+\sqrt{\mu_1}(\rho-8\eta) \|\textbf{g}_1-\widehat{{\textbf{g}}}\|_\mathbb{H}^2,$$
where $\lambda_1=4\sqrt{(\widetilde{\mu}_1+\frac{\sigma^2}{p})(\|K\|_\infty+\sigma^2)}\Big(4\sqrt{\frac{\log(DM)}{n}}+\frac{\log(MD)}{n}\Big)$,
where $E(\textbf{g}_1)$ is defined in \eqref{eq:def:E(g)} page \pageref{eq:def:E(g)}.
\end{lemma}
The proof is provided in Subsection \ref{proof:upper bound on E(g)} page \pageref{proof:upper bound on E(g)}.

\subsubsection{End of the proof of Theorem \ref{theo:multi}}
Since $\widehat{\textbf{g}}$ is a solution to the optimization problem \eqref{Eq:function optimization on gamma lasso}, we know it satisfies the following: 
\begin{equation}\label{eq:sub-diff:theo}
   \langle\widehat{\dot{{{R}}}}_\phi(\widehat{{\textbf{g}}})+\lambda\partial\|\widehat{{\textbf{g}}}\|_1,\textbf{g}_1-\widehat{{\textbf{g}}}\rangle_\mathbb{H}\geq 0,
\end{equation}
where  $\partial\|\widehat{{\textbf{g}}}\|_1$ is the sub-differential of the 1-norm  evaluated at $\widehat{{\textbf{g}}}$. By Taylor
expansion of the loss function $R$ (from equation \eqref{eq:diff:R:taylor}), we obtain:

$$R(\textbf{g}_1)-R(\widehat{{\textbf{g}}})=\langle\dot{R}(\widehat{{\textbf{g}}}),\textbf{g}_1-\widehat{{\textbf{g}}}\rangle_\mathbb{H}+\frac{1}{2}\langle\textbf{g}_1-\widehat{{\textbf{g}}},\ddot{R}(\textbf{g}^*)(\textbf{g}_1-\widehat{{\textbf{g}}})\rangle_\mathbb{H},$$
where there exists $t\in]0,1[$ such that $\textbf{g}^*=t\textbf{g}_1+(1-t)\widehat{{\textbf{g}}}$, thus $\|\textbf{g}_1-{{\textbf{g}}}^*\|_\mathbb{H}\leq \eta$
and applying Lemma \ref{Lemma hermitian} we know that $\langle\textbf{g}_1-\widehat{{\textbf{g}}},\ddot{R}(\textbf{g}^*)(\textbf{g}_1-\widehat{{\textbf{g}}})\rangle_\mathbb{H}>0$, which implies that :

\begin{equation}\label{eq:hermit:theo}
    R(\textbf{g}_1)-R(\widehat{{\textbf{g}}})\geq\langle\dot{R}(\widehat{{\textbf{g}}}),\textbf{g}_1-\widehat{{\textbf{g}}}\rangle_\mathbb{H}.
\end{equation}
Combining inequalities \eqref{eq:sub-diff:theo} and \eqref{eq:hermit:theo} gives the following:

\begin{equation}
    0\geq R(\widehat{{\textbf{g}}})-R(\textbf{g}_1)+\langle\dot{R}(\widehat{{\textbf{g}}})-\widehat{\dot{{{R}}}}_\phi(\widehat{{\textbf{g}}}),\textbf{g}_1-\widehat{{\textbf{g}}}\rangle_\mathbb{H}+\langle\lambda\partial\|\widehat{{\textbf{g}}}\|_1,\widehat{{\textbf{g}}}-\textbf{g}_1\rangle_\mathbb{H}.
\end{equation}
Since from equation \eqref{eq:subdiif:norme:1} we know that $\partial\|\widehat{{\textbf{g}}}\|_1=Sign(\widehat{\textbf{g}})$, by definition we have $\langle\partial\|\widehat{{\textbf{g}}}\|_1,\widehat{{\textbf{g}}}\rangle_\mathbb{H}=\|\widehat{{\textbf{g}}}\|_1$ and since $\|Sign(\widehat{\textbf{g}})\|_\infty\leq 1$ it implies that $\langle\partial\|\widehat{{\textbf{g}}}\|_1,-\textbf{g}_1\rangle_\mathbb{H}\leq \|\textbf{g}_1\|_1$, thus we have:

\begin{align}\label{eq:upper_bound}
   \lambda\|\textbf{g}_1\|_1- \langle\dot{R}(\widehat{{\textbf{g}}})-\widehat{\dot{{{R}}}}_\phi(\widehat{{\textbf{g}}}),\textbf{g}_1-\widehat{{\textbf{g}}}\rangle_\mathbb{H}&\geq R(\widehat{{\textbf{g}}})-R(\textbf{g}_1)+\lambda\|\widehat{{\textbf{g}}}\|_1
\end{align}
we set:
\begin{align}
    E(\textbf{g}_1)& =- \langle\dot{R}(\widehat{{\textbf{g}}})-\widehat{\dot{{{R}}}}_\phi(\widehat{{\textbf{g}}}),\textbf{g}_1-\widehat{{\textbf{g}}}\rangle_\mathbb{H}\label{eq:def:E(g)}\\&=- \langle\dot{R}(\textbf{g}_1)-\widehat{\dot{{{R}}}}_\phi(\textbf{g}_1),\textbf{g}_1-\widehat{{\textbf{g}}}\rangle_\mathbb{H}+\langle\dot{R}(\textbf{g}_1-\widehat{{\textbf{g}}})-\widehat{\dot{{{R}}}}_\phi(\textbf{g}_1-\widehat{{\textbf{g}}}),\textbf{g}_1-\widehat{{\textbf{g}}}\rangle_\mathbb{H}\nonumber.
\end{align}
Thus rewriting the Inequality \eqref{eq:upper_bound}, we get :

\begin{align*}
   R(\widehat{{\textbf{g}}})-R(\textbf{g}_1)+\lambda\|\widehat{{\textbf{g}}}\|_1&\leq \lambda\|\textbf{g}_1\|_1 +E(\textbf{g}_1).
\end{align*}
To conclude we use one last time a Taylor expansion, keeping mind that since $\textbf{g}_1$ is the minimizer in this situation $\dot{R}(\textbf{g}_1)=0$:

$$R(\widehat{{\textbf{g}}})-R(\textbf{g}_1)=\frac{1}{2}\langle\textbf{g}_1-\widehat{{\textbf{g}}},\ddot{R}(\textbf{g}^*_1)(\textbf{g}_1-\widehat{{\textbf{g}}})\rangle_\mathbb{H}$$
where there exists $t\in]0,1[$ such that $\textbf{g}_1^*=t\textbf{g}_1+(1-t)\widehat{{\textbf{g}}}$, thus $\|\textbf{g}_1-{{\textbf{g}}}_1^*\|_\mathbb{H}\leq \eta$
and applying Lemma \ref{Lemma hermitian} we know that :

$$R(\widehat{{\textbf{g}}})-R(\textbf{g}_1)\geq 2\sqrt{\mu_1}(\rho-8\eta) \|\textbf{g}_1-\widehat{{\textbf{g}}}\|_\mathbb{H}^2.$$
Thus we have:
\begin{align*}
   2\sqrt{\mu_1}(\rho-8\eta) \|\textbf{g}_1-\widehat{{\textbf{g}}}\|_\mathbb{H}^2+\lambda\|\widehat{{\textbf{g}}}\|_1&\leq \lambda\|\textbf{g}_1\|_1 +E(\textbf{g}_1).
\end{align*}
Using Lemma \ref{upper bound on E(g)} we know that for all $\lambda\geq 4\Big(\|\textbf{g}_1\|_\mathbb{H}(\lambda_1+\frac{8\sqrt{L\|K\|_\infty s}}{M^\alpha})+\|\textbf{g}_1\|_\infty\frac{\sigma^2}{p}+\lambda_1\Big)$ with probability at least $1-2\frac{\log(T)+1}{MD}$ we have:
$$|E(\textbf{g}_1)|\leq \lambda\|\textbf{g}_1-\widehat{{\textbf{g}}}\|_1+\sqrt{\mu_1}(\rho-8\eta) \|\textbf{g}_1-\widehat{{\textbf{g}}}\|_\mathbb{H}^2,$$
which implies that 

\begin{align*}
   \sqrt{\mu_1}(\rho-8\eta) \|\textbf{g}_1-\widehat{{\textbf{g}}}\|_\mathbb{H}^2+\lambda\|\widehat{{\textbf{g}}}\|_1&\leq  \lambda\|\textbf{g}_1\|_1+ \lambda\|\textbf{g}_1-\widehat{{\textbf{g}}}\|_1.
\end{align*}
Since $\|\textbf{g}_1\|_0=s$. Denoting by $\textbf{g}_1^S$ the projection on those components, we have :

\begin{align*}
    \|\widehat{{\textbf{g}}}\|_1&=\|\widehat{{\textbf{g}}}^S\|_1+\|(\widehat{{\textbf{g}}}^{S})^c\|_1\\
    \|\textbf{g}_1\|_1&=\|\textbf{g}_1^S\|_1\\
    \|\textbf{g}_1-\widehat{{\textbf{g}}}\|_1&=\|(\textbf{g}_1-\widehat{{\textbf{g}}})^S\|_1+\|(\widehat{{\textbf{g}}}^{S})^c\|_1,
\end{align*}
using this, we have

\begin{align*}
   \sqrt{\mu_1}(\rho-8\eta) \|\textbf{g}_1-\widehat{{\textbf{g}}}\|_\mathbb{H}^2&\leq  2\lambda\|(\textbf{g}_1-\widehat{{\textbf{g}}})^S\|_1\\&\leq  2\lambda\sqrt{s}\|(\textbf{g}_1-\widehat{{\textbf{g}}})^S\|_\mathbb{H}\\&\leq  (2\lambda\sqrt{s})\|\textbf{g}_1-\widehat{{\textbf{g}}}\|_\mathbb{H},
\end{align*}
and 
\[\|\textbf{g}_1-\widehat{{\textbf{g}}}\|_\mathbb{H}\leq  \frac{2\lambda\sqrt{s}}{\sqrt{\mu_1}(\rho-8\eta)}\]
which ends the proof of Theorem \ref{theo:multi}.

\subsection{Proof of Lemma \ref{upper bound on E(g)}}
The term $E(\textbf{g}_1)$ encapsulates error terms, we recall its definition:
\begin{align*}
    E(\textbf{g}_1)& =- \langle\dot{R}(\widehat{{\textbf{g}}})-\widehat{\dot{{{R}}}}_\phi(\widehat{{\textbf{g}}}),\textbf{g}_1-\widehat{{\textbf{g}}}\rangle_\mathbb{H}\\&=- \langle\dot{R}(\textbf{g}_1)-\widehat{\dot{{{R}}}}_\phi(\textbf{g}_1),\textbf{g}_1-\widehat{{\textbf{g}}}\rangle_\mathbb{H}+\langle\dot{R}(\textbf{g}_1-\widehat{{\textbf{g}}})-\widehat{\dot{{{R}}}}_\phi(\textbf{g}_1-\widehat{{\textbf{g}}}),\textbf{g}_1-\widehat{{\textbf{g}}}\rangle_\mathbb{H}.
\end{align*}
As consequence of Proposition \ref{prop:differential:multi}, we know that (more details are provided in the rest of the proof): $$\dot{R}(\bullet)-\widehat{\dot{{{R}}}}_\phi(\bullet)=-4\Big(\Gamma(\bullet)-\widehat{\Gamma}_\phi(\bullet)\Big).$$ Thus we need to have control over $\Gamma(\bullet)-\widehat{\Gamma}_\phi(\bullet)$, the lemmas below provide precisely that.
\begin{lemma}\label{Lemma bias infty norm}Let $\mathbb{P}_\textbf{Z}\in R_\alpha^{(D)}(L)$, $s\in\{1,\dots,D\}$ and $\| \textbf{g}_1\|_0=s$, we have 

$$\|({{\Gamma}}_\phi-{\Gamma})(\textbf{g}_1)\|_\infty\leq\frac{8\sqrt{sL\|K\|_\infty}\|\textbf{g}_1\|_\mathbb{H}}{M^\alpha}+\frac{\sigma^2\|\textbf{g}_1\|_\infty}{p}.$$
\end{lemma}
\begin{lemma}
\label{Lemma:multi:bias:oracle}Let $\mathbb{P}_\textbf{Z}\in R_\alpha^{(D)}(L)$ we have for all $ \textbf{g}\in\mathbb{H}$

$$|\langle({\Gamma}_\phi-\Gamma)\textbf{g},\textbf{g}\rangle_\mathbb{H}|\leq \|\textbf{g}\|_\mathbb{H}^2 \Big(\frac{8D\sqrt{\|K\|_\infty L}}{(\alpha+1)M^{\alpha}}+\frac{\sigma^2}{p}\Big).$$
\end{lemma}
Proofs of these lemmas are provided in Subsection \ref{Proof determinstic} page \pageref{Proof determinstic}. Similarly for the random terms.
\begin{lemma}\label{Lemma concentration}
Assuming the observations are Gaussian, we have with probability at least $1-\frac{2}{MD}$

\begin{align*}
    \|(\widehat{{\Gamma}}_\phi-{\Gamma}_\phi)(\textbf{g}_1)\|_\infty&\leq 4\|\textbf{g}_1\|_\mathbb{H}\sqrt{(\widetilde{\mu}_1+\frac{\sigma^2}{p})(\|K\|_\infty+\frac{\sigma^2}{p})}[4\sqrt{\frac{\log(DM)}{n}}+\frac{\log(DM)}{n}]\\&= \|\textbf{g}_1\|_\mathbb{H}\lambda_1,
\end{align*}
where $\widetilde{\mu}_1+\frac{\sigma^2}{p}$ is the largest eigenvalue of $\frac{\mathbb{E}[\textbf{Y}_1^T\textbf{Y}_1]}{p}$.
\end{lemma}
\begin{lemma}\label{lemma:most:important}
Let $J:=\log(T)$ and $\lambda_0=\sqrt{\frac{2\log(2pD)}{n}}$. Then with probability at least $1-2\frac{J+1}{MD}$, it holds for all $ \textbf{g}\in\mathcal{B}(\eta)  ,\|\textbf{g}\|_1\leq T$:
\begin{equation}
 |\langle(\widehat{\Gamma}_\phi-\Gamma_\phi)(\textbf{g}),\textbf{g}\rangle_\mathbb{H}|\leq \lambda_1\|\textbf{g}\|_\mathbb{H}+\mathcal{Q}(\|\textbf{g}\|_1^2,3\sqrt{\widetilde{\mu}_1+\frac{\sigma^2}{p}})\|\textbf{g}\|_\mathbb{H}^2,
\end{equation}
with $\mathcal{Q}=(x,\zeta)=108 \zeta^2[3x^2\lambda_0+\sqrt{6}x\lambda_0]$ and $\lambda_0=\sqrt{\frac{2\log(2pD)}{n}}$.
\end{lemma} 
Proofs of these Lemma are provided in Subsection \ref{Proof random} page \pageref{Proof random}.

Finally we recall the statement of Lemma \ref{upper bound on E(g)}.\\
{\bf Lemma \ref{upper bound on E(g)} }{ \it
 Assuming that the oracle condition is valid, for all $\lambda\geq 4\Big(\|\textbf{g}_1\|_\mathbb{H}(\lambda_1+\frac{8\sqrt{L\|K\|_\infty s}}{M^\alpha})+\|\textbf{g}_1\|_\infty\frac{\sigma^2}{p}+\lambda_1\Big)$ gives that with probability at least $1-2\frac{\log(T)+1}{MD}$
$$|E(\textbf{g}_1)|\leq \lambda\|\textbf{g}_1-\widehat{{\textbf{g}}}\|_1+\sqrt{\mu_1}(\rho-8\eta) \|\textbf{g}_1-\widehat{{\textbf{g}}}\|_\mathbb{H}^2$$
}

\subsubsection{Proof of Lemma \ref{upper bound on E(g)}}\label{proof:upper bound on E(g)}

We focus on the control of $ E(\textbf{g}_1)$:

\begin{align}\label{eq:ineq:E(g)}
    |E(\textbf{g}_1)|\leq \|\dot{R}(\textbf{g}_1)-\widehat{\dot{{{R}}}}_\phi(\textbf{g}_1)\|_\infty\|\textbf{g}_1-\widehat{{\textbf{g}}}\|_1+|\langle\dot{R}({{\textbf{g}}}_1-\widehat{{\textbf{g}}})-\widehat{\dot{{{R}}}}_\phi(\textbf{g}_1-\widehat{{\textbf{g}}}),{{\textbf{g}}}_1-\widehat{{\textbf{g}}}\rangle_\mathbb{H}|.
\end{align}
We first recall the definitions of $\Dot{R}$ and $\widehat{\dot{{{R}}}}_\phi$, let $h\in\mathbb{H}$ we have:
\begin{align*}
    \Dot{R}(\textbf{h}) &=4(\|\textbf{h}\|_\mathbb{H}^2\textbf{h}-\Gamma(\textbf{h})) \\
    \widehat{\dot{{{R}}}}_\phi(\textbf{h})&=4(\|\textbf{h}\|_\mathbb{H}^2\textbf{h}-\widehat{\Gamma}_\phi(\textbf{h})),
\end{align*}
which implies that:
\begin{align*}
    \|\dot{R}(\textbf{g}_1)-\widehat{\dot{{{R}}}}_\phi(\textbf{g}_1)\|_\infty&=4\|(\Gamma-\widehat{\Gamma}_\phi)(\textbf{g}_1)\|_\infty\\&\leq 4\Big(\|(\Gamma-{\Gamma}_\phi)(\textbf{g}_1)\|_\infty+\|(\Gamma_\phi-\widehat{\Gamma}_\phi)(\textbf{g}_1)\|_\infty\Big).
\end{align*}
Thus we need an upper bound on $\|(\Gamma-\widehat{\Gamma}_\phi)(\textbf{g}_1)\|_\infty$ and $\|(\Gamma_\phi-\widehat{\Gamma}_\phi)(\textbf{g}_1)\|_\infty$. However in Lemma \ref{Lemma bias infty norm} and Lemma \ref{Lemma concentration} we showed that:

\begin{align*}
    \|(\Gamma-{\Gamma}_\phi)(\textbf{g}_1)\|_\infty&\leq 8\sqrt{L\|K\|_\infty s}\frac{\|\textbf{g}_1\|_\mathbb{H}}{M^\alpha}+\frac{\sigma^2\|\textbf{g}_1\|_\infty}{p},
\end{align*}
and with probability at least $1-\frac{2}{MD}$ we have :

\begin{align*}
    \|(\Gamma_\phi-\widehat{\Gamma}_\phi)(\textbf{g}_1)\|_\infty&\leq \|\textbf{g}_1\|_\mathbb{H}\lambda_1.
\end{align*}
Finally, we have:

\begin{align}\label{eq:dot:R:upper:bound}
    \|\dot{R}(\textbf{g}_1)-\widehat{\dot{{{R}}}}_\phi(\textbf{g}_1)\|_\infty&\leq 4\|(\Gamma_\phi-\widehat{\Gamma}_\phi)(\textbf{g}_1)\|_\infty+4\|(\Gamma-{\Gamma}_\phi)(\textbf{g}_1)\|_\infty\nonumber\\&\leq4\big(\|\textbf{g}_1\|_\mathbb{H}(\lambda_1+\frac{8\sqrt{L\|K\|_\infty s}}{M^\alpha})+\|\textbf{g}_1\|_\infty\frac{\sigma^2}{p}\big).
\end{align}
Now we focus on the last term of the upper bound \eqref{eq:ineq:E(g)} namely $|\langle\dot{R}({{\textbf{g}}}_1-\widehat{{\textbf{g}}})-\widehat{\dot{{{R}}}}_\phi(\textbf{g}_1-\widehat{{\textbf{g}}}),{{\textbf{g}}}_1-\widehat{{\textbf{g}}}\rangle_\mathbb{H}|$. First of all note that:
\begin{align*}
    |\langle\dot{R}({{\textbf{g}}}_1-\widehat{{\textbf{g}}})-\widehat{\dot{{{R}}}}_\phi(\textbf{g}_1-\widehat{{\textbf{g}}}),{{\textbf{g}}}_1-\widehat{{\textbf{g}}}\rangle_\mathbb{H}|&\leq |\langle\dot{R}_\phi({{\textbf{g}}}_1-\widehat{{\textbf{g}}})-\widehat{\dot{{{R}}}}_\phi(\textbf{g}_1-\widehat{{\textbf{g}}}),{{\textbf{g}}}_1-\widehat{{\textbf{g}}}\rangle_\mathbb{H}|\\&+|\langle\dot{R}({{\textbf{g}}}_1-\widehat{{\textbf{g}}})-\Dot{{{{R}}}}_\phi(\textbf{g}_1-\widehat{{\textbf{g}}}),{{\textbf{g}}}_1-\widehat{{\textbf{g}}}\rangle_\mathbb{H}|.
\end{align*}
Using Lemmas \ref{Lemma:multi:bias:oracle} and \ref{lemma:most:important} we have that:

\begin{align}\label{ineq:2nd part:EG}
    |\langle\dot{R}({{\textbf{g}}}_1-\widehat{{\textbf{g}}})-\Dot{{{{R}}}}_\phi(\textbf{g}_1-\widehat{{\textbf{g}}}),{{\textbf{g}}}_1-\widehat{{\textbf{g}}}\rangle_\mathbb{H}|\leq4 \|{{\textbf{g}}}_1-\widehat{{\textbf{g}}}\|_\mathbb{H}^2(\frac{8\sqrt{L\|K\|_\infty}D}{M^\alpha}+\frac{\sigma^2}{p}),
\end{align}
and with probability at least $1-2\frac{\log(T)+1}{MD} $
\begin{align*}
    |\langle\dot{R}_\phi({{\textbf{g}}}_1-{{\textbf{g}}})-\Dot{{{\widehat{R}}}}_\phi(\textbf{g}_1-\widehat{{\textbf{g}}}),{{\textbf{g}}}_1-\widehat{{\textbf{g}}}\rangle_\mathbb{H}|\leq 4\Big( \mathcal{Q}(\|{{\textbf{g}}}_1-\widehat{{\textbf{g}}}\|_1^2)\|{{\textbf{g}}}_1-\widehat{{\textbf{g}}}\|_\mathbb{H}^2+\lambda_1\|{{\textbf{g}}}_1-\widehat{{\textbf{g}}}\|_1\Big).
\end{align*}
To conclude we work on $\mathcal{Q}(\|\textbf{g}_1-\widehat{{\textbf{g}}}\|_1^2)$ and show this term can be controlled. Indeed, we have  
 
 $$\mathcal{Q}(\|\textbf{g}_1-\widehat{{\textbf{g}}}\|_1^2)=108(\widetilde{\mu}_1+\frac{\sigma^2}{p})\Big[\|\textbf{g}_1-\widehat{{\textbf{g}}}\|_1^4\frac{3\log(pD)}{n}+\|\textbf{g}_1-\widehat{{\textbf{g}}}\|_1^2\sqrt{\frac{6\log(pD)}{n}}\Big],$$
in order to do so note that, the constraint gives that $\|\widehat{{\textbf{g}}}\|_1\leq T$ :

$$\|\textbf{g}_1-\widehat{{\textbf{g}}}\|_1^2\sqrt{\frac{3\log(pD)}{n}}]\leq (\|\textbf{g}_1\|_1+T)^2\sqrt{\frac{3\log(pD)}{n}}],$$
we define $C_T$ such that
$$ C_T:=(\|\textbf{g}_1\|_1+T)^2\sqrt{\frac{3\log(pD)}{n}}].$$
Then it follows 
\begin{align*}
    \|\textbf{g}_1-\widehat{{\textbf{g}}}\|_1^4\frac{3\log(pD)}{n}+\|\textbf{g}_1-\widehat{{\textbf{g}}}\|_1^2\sqrt{\frac{6\log(pD)}{n}}&\leq\|\textbf{g}_1-\widehat{{\textbf{g}}}\|_1^2\sqrt{\frac{3\log(pD)}{n}}[C_T+\sqrt{2}]\\&\leq C_T(C_T+\sqrt{2}).
\end{align*}
Thus, we have:
\begin{equation}\label{eq:t_n}
    \mathcal{Q}(\|\textbf{g}_1-\widehat{{\textbf{g}}}\|_1^2)\|\textbf{g}_1-\widehat{{\textbf{g}}}\|_\mathbb{H}^2\leq 108(\widetilde{\mu}_1+\frac{\sigma^2}{p})C_T(C_T+\sqrt{2})\|\textbf{g}_1-\widehat{{\textbf{g}}}\|_\mathbb{H}^2.
\end{equation}
Combining Equations \eqref{eq:t_n} and \eqref{eq:dot:R:upper:bound} in \eqref{ineq:2nd part:EG} gives:
\begin{align*}
  |E(\textbf{g}_1)|&\leq4\lambda_1\|\textbf{g}_1-\widehat{{\textbf{g}}}\|_1\\&+4\|\textbf{g}_1-\widehat{{\textbf{g}}}\|_\mathbb{H}^2 (\frac{8\sqrt{L\|K\|_\infty}D}{M^\alpha}+\frac{\sigma^2}{p}+108(\widetilde{\mu}_1+\frac{\sigma^2}{p})C_T(C_T+\sqrt{2}))\\&+4\|\textbf{g}_1-\widehat{{\textbf{g}}}\|_1\Big(\|\textbf{g}_1\|_\mathbb{H}\lambda_1+\|\textbf{g}_1\|_\mathbb{H}\frac{8\sqrt{L\|K\|_\infty s}}{M^\alpha}+\|\textbf{g}_1\|_\infty\frac{\sigma^2}{p}\Big),
\end{align*}
taking $\lambda\geq 4\Big(\|\textbf{g}_1\|_\mathbb{H}(\lambda_1+\frac{8\sqrt{L\|K\|_\infty s}}{M^\alpha})+\|\textbf{g}_1\|_\infty\frac{\sigma^2}{p}+\lambda_1\Big)$ gives that with probability at least $1-2\frac{\log(T)+1}{MD}$:

\begin{align*}
   |E(\textbf{g}_1)|&\leq \lambda\|\textbf{g}_1-\widehat{{\textbf{g}}}\|_1\\&+4\|\textbf{g}_1-\widehat{{\textbf{g}}}\|_\mathbb{H}^2 \Big(\frac{8\sqrt{L\|K\|_\infty}D}{M^\alpha}+\frac{\sigma^2}{p}+108(\widetilde{\mu}_1+\frac{\sigma^2}{p})C_T(C_T+\sqrt{2})\Big).
\end{align*}
To conclude we recall the oracle condition. We assume that $p,M$ and $T$ are such that:
\begin{equation*}
 4\Big( \frac{8\sqrt{L\|K\|_\infty}D}{M^\alpha}+\frac{\sigma^2}{p}+108\big(\widetilde{\mu}_1+\frac{\sigma^2}{p}\big)C_T(C_T+\sqrt{2})\Big)\leq \sqrt{\mu_1}(\rho-8\eta),
\end{equation*}
which gives:

\begin{align*}
   |E(\textbf{g}_1)|\leq \lambda\|\textbf{g}_1-\widehat{{\textbf{g}}}\|_1+\|\textbf{g}_1-\widehat{{\textbf{g}}}\|_\mathbb{H}^2\sqrt{\mu_1}(\rho-8\eta),
\end{align*}
which ends the proof of Lemma \ref{upper bound on E(g)}.

\subsection{Preliminary results for the deterministic upper bounds, proofs of Lemmas \ref{Lemma bias infty norm} and \ref{Lemma:multi:bias:oracle}}\label{Proof determinstic}
In this subsection, we will establish an upper bound over the terms $\|({\Gamma}_\phi-\Gamma)(\textbf{g}_1)\|_\infty$ and $|\langle({\Gamma}_\phi-\Gamma)(\textbf{g}),\textbf{g}\rangle_\mathbb{H}|$ for any $\textbf{g}\in\mathbb{H}$. For that purpose, we will introduce in the sequel the operators $\Pi_{S_M^D},T^{(K)}$ and $T^{(N)}$ that will be useful in our proof. We recall that $(\phi_\lambda)_{\lambda\in\Lambda_M}$ is an orthonormal system of $\L^2$ of size $M$.
\begin{itemize}
    \item Let $\textbf{g}\in\mathbb{H}$, we define the projection operator $\Pi_{S_M^D}$ such that for any $d\in\{1,\dots,D\}$:
    
    \[(\Pi_{S_M^D}(\textbf{g}))_d:=\sum_{\lambda\in\Lambda_M}\langle\phi_\lambda,g_d\rangle\phi_\lambda.\]
    \item Let $\textbf{g}\in\mathbb{H}$, we define $T^{(K)}$ the operator that encapsulates the error of discretization of the kernel $K$, such that for any $d\in\{1,\dots,D\}$:
    \[(T^{(K)}(\textbf{g}))_d:=\sum_{d'=1}^D\sum_{\lambda,\lambda'\in\Lambda_M}\Big(R^{(K)}_{d,d',\lambda,\lambda'}\Big)\langle\phi_\lambda,g_{d'}\rangle\phi_{\lambda'},\]
     such that for any $(\lambda,\lambda')\in\Lambda_M^2$ and $(d,d')\in\{1,\dots,D\}^2$ we have $$R^{(K)}_{d,d',\lambda,\lambda'}:=\frac1{p^2}\sum_{h,h'=0}^{p-1}K_{d,d'}(t_h,t_{h'})\phi_\lambda(t_h)\phi_{\lambda'}(t_{h'})-\int_0^1\int_0^1K_{d,d'}(s,t)\phi_{\lambda}(s)\phi_{\lambda'}(t)dsdt.$$
    \item Let $\textbf{g}\in\mathbb{H}$, we define $T^{(N)}$ the operator that encapsulate the error of discretization of the noise, such that for any $d\in\{1,\dots,D\}$:
    \[(T^{(N)}(\textbf{g}))_d:=\sum_{d'=1}^D\sum_{\lambda,\lambda'\in\Lambda_M}\Big(R^{(N)}_{d,d',\lambda,\lambda'}\Big)\langle\phi_\lambda,g_{d'}\rangle\phi_{\lambda'},\]
    such that for any $(\lambda,\lambda')\in\Lambda_M^2$ and $(d,d')\in\{1,\dots,D\}^2$ we have $$R^{(N)}_{d,d',\lambda,\lambda'}:=1_{d=d'}\frac{\sigma^2}{p^2}\Big(\sum_{h=0}^{p-1}\phi_\lambda(t_h)\phi_{\lambda'}(t_{h})-\langle\phi_\lambda,\phi_{\lambda'}\rangle\Big).$$
        
\end{itemize}
\begin{lemma}\label{technical lemma}Denoting by ${\Gamma}_\phi=\mathbb{E}[\widehat{\Gamma}_\phi]$, we have for all $\textbf{g}\in\mathbb{H}$
\begin{align*}
    \|({\Gamma}_\phi-\Gamma)(\textbf{g})\|_\infty&\leq \|\Pi_{S_M^D}(\Gamma(\Pi_{S_M^D}(\textbf{g})))-\Gamma(\textbf{g})\|_\infty+\frac{\sigma^2}{p}\|\Pi_{S_M^D}(\textbf{g})\|_\infty+\|T^{(K)}(\textbf{g})\|_\infty+\|T^{(N)}(\textbf{g})\|_\infty
\end{align*}

\end{lemma}

\begin{proof}[of Lemma ~\ref{technical lemma}]
Denote by $ K_\phi=\mathbb E[\widehat{K}_\phi]$ and $\Gamma_\phi=\mathbb E[\widehat{\Gamma}_\phi]$, we have for all $ (d,d')\in\{1,\dots,D\}$
\begin{eqnarray}
  K_{\phi,d,d'}(s,t)&=&\sum_{\lambda,\lambda'\in\Lambda_M}\frac1{p^2}\sum_{h,h'=0}^{p-1}K_{d,d'}(t_h,t_{h'})\phi_\lambda(t_h)\phi_{\lambda'}(t_{h'})\phi_\lambda(s)\phi_{\lambda'}(t)\nonumber\\
&&+\frac{1_{d=d'}\sigma^2}{p^2}\sum_{\lambda,\lambda'\in\Lambda_M}\sum_{h=0}^{p-1}\phi_{\lambda}(t_h)\phi_{\lambda'}(t_{h})\phi_{\lambda}(s)\phi_{\lambda'}(t)\nonumber\\
&=&\Pi_{(S_M^D)^2}K_{d,d'}(s,t)+\frac{1_{d=d'}\sigma^2}{p}\sum_{\lambda\in\Lambda_M}\phi_{\lambda}(s)\phi_{\lambda}(t)+R_{d,d'}^{(K)}(s,t)+R_{d,d'}^{(N)}(s,t),\nonumber\\\label{eq:decompKtilde:multi}
\end{eqnarray}
where 
\begin{eqnarray*}
R_{d,d'}^{(K)}(s,t)=\sum_{\lambda,\lambda'\in\Lambda_M}\left(R_{d,d',\lambda,\lambda'}^{(K)}\right)\phi_\lambda(s)\phi_{\lambda'}(t),
\end{eqnarray*}
and
\begin{eqnarray*}
R_{d,d'}^{(N)}(s,t)
=\sum_{\lambda,\lambda'\in\Lambda_M}\left(R_{d,d',\lambda,\lambda'}^{(N)}\right)\phi_\lambda(s)\phi_{\lambda'}(t),
\end{eqnarray*}
and $\Pi_{(S_M^D)^2}$ is the orthogonal projection onto $(S_M^D)^2=\prod_{d=1}^D\text{span}\{(s,t)\mapsto \phi_\lambda(s)\phi_{\lambda'}(t), \lambda,\lambda'\in\Lambda_M,\}$.
Then, by denoting $\Gamma_\phi$ the integral operator associated with $ K_\phi$, we have for any fonction $\textbf{f}$ and any $t\in[0,1]$, from the decomposition of the function $K_\phi$ given in Eq.~\eqref{eq:decompKtilde:multi}, let $d\in\{1,\dots,D\}$
\begin{align*}
(\Gamma_\phi(\textbf{f})(t))_d&=\sum_{d'=1}^D\int_0^1  K_{\phi,d,d'}(s,t)f_{d'}(s)ds\\
&=\sum_{d'=1}^D\int_0^1\Pi_{(S_M^D)^2}K_{d,d'}(s,t)f_{d'}(s)ds+\frac{\sigma^2}{p}\sum_{\lambda\in\Lambda_M}\int_0^1\phi_{\lambda}(s)f_d(s)ds\,\phi_{\lambda}(t)\\&+\sum_{d'=1}^D\sum_{\lambda,\lambda'\in\Lambda_M}\left(R_{d,d',\lambda,\lambda'}^{(K)}\right)\int_0^1\phi_{\lambda'}(s)f_{d'}(s)ds\phi_{\lambda}(t)\\&+\sum_{d'=1}^D\sum_{\lambda,\lambda'\in\Lambda_M}\left(R_{d,d',\lambda,\lambda'}^{(N)}\right)\int_0^1\phi_{\lambda'}(s)f_{d'}(s)ds\phi_{\lambda}(t)\\
&=\frac{\sigma^2}{p}(\Pi_{S_M^D}(\textbf{f})(t))_d+(T^{(N)}(\textbf{f})(t))_d+\sum_{d'=1}^D\int_0^1\Pi_{(S_M^D)^2}K_{d,d'}(s,t)f_{d'}(s)ds\\&+(T^{(K)}(\textbf{f})(t))_d,
\end{align*}
where $T^{(K)}$ (resp. $T^{(N)}$) is the integral operator associated to the kernel $R^{(K)}$ (resp. $R^{(N)}$) and $\Pi_{S_M^D}$ the orthogonal projection of each component onto $S_M$ onto $S_M={\rm span}\{\phi_\lambda,\lambda\in\Lambda_M\}$.
Now, 
\begin{align*}
&\int_0^1\Pi_{(S_M^D)^2}K_{d,d'}(s,t)f_{d'}(s)ds\\&=\sum_{\lambda,\lambda'\in\Lambda_M}\int_0^1\int_0^1\int_0^1K_{d,d'}(u,v)\phi_{\lambda}(u)\phi_{\lambda'}(v)dudv\,\phi_{\lambda}(s)\phi_{\lambda'}(t)f_{d'}(s)ds\\
&=\sum_{\lambda,\lambda'\in\Lambda_M}\langle \phi_{\lambda},f_{d'}\rangle\int_0^1\int_0^1 K_{d,d'}(u,v)\phi_{\lambda}(u)\phi_{\lambda'}(v)dudv\, \phi_{\lambda'}(t)\\
&=\sum_{\lambda,\lambda'\in\Lambda_M}\langle \phi_{\lambda},f_{d'}\rangle \langle \Gamma_{d,d'}(\phi_\lambda),\phi_{\lambda'}\rangle \, \phi_{\lambda'}(t)\\
&=\sum_{\lambda'\in\Lambda_M}\langle\Gamma_{d,d'}(\sum_{\lambda\in\Lambda_M}\langle \phi_{\lambda},f_{d'}\rangle\phi_{\lambda}),\phi_{\lambda'}\rangle\, \phi_{\lambda'}(t)\\
&=\Pi_{S_M^D}(\Gamma_{d,d'}(\Pi_{S_M^D}(f_{d'})))(t).
\end{align*}

Thus we have 

\[\sum_{d'=1}^D\int_0^1\Pi_{(S_M^D)^2}K_{d,d'}(s,t)f_{d'}(s)ds=\sum_{d'=1}^D\Pi_{S_M^D}(\Gamma_{d,d'}(\Pi_{S_M^D}(f_{d'})))(t)\]

Hence, with a slight abuse of notation, we obtain:
\begin{equation}\label{eq:dec:norme:op}
    \Gamma_\phi=\Pi_{S_M^D}\Gamma\Pi_{S_M^D}+\frac{\sigma^2}p\Pi_{S_M^D}+T^{(K)}+T^{(N)}.
\end{equation}

where $\Pi_{S_M^D}$ is the orthogonal projection of each component onto $S_M$, we now have the following control on $\|({\Gamma}_\phi-\Gamma)(\textbf{g}_1)\|_\infty$
\begin{align*}
    \|({\Gamma}_\phi-\Gamma)(\textbf{g}_1)\|_\infty&\leq \|(\Pi_{S_M^D}\Gamma\Pi_{S_M^D}-\Gamma)(\textbf{g}_1)\|_\infty+\frac{\sigma^2}{p}\|\Pi_{S_M^D}(\textbf{g}_1)\|_\infty+\|T^{(K)}(\textbf{g}_1)\|_\infty+\|T^{(N)}(\textbf{g}_1)\|_\infty,
\end{align*}

 which ends the proof of Lemma \ref{technical lemma}.
\end{proof}

\begin{lemma}
For all $ (s,s',t,t')\in[0,1]^4$ and $(d,d')\in\{1,\dots,D\}$ we have :
\begin{align}
    |{K}_{d',d}(t',s')-K_{d',d}(t,s)|\leq\sqrt{L\|K\|_\infty}\big[|s-s'|^\alpha+|t-t'|^\alpha\big]\label{technical lemmas}
\end{align}

where $\|K_{d,d}\|_\infty=\sup_{(t,s)\in[0,1]^2}\mathbb{E}[Z_d(t)Z_d(s)]$ and $\|K\|_\infty=\max_{d\in\{1,\dots,D\}}\|K_{d,d}\|_\infty$

\end{lemma}

\begin{proof}
Let $(s,s',t,t')\in[0,1]^4$ and $(d,d')\in\{1,\dots,D\}$, we have:
\begin{align*}
    |{K}_{d',d}(t',s')-K_{d',d}(t,s)|&\leq|\mathbf{E}[Z_{d'}(t')Z_{d}(s')-Z_{d'}(t)Z_{d}(s)]|\\&\leq|\mathbf{E}[Z_{d'}(t')Z_{d}(s')-Z_{d'}(t')Z_{d}(s)+Z_{d'}(t')Z_{d}(s)-Z_{d'}(t)Z_{d}(s)]|\\&\leq|\mathbf{E}[Z_{d'}(t')(Z_{d}(s')-Z_{d}(s))+Z_{d}(s)(Z_{d'}(t')-Z_{d'}(t))]|\\&\leq \sqrt{L}\sqrt{\mathbf{E}[Z_{d'}(t_h)^2]}|s'-s|^\alpha+\sqrt{L}\sqrt{\mathbf{E}[Z_{d}(s)^2]}|t-t'|^\alpha.
\end{align*}
\end{proof}

\begin{lemma}
Let $\mathbb{P}_\textbf{Z}\in R_\alpha^{(D)}(L)$ we have :
\begin{align}
    \max_{(\lambda,\lambda')\in\Lambda_M^2}\max_{(d,d')\in\{1,\dots,D\}^2}|R_{d,d',\lambda,\lambda'}^{(K)}|&\leq \frac{4\sqrt{\|K\|_\infty L}}{\alpha+1}M^{-1}p^{-\alpha}\label{Result on R^K}\\R^{(N)}=0\label{Result on R^N} 
\end{align}

\end{lemma}
\begin{proof}
Let $b_h=h/p,h=0,\dots,p$, observe that $t_h=\frac{h}{p-1}\in[b_h,b_{h+1}]$. We define for any $\lambda:=0,\dots,M-1$ 
$$J_\lambda:=\{h=0,\dots,p-1:Leb([b_h,b_{h+1}]\cap I_\lambda)\neq 0\}.$$ 
Note that for any $(\lambda,\lambda')\in\{0,\hdots,M-1\}$ such that $\lambda\neq \lambda'$, we have $J_{\lambda}\cap J_{\lambda'}=\emptyset$ thus:
\begin{eqnarray*}
R_{d,d',\lambda,\lambda'}^{(N)}&=&1_{d=d'}\frac{\sigma^2}{p}\left(\frac1p\sum_{h=0}^{p-1}\phi_\lambda(t_h)\phi_{\lambda'}(t_{h})-\langle\phi_\lambda,\phi_{\lambda'}\rangle\right)\\
&=&1_{d=d'}\frac{\sigma^2}{p}\left(\frac1p\sum_{h\in J_{\lambda}\cap J_{\lambda'}}M-1_{\{\lambda=\lambda'\}}\right)\\
&=&1_{d=d'}\frac{\sigma^2}{p}\left(\frac{M}{p}card(J_\lambda\cap J_{\lambda'})-1_{\{\lambda=\lambda'\}}\right)=0,
\end{eqnarray*}
and

\begin{eqnarray*}
R_{d,d',\lambda,\lambda'}^{(K)}&=&\frac1{p^2}\sum_{h,h'=0}^{p-1}K_{d,d'}(t_h,t_{h'})\phi_\lambda(t_h)\phi_{\lambda'}(t_{h'})-\int_0^1\int_0^1K_{d,d'}(s,t)\phi_{\lambda}(s)\phi_{\lambda'}(t)dsdt\\
&=&\sum_{h,h'=0}^{p-1}\int_{b_h}^{b_{h+1}}\int_{b_{h'}}^{b_{h'+1}}\Big[K_{d,d'}(t_h,t_{h'})-K_{d,d'}(s,t)\Big]\phi_{\lambda}(s)\phi_{\lambda'}(t)dsdt\\
&&+\sum_{h,h'=0}^{p-1}\int_{b_h}^{b_{h+1}}\int_{b_{h'}}^{b_{h'+1}}K_{d,d'}(t_h,t_{h'})\Big[\phi_{\lambda}(t_h)\phi_{\lambda'}(t_{h'})-\phi_{\lambda}(s)\phi_{\lambda'}(t)\Big]dsdt\\
&=&M\sum_{h\in J_\lambda}\sum_{h'\in J_{\lambda'}}\int_{b_h}^{b_{h+1}}\int_{b_{h'}}^{b_{h'+1}}\Big[K_{d,d'}(t_h,t_{h'})-K_{d,d'}(s,t)\Big]dsdt.
\end{eqnarray*}

Therefore,
\begin{eqnarray*}
|R_{d,d',\lambda,\lambda'}^{(K)}|&\leq&M\sum_{h\in J_\lambda}\sum_{h'\in J_{\lambda'}}\int_{b_h}^{b_{h+1}}\int_{b_{h'}}^{b_{h'+1}}\sqrt{\|K\|_\infty L}\Big(|s-t_h|^\alpha+|t-t_{h'}|^\alpha\Big)dsdt\\
&\leq&2\sqrt{\|K_{d,d'}\|_\infty L}\times Mp^{-1}\text{card}(J_{\lambda'})\sum_{h\in J_\lambda}\int_{b_h}^{b_{h+1}}|s-t_h|^\alpha ds\\
&\leq&2\sqrt{\|K_{d,d'}\|_\infty L}\times Mp^{-1}\text{card}(J_{\lambda'})\text{card}(J_{\lambda})\times \frac{2}{\alpha+1}p^{-\alpha-1}\\
&\leq&\frac{4\sqrt{\|K_{d,d'}\|_\infty L}}{\alpha+1}M^{-1}p^{-\alpha}.
\end{eqnarray*}
Finally,
$$\max_{(\lambda,\lambda')\in\Lambda_M^2}\max_{(d,d')\in\{1,\dots,D\}^2}|R_{d,d',\lambda,\lambda'}^{(K)}|\leq \frac{4\sqrt{\|K\|_\infty L}}{\alpha+1}M^{-1}p^{-\alpha}.$$
\end{proof}

\begin{lemma}\label{technical lemma norm sup K}
For all $ (s,t)\in[0,1]^2$ and for all $ (d,d')\in\{1,\dots,D\}^2$ we have :

\begin{equation}
    \|\Pi_{(S_M^D)^2}K_{d,d'}-K_{d,d'}\|_\infty\leq  \frac{4\sqrt{L\|K_{d,d'}\|_\infty}}{\alpha+1}M^{-\alpha} \label{Result on Pi}
\end{equation}
\end{lemma}

\begin{proof}
Let $(s,t)\in [0,1]^2$, then there exists a unique couple $(\lambda,\lambda')\in\Lambda_D^2$ such that $s\in I_\lambda$ and $t\in I_{\lambda'}$. Therefore, $\phi_{\lambda''}(s)=0$ for $\lambda''\not=\lambda$ and $\phi_{\lambda'''}(t)=0$ for $\lambda'''\not=\lambda'$ and then,
\begin{eqnarray*}
&&\Pi_{(S_M^D)^2}K_{d,d'}(s,t)-K_{d,d'}(s,t)\\&=&\sum_{\lambda'',\lambda'''\in\Lambda_M}\int_0^1\int_0^1 K_{d,d'}(s',t')\phi_{\lambda''}(s')\phi_{\lambda'''}(t')ds'dt' \phi_{\lambda''}(s)\phi_{\lambda'''}(t)-K_{d,d'}(s,t)\\
&=&\int_0^1\int_0^1 K_{d,d'}(s',t')\phi_{\lambda}(s')\phi_{\lambda'}(t')ds'dt' \phi_{\lambda}(s)\phi_{\lambda'}(t)-K_{d,d'}(s,t)\\
&=&M^2\int_{I_\lambda}\int_{I_{\lambda'}} (K_{d,d'}(s',t')-K_{d,d'}(s,t))ds'dt'
\end{eqnarray*}
Then, \eqref{technical lemmas} gives
\begin{align*}
\left|\Pi_{(S_M^D)^2}K_{d,d'}(s,t)-K_{d,d'}(s,t)\right|&\leq M^2\sqrt{L\|K_{d,d'}\|_\infty}\int_{I_\lambda}\int_{I_{\lambda'}}\Big[ |s'-s|^\alpha + |t-t'|^{\alpha}\Big]ds'dt'\\
&\leq \frac{4\sqrt{L\|K_{d,d'}\|_\infty}}{\alpha+1}M^{-\alpha}.
\end{align*}
\end{proof}

\begin{lemma}\label{lemma:control:sup:norme:1}
Let $(a_\lambda)_{\lambda\in\Lambda_M}$ be a sequence of reel numbers, we have:

\begin{equation}
    \sup_{t\in[0,1]}|\sum_{\lambda\in\Lambda_M}a_\lambda\phi_\lambda(t)|\leq\max_{\lambda\in\Lambda_M}|\sqrt{M}a_\lambda|.\label{sup hist = max for each}
\end{equation}
\end{lemma}
\begin{proof}
We have the following:

\begin{align*}
      \sup_{t\in[0,1]}|\sum_{\lambda\in\Lambda_M}a_\lambda\phi_\lambda(t)|&\leq \max_{\lambda\in\Lambda_M}|a_\lambda| \sup_{t\in[0,1]}|\sum_{\lambda\in\Lambda_M}\phi_\lambda(t)|.
\end{align*}
Note that since the $\phi_\lambda$'s have disjoint support, for any $t\in[0,1]$ we have

$$|\sum_{\lambda\in\Lambda_M}\phi_\lambda(t)|=\sqrt{M}|\sum_{\lambda\in\Lambda_M}1_\lambda(t)|\leq\sqrt{M}.$$
Finally we have 
\begin{align*}
      \sup_{t\in[0,1]}|\sum_{\lambda\in\Lambda_M}a_\lambda\phi_\lambda(t)|&\leq \max_{\lambda\in\Lambda_M}|a_\lambda\sqrt{M}|.
\end{align*}
\end{proof}
{\bf Lemma~\ref{Lemma bias infty norm} }{ \it Let $\mathbb{P}_\textbf{Z}\in R_\alpha^{(D)}(L)$, $s\in\{1,\dots,D\}$ and $\| \textbf{g}_1\|_0=s$, we have 

$$\|({{\Gamma}}_\phi-{\Gamma})(\textbf{g}_1)\|_\infty\leq\frac{8\sqrt{sL\|K\|_\infty}\|\textbf{g}_1\|_\mathbb{H}}{M^\alpha}+\frac{\sigma^2\|\textbf{g}_1\|_\infty}{p}.$$
}

\begin{proof}
We need to establish an upper bound for $\|({{\Gamma}}_\phi-{\Gamma})(\textbf{g}_1)\|_\infty$. Using lemma \ref{technical lemma} we have :
$$ \|({\Gamma}_\phi-\Gamma)(\textbf{g}_1)\|_\infty\leq \|(\Pi_{S_M^D}\Gamma\Pi_{S_M^D}-\Gamma)(\textbf{g}_1)\|_\infty+\frac{\sigma^2}{p}\|\Pi_{S_M^D}(\textbf{g}_1)\|_\infty+\|T^{(K)}(\textbf{g}_1)\|_\infty+\|T^{(N)}(\textbf{g}_1)\|_\infty.$$
Thus in the sequel, we establish an upper bound for each term.
\begin{itemize}
    \item We have:
    \[\|(\Pi_{S_M^D}\Gamma\Pi_{S_M^D}-\Gamma)(\textbf{g}_1)\|_\infty=\max_{d\in{1,\dots,D}}\sup_{x\in[0,1]}|((\Pi_{S_M^D}\Gamma\Pi_{S_M^D}-\Gamma)(\textbf{g}_1))_d(x)|\]
    Thus we need a control on $\max_{d\in{1,\dots,D}}\sup_{x\in[0,1]}|((\Pi_{S_M^D}\Gamma\Pi_{S_M^D}-\Gamma)(\textbf{g}_1))_d(x)|$. Since $\textbf{g}_1$ is $s-$sparse using Lemma , for any $x\in[0,1]$:
 
    \begin{align*}
    |((\Pi_{S_M^D}\Gamma\Pi_{S_M^D}-\Gamma)(\textbf{g}_1)(x))_d|&=|\sum_{d'=1}^D\int_0^1 (\Pi_{(S_M^D)^2}K_{d,d'}(x,t)-K_{d,d'}(x,t))g_{1,d'}(t)dt|\\&=|\sum_{d'=1}^D\int_0^1 (\Pi_{(S_M^D)^2}K_{d,d'}(x,t)-K_{d,d'}(x,t))g_{1,d'}(t)dt|\\&\leq \sum_{d'=1}^D\int_0^1 |\Pi_{(S_M^D)^2}K_{d,d'}(x,t)-K_{d,d'}(x,t)||g_{1,d'}(t)|dt\\&\leq\frac{4\sqrt{L\|K\|_\infty}}{M^\alpha}\sum_{d'=1}^D\int_0^1 |g_{1,d'}(t)|dt=\frac{4\sqrt{L\|K\|_\infty}}{M^\alpha}\|\textbf{g}_1\|_1\\&\leq \frac{4\sqrt{sL\|K\|_\infty}\|\textbf{g}_1\|_\mathbb{H}}{M^\alpha},
\end{align*}
since $\|\textbf{g}_1\|_1\leq \sqrt{s}\|\textbf{g}_1\|_\mathbb{H}$.
    \item Using Lemma \ref{lemma:control:sup:norme:1} we have $\|\Pi_{S_M^D}(\textbf{g}_1)\|_\infty$:

\begin{align*}
    \|\Pi_{S_M^D}(\textbf{g}_1)\|_\infty &=  \max_{d\in{1,\dots,D}}\sup_{t\in[0,1]} |\sum_{\lambda\in\Lambda_M}\langle\phi_\lambda,g_{1,d}\rangle\phi_{\lambda}(t)|\\&=\max_{d\in{1,\dots,D}}\max_{\lambda\in\Lambda_M} \sqrt{M}|\langle\phi_\lambda,g_{1,d}\rangle|\\&\leq \|\textbf{g}_1\|_\infty\sqrt{M}\|\phi_\lambda\|_1\\&=\|\textbf{g}_1\|_\infty.
\end{align*}
    \item For the third term of the inequality $\|T^{(K)}(\textbf{g}_1)\|_\infty$, we have:

\begin{align*}
    \|T^{(K)}(\textbf{g}_1)\|_\infty&=\max_{d\in{1,\dots,D}}\sup_{s\in[0,1]}|\sum_{d'=1}^D\int_{0}^1R_{d,d'}^{(K)}(s,t)g_{1,d'}(t)dt|\\&=\max_{d\in{1,\dots,D}}\sup_{s\in[0,1]}|\sum_{d'=1}^D\sum_{\lambda,\lambda'\in\Lambda_M}\left(R_{d,d',\lambda,\lambda'}^{(K)}\right)\phi_\lambda(s)\phi_{\lambda'}(t)g_{1,d'}(t)dt|\\ &=\max_{(d,d')\in\{1,\dots,D\}^2,(\lambda,\lambda')\in\Lambda_M^2}|R^{(K)}_{d,d',\lambda,\lambda'}|\sup_{s\in[0,1]}|\sum_{d'=1}^D\sum_{\lambda,\lambda'\in\Lambda_M}\langle\phi_{\lambda'},g_{1,d'}\rangle\phi_{\lambda}(s)|\end{align*}
    Using Lemmas \ref{technical lemma norm sup K} and \ref{lemma:control:sup:norme:1}, we deduce the following:
    
    \begin{align*}\|T^{(K)}(\textbf{g}_1)\|_\infty&\leq\frac{4\sqrt{\|K\|_\infty L}}{\alpha+1}M^{-1}p^{-\alpha}\max_{\lambda'\in\Lambda_M}|\sum_{d'=1}^D\sqrt{M}\sum_{\lambda\in\Lambda_M}\langle\phi_\lambda,g_{1,d'}\rangle|\\&\leq\frac{4\sqrt{\|K\|_\infty L}}{\alpha+1}M^{-1}p^{-\alpha}\sum_{d'=1}^DM\|g_{1,d'}\|_1\\&\leq\frac{4\sqrt{\|K\|_\infty L}}{\alpha+1}p^{-\alpha}\|\textbf{g}_{1}\|_1\\&\leq \frac{4\sqrt{s\|K\|_\infty L}}{\alpha+1}p^{-\alpha}\|\textbf{g}_{1}\|_\mathbb{H}.
\end{align*}

\item it comes from equation \eqref{Result on R^N} that $\|T^{(N)}(\textbf{g}_1)\|_\infty=0$.
\end{itemize}

This ends the proof of Lemma.
\end{proof}
{\bf Lemma~\ref{Lemma:multi:bias:oracle} }{ \it Let $\mathbb{P}_\textbf{Z}\in R_\alpha^{(D)}(L)$ we have for all $ \textbf{g}\in\mathbb{H}$

$$|\langle({\Gamma}_\phi-\Gamma)\textbf{g},\textbf{g}\rangle_\mathbb{H}|\leq \|\textbf{g}\|_\mathbb{H}^2 \Big(\frac{8D\sqrt{\|K\|_\infty L}}{(\alpha+1)M^{\alpha}}+\frac{\sigma^2}{p}\Big).$$}

\begin{proof}
Let $\textbf{g}\in\mathbb{H}$, we need to control $|\langle({\Gamma}_\phi-\Gamma)\textbf{g},\textbf{g}\rangle_\mathbb{H}|$, using equation \eqref{eq:dec:norme:op} we have 
\begin{align*}
    |\langle({\Gamma}_\phi-\Gamma)(\textbf{g}),\textbf{g}\rangle_\mathbb{H}|\leq |\langle(\Pi_{S_M^D}\Gamma\Pi_{S_M^D}-\Gamma)(\textbf{g}),\textbf{g}\rangle|+\frac{\sigma^2}{p}\|\textbf{g}\|_\mathbb{H}^2+|\langle T^{(K)}(\textbf{g}),\textbf{g}\rangle|+|\langle T^{(N)}(\textbf{g}),\textbf{g}\rangle|.
\end{align*}
We showed earlier that $T^{(N)}=0$ in the case of histogram. Thus we have $|\langle({\Gamma}_\phi-\Gamma)\textbf{g},\textbf{g}\rangle_\mathbb{H}|\leq |\langle\Pi_{S_M^D}\Gamma\Pi_{S_M^D}-\Gamma)(\textbf{g}),\textbf{g}\rangle|+\frac{\sigma^2}{p}\|\textbf{g}\|_\mathbb{H}+|\langle T^{(K)}\textbf{g},\textbf{g}\rangle|,$
we will upper bound the rest of the terms in what follows:

\begin{itemize}
    \item For the first term of the inequality, using equation \eqref{Result on Pi} we have 
    \begin{align*}
       | \langle(\Pi_{S_M^D}\Gamma\Pi_{S_M^D}-\Gamma)\textbf{g},\textbf{g}\rangle_\mathbb{H}|&=|\sum_{d,d'=1}^D\langle(\Pi_{S_M^D}\Gamma\Pi_{S_M^D}-\Gamma)_{d,d'}g_d,g_{d'}\rangle|\\&\leq\sum_{d,d'=1}^D\int_0^1\int_0^1|\Pi_{(S_M^D)^2}K_{d,d'}(x,t)-K_{d,d'}(x,t)|\\&\times|g_d(t)||g_{d'}(s)|dtds\\&\leq \|\Pi_{(S_M^D)^2}K_{d,d'}-K_{d,d'}\|_\infty\|\textbf{g}\|_1^2\\&\leq \|\Pi_{(S_M^D)^2}K_{d,d'}-K_{d,d'}\|_\infty D\|\textbf{g}\|_H^2
    \end{align*}
However, since we showed $\|\Pi_{(S_M^D)^2}K_{d,d'}-K_{d,d'}\|_\infty\leq\frac{4\sqrt{L\|K\|_\infty}}{\alpha+1}M^{-\alpha}$ in Lemma \ref{technical lemma norm sup K}, we have:
   \begin{align*}
       | \langle(\Pi_{S_M^D}\Gamma\Pi_{S_M^D}-\Gamma)\textbf{g},\textbf{g}\rangle_\mathbb{H}|&\leq\frac{4\sqrt{L\|K\|_\infty}}{\alpha+1}M^{-\alpha} D\|\textbf{g}\|_\mathbb{H}^2.
    \end{align*}
    \item For the second term of the inequality we have 
    \begin{align*}
         |\langle(T^{(K)})\textbf{g},\textbf{g}\rangle_\mathbb{H}|&\leq\|T^{(K)}(\textbf{g})\|_\infty\|\textbf{g}\|_1,\end{align*}
         using the same logic as Lemma \ref{Lemma bias infty norm}, we know that $\|T^{(K)}(\textbf{g})\|_\infty\leq\frac{4\sqrt{\|K\|_\infty L}}{\alpha+1}p^{-\alpha}$ we have:
         \begin{align*}|\langle(T^{(K)})\textbf{g},\textbf{g}\rangle_\mathbb{H}| &\leq\frac{4\sqrt{\|K\|_\infty L}}{\alpha+1}p^{-\alpha}\|\textbf{g}\|_1^2\\&\leq \frac{4D\sqrt{\|K\|_\infty L}}{\alpha+1}p^{-\alpha}\|\textbf{g}\|_\mathbb{H}^2\\&\leq \frac{4D\sqrt{\|K\|_\infty L}}{\alpha+1}M^{-\alpha}\|\textbf{g}\|_\mathbb{H}^2.
         \end{align*}
\end{itemize}
\end{proof}

\subsection{Preliminary results for the random upper bounds: Proof of     Lemmas \ref{Lemma concentration} and \ref{lemma:most:important} }\label{Proof random}
In this subsection, we will establish now upper bounds on random terms, namely $\|(\widehat{\Gamma}_\phi-\Gamma_\phi)(\textbf{g}_1)\|_\infty$ and $|\langle(\widehat{\Gamma}_\phi-\Gamma_\phi)(\textbf{g}),\textbf{g}\rangle_\mathbb{H}|$. We adapt to the functional multivariate setting the results of \cite{VAN}. We recall that  $\{\textbf{Y}_i\}_{i=1,\dots,n}$ are i.i.d vectors of size $pD$ containing all observations, i.e.for any $j\in\{1,\dots,pD\}$ denoting by $q$ and $r$ the quotient and the rest of the Euclidean division of $j$ by $D$ we have 

$$(\textbf{Y}_i)_j=Y_{i,r}(t_q).$$
\begin{itemize}
    \item We denote by $\Sigma$ the matrix of covariance of the random vector $\textbf{Y}_1$, i.e.\begin{equation}\label{eq:def:Sigma:multi}
        \Sigma:=\mathbb{E}[\textbf{Y}_1\textbf{Y}_1^T].
    \end{equation}
    Note that since $\textbf{Y}_1$ is assumed to be Gaussian it means that $\textbf{Y}_1\sim\mathcal{N}(0,\Sigma)$.
    \item We denote by $\phi$ the matrix of size $p\times M$ defined such that, for any $\lambda\in\Lambda_M$ and for any $h\in\{0,\dots,p-1\}$ we have 
    
    $$(\phi)_{h,\lambda}=\frac{\phi_{\lambda}(t_h)}{\sqrt{p}},$$
    and $\psi$ the block matrix such that:
    $$
\psi=\begin{pmatrix}
\phi & 0 & \dots & 0\\
0 & \phi & \ddots & \vdots\\ 
\vdots &\ddots&\ddots&  0\\
0 & \dots & 0 & \phi
\end{pmatrix}.
$$
Note that as defined we have $(\phi^T\phi)=\frac{\sum_{h=1}^p\phi_\lambda(t_h)\phi_{\lambda'}(t_h)}{p}=1_{\lambda=\lambda'}$ and $\psi^T\psi=I_{MD}$. Next, since for any $i\in\{1,\dots,n\},d\in\{1,\dots,D\}$ and $\lambda\in\Lambda_M$ we have $\widetilde{y}_{i,d,\lambda}=\frac{1}{p}\sum_{h=0}^{p-1}\phi_\lambda(t_h)Y_{i,d}(t_h)$, it implies that for all $i\in\{1,\dots,n\}$ we have $\widetilde{\textbf{y}}_i=\frac{1}{\sqrt{p}}\psi ^T\textbf{Y}_i$, and since $\textbf{Y}_i$'s are i.i.d Gaussian variables with covairance $\Sigma$, we have:
\begin{align}
    \forall i\in\{1,\dots,n\}\quad\widetilde{\textbf{y}}_i&\sim\mathcal{N}\Big(0,\frac{1}{p}\psi^T\Sigma\psi\Big).
\end{align}

\item Lastly, we introduce the matrix norm denoted $\|\bullet\|_2$, the trace $Tr(\bullet)$ and we compute an upper bound on $\|\Sigma\|_2$, let $M$ be a symmetric matrix of size $pD\times pD$, we assume it to be positive semi-definite , i.e., let $(\mu_j)_{j=1}^{pD}$ be its eigenvalues ordered such that $\mu_1>\mu_2>\dots>\mu_{pD}$, we know that $\mu_{pD}\geq 0$ :
\begin{align*}
    \|M\|_2&:=\sup_{\|u\|_{\ell_2}=1}\|M u\|_{\ell_2}=\mu_1,\\Tr(M)&:=\sum_{j=1}^{pD}\mu_j=\sum_{j=1}^{pD}M_{j,j}.
\end{align*}

First as defined it is straightforward to see that $$\Sigma=\mathbb{E}[\textbf{Y}_1\textbf{Y}_1^T]=\Sigma_Z+\sigma^2I_{pD},$$
where $\Sigma_Z$ is the covariance matrix of the vector containing all the $Z_d(t_h)$, i.e., for any $j,j'\in\{1,\dots,pD\}$ denoting by $q,q'$ and $r,r'$ the quotients and the rests of the Euclidean division of $j,j'$ by $D$ we have 

\[(\Sigma_Z)_{j,j'}:=\mathbb{E}[Z_{r}(t_q)Z_{r'}(t_{q'})].\]

Thus we have $\|\Sigma\|_2= \|\Sigma_Z\|_2+\sigma^2$. Finally note that :
\begin{align*}
    \frac{1}{p}\|\Sigma_Z\|_2&\leq\frac{1}{p}Tr(\Sigma_Z)\\&=\frac{1}{p}\sum_{d=1}^D\sum_{h=0}^{p-1}\mathbb{E}[Z_d(t_h)^2]\\&=\frac{1}{p}\sum_{d=1}^D\sum_{h=0}^{p-1} K_{d,d}(t_h,t_h).
\end{align*}

Note that $\frac{1}{p}\sum_{d=1}^D\sum_{h=1}^p K_{d,d}(t_h,t_h)$ is the Riemann sum associated with $\sum_{d=1}^D \int_0^1 K_{d,d}(t,t)dt$, it converges when $ p\rightarrow\infty$ to $\sum_{d=1}^D \int_0^1 K_{d,d}(t,t)dt$, which means that $\frac{1}{p}\|\Sigma_Z\|_2$ is of the order of a constant. In the rest of the section we denote by $\widetilde{\mu}_1$ the largest eigenvalue of $\frac{1}{p}\|\Sigma_Z\|_2$. Thus we have:
    \begin{equation}
        \|\Sigma\|_2=\widetilde{\mu}_1p+\sigma^2\label{norm of Sigma}.
    \end{equation}

\end{itemize}

\begin{remark}
Let $\textbf{f}\in\mathbb{H}$, as $\Gamma_\phi$ is defined $\langle{\Gamma}_\phi({\textbf{f}}),{\textbf{f}}\rangle_\mathbb{H}$ can be written as:
 \begin{align*}
    &\langle{\Gamma}_\phi({\textbf{f}}),{\textbf{f}}\rangle_\mathbb{H}\\&=\sum_{d,d'=1}^D\sum_{\lambda,\lambda'\in\lambda_M}\Big(\frac{1}{p^2}\sum_{h,h'=1}^p\phi_\lambda(t_h)\phi_{\lambda'}(t_{h'})K(t_h,t_{h'})\Big)\\&\times\langle\phi_{\lambda},f_{d}\rangle\langle\phi_{\lambda'},f_{d'}\rangle\\&=\frac{1}{p}\langle\psi(\Sigma)\psi^T a_{\textbf{f}},a_{\textbf{f}})\rangle_{\ell_2},
\end{align*}
where $(a_{\textbf{f}})_{\lambda,d}:=\langle\phi_\lambda,f_{d}\rangle$, note that $\|a_\textbf{f}\|_{\ell_2}^2=\sum_{\lambda\in\Lambda_M}\sum_{d=1}^D\langle\phi_\lambda,f_{d}\rangle^2=\|\Pi_{S_M^D}(\textbf{f})\|_\mathbb{H}^2= 1$, thus :
\begin{align*}
    \sup_{\|\textbf{f}\|_\mathbb{H}=1}\langle{\Gamma}_\phi({\textbf{f}}),{\textbf{f}}\rangle_\mathbb{H}=\sup_{\|\textbf{f}\|_\mathbb{H}=1}\frac{1}{p}\langle\psi(\Sigma)\psi^T a_{\textbf{f}},a_{\textbf{f}})\rangle_{\ell_2}= \frac{\|\psi^T(\Sigma)\psi\|_2}{p}.
\end{align*}    
However, since $\psi^T\psi=I_{MD}$, it implies that  $\|\psi^T(\Sigma)\psi\|_2=\|\Sigma\|_2$ which is equal to $\widetilde{\mu}_1p+\sigma^2$ thus
  \begin{align*}  
     \sup_{\|\textbf{f}\|_\mathbb{H}=1}\langle{\Gamma}_\phi({\textbf{f}}),{\textbf{f}}\rangle_\mathbb{H} =\widetilde{\mu}_1+\frac{\sigma^2}{p}
\end{align*}
Note that Lemma \ref{Lemma:multi:bias:oracle} implies the following:

\begin{equation}\label{eq:se:debarasser:de mu_1}
    |\widetilde{\mu}_1+\frac{\sigma^2}{p}-{\mu}_1|\leq \sup_{\|\textbf{f}\|_\mathbb{H}=1}|\langle({\Gamma}_\phi-\Gamma)\textbf{f},\textbf{f}\rangle_\mathbb{H}|\leq  \Big(\frac{8D\sqrt{\|K\|_\infty L}}{(\alpha+1)M^{\alpha}}+\frac{\sigma^2}{p}\Big).
    \end{equation}
\end{remark}
We use the following Bernstein concentration inequality.
\begin{theorem}\label{thm:bernstein}
Let $X_1,\dots,X_n$ be centered independent real-valued
random variables. Assume that there exist positive numbers $v$ and $c$ such that $\mathbb{E}[X_1^2]\leq v$ and 

\begin{equation*}
    \mathbb{E}[|X_1|^k]\leq \frac{k!}{2}vc^{k-2}\quad\text{for all integers $k\geq 3$,}
\end{equation*}
then for all $t>0$,

\begin{equation}
    \mathbb{P}(\frac{1}{n}|\sum_{i=1}^nX_i|\geq \sqrt{2vt}+ct)\leq 2e^{-nt}.
\end{equation}

\end{theorem}
The proof is provided in \cite{massart} Section 2.8. We recall\\
{\bf Lemma~\ref{Lemma concentration} }{ \it Assuming the observations are Gaussian, we have with probability at least $1-\frac{2}{MD}$

\begin{align*}
    \|(\widehat{{\Gamma}}_\phi-{\Gamma}_\phi)(\textbf{g}_1)\|_\infty&\leq 4\|\textbf{g}_1\|_\mathbb{H}\sqrt{(\widetilde{\mu}_1+\frac{\sigma^2}{p})(\|K\|_\infty+\frac{\sigma^2}{p})}[4\sqrt{\frac{\log(DM)}{n}}+\frac{\log(DM)}{n}]\\&= \|\textbf{g}_1\|_\mathbb{H}\lambda_1,
\end{align*}
where $\widetilde{\mu}_1+\frac{\sigma^2}{p}$ is the largest eigenvalue of $\frac{\mathbb{E}[\textbf{Y}_1^T\textbf{Y}_1]}{p}$.
}

\subsubsection{Proof of Lemma \ref{Lemma concentration}}
The proof is based the following claims:
\begin{itemize}
    \item We first show the following equality:
    $$\|(\widehat{{\Gamma}}_\phi-{\Gamma}_\phi)(\textbf{g}_1)\|_\infty=\sup_{\lambda\in\Lambda_M,d'\in\{1,\dots,D\}}|\sum_{d=1}^D\int_0^1(\widehat{K}_\phi(s,\frac{\lambda}{M})-K_\phi(s,\frac{\lambda}{M}))_{d',d}g_{1,d}(s)ds|.$$
    \item We show that $\sum_{d=1}^D\int_0^1(\widehat{K}_\phi(s,\frac{\lambda}{M}))_{d',d}g_{1,d}(s)ds$ can be written as a function of the $\widetilde{\textbf{y}}_i$'s namely $\frac{1}{n}\sum_{i=1}^nF(\widetilde{y}_{i,d'},\widetilde{\textbf{y}}_i)$ (where $F$ is specified in the proof) and compute its moments.
    \item We compute an upper bound on $|\sum_{d=1}^D\int_0^1(\widehat{K}_\phi(s,\frac{\lambda}{M})-K_\phi(s,\frac{\lambda}{M}))_{d',d}g_{1,d}(s)ds|$ using Theorem \ref{thm:bernstein} and we conclude by using a union bound to compute an upper bound on $\|(\widehat{{\Gamma}}_\phi-{\Gamma}_\phi)(\textbf{g}_1)\|_\infty$. 
\end{itemize}
\paragraph{Proof of the claims}
\begin{itemize}
    \item  Recall that we want to establish an upper bound on $\|(\widehat{{\Gamma}}_\phi-{\Gamma}_\phi)(\textbf{g}_1)\|_\infty$, we explicit its form:
    \[\|(\widehat{{\Gamma}}_\phi-{\Gamma}_\phi)(\textbf{g}_1)\|_\infty=\sup_{t\in[0,1],d'\in\{1,\dots,D\}}|\sum_{d=1}^D\int_0^1(\widehat{K}_\phi(s,t)-K_\phi(s,t))_{d',d}g_{1,d}(s)ds|.\]
    Before proceeding note that, as defined $K_\phi$ and $\widehat{K}_\phi$ are both piece-wise constant, i.e., for any $\lambda\in\Lambda_M$ and $t\in [\frac{\lambda}{M},\frac{\lambda+1}{M})$ we have
    \begin{align*}
        &|\sum_{d=1}^D\int_0^1(\widehat{K}_\phi(s,t)-K_\phi(s,t))_{d',d}g_{1,d}(s)ds|\\&=|\sum_{d=1}^D\int_0^1(\widehat{K}_\phi(s,\frac{\lambda}{M})-K_\phi(s,\frac{\lambda}{M}))_{d',d}g_{1,d}(s)ds|.
    \end{align*}
    Thus the supremum over all values of $t\in[0,1]$ is equal to the supremum taken over all values $\{t=\frac{\lambda}{M}\}_{\lambda\in\Lambda_M}$, which implies that:
    \[\|(\widehat{{\Gamma}}_\phi-{\Gamma}_\phi)(\textbf{g}_1)\|_\infty=\sup_{\lambda\in\Lambda_M,d'\in\{1,\dots,D\}}|\sum_{d=1}^D\int_0^1(\widehat{K}_\phi(s,\frac{\lambda}{M})-K_\phi(s,\frac{\lambda}{M}))_{d',d}g_{1,d}(s)ds|.\]
    \item   We recall the definition of $\widehat{K}_\phi$, let $\lambda_0\in\Lambda_M$ and $(d,d')\in\{1,\dots,D\}^2$:
    \begin{align*}
        (\widehat{ K}_\phi(s,\frac{\lambda_0}{M}))_{d',d}&=(\frac1n\sum_{i=1}^n  \widetilde{\textbf{Y}}_i(\frac{\lambda_0}{M})^T \widetilde{\textbf{Y}}_i(s))_{d',d}\\&=\frac{1}{n}\sum_{i=1}^n\Big(\sum_{(\lambda,\lambda')\in\Lambda_M^2}\widetilde{y}_{i,d,\lambda}\phi_{\lambda}(s)\widetilde{y}_{i,d',\lambda'}\phi_{\lambda'}(\frac{\lambda_0}{M})\Big)\\&=\frac{1}{n}\sum_{i=1}^n\Big(\sum_{\lambda\in\Lambda_M}\widetilde{y}_{i,d,\lambda}\phi_{\lambda}(s)\widetilde{y}_{i,d',\lambda_0}\sqrt{M}\Big),
    \end{align*}
    which implies that for $\lambda_0\in\Lambda_M$ and $d'\in\{1,\dots,D\}$:
        \begin{align*}
     \sum_{d=1}^D\int_0^1(\widehat{K}_\phi(s,\frac{\lambda}{M}))_{d',d}g_{1,d}(s)ds&=\sum_{d=1}^D\int_0^1\frac{1}{n}\sum_{i=1}^n\Big(\sum_{\lambda\in\Lambda_M}\widetilde{y}_{i,d,\lambda}\phi_{\lambda}(s)\widetilde{y}_{i,d',\lambda_0}\sqrt{M}\Big)g_{1,d}(s)ds\\&=\sum_{d=1}^D\frac{1}{n}\sum_{i=1}^n\sum_{\lambda\in\Lambda_M}\widetilde{y}_{i,d',\lambda_0}\sqrt{M}\widetilde{y}_{i,d,\lambda}\langle\phi_{\lambda},g_{1,d}\rangle\\&=\frac{1}{n}\sum_{i=1}^n\widetilde{y}_{i,d',\lambda_0}\sqrt{M}\Big(\sum_{d=1}^D\sum_{\lambda\in\Lambda_M}\widetilde{y}_{i,d,\lambda_0}\langle\phi_{\lambda},g_{1,d}\rangle\Big).
  \end{align*}
    Denoting by $I_{\Lambda_M}(\textbf{g}_1)$ the vector of size $MD$ containing all the values $(\langle\phi_{\lambda},g_{1,d}\rangle)_{\lambda\in\Lambda_M,d=1,\dots,D}$, i.e., let $j\in\{1,\dots,MD\}$ and $r$ the rest of the euclidean division of $j$ by $D$ and $q$ the quotient:
    \[(I_{\Lambda_M}(\textbf{g}_1))_j=\langle\phi_{q},g_{1,r}\rangle.\]
    As defined $I_{\Lambda_M}(\textbf{g}_1)$ has the following properties:
    \begin{itemize}
        \item We have the following upper bound on $ \|I_{\Lambda_M}(\textbf{g}_1)\|_{\ell_2}^2$: \begin{align*}
        \|I_{\Lambda_M}(\textbf{g}_1)\|_{\ell_2}^2=\sum_{\lambda\in\Lambda_M}\sum_{d=1}^D\langle\phi_{\lambda},g_{1,d}\rangle^2=\|\Pi_{S_M^D}(\textbf{g}_1)\|_\mathbb{H}^2\leq \|\textbf{g}_1\|_\mathbb{H}^2,
    \end{align*}
        \item Since $\widetilde{\textbf{y}}_1\sim\mathcal{N}(0,\frac{1}{p}\psi^T\Sigma\psi)$ and $\langle I_{\Lambda_M}(\textbf{g}_1),\widetilde{\textbf{y}}_i\rangle_{\ell_2}=\sum_{d=1}^D\sum_{\lambda\in\Lambda_M}\widetilde{y}_{i,d,\lambda}\langle\phi_{\lambda},g_{1,d}\rangle$ it implies that: 
        $$\langle I_{\Lambda_M}(\textbf{g}_1),\widetilde{\textbf{y}}_i\rangle_{\ell_2}\sim\mathcal{N}\Big(0,\frac{1}{p}\langle\psi I_{\Lambda_M}(\textbf{g}_1)\Sigma,\psi I_{\Lambda_M}(\textbf{g}_1)\rangle_{\ell_2}\Big).$$
        
    \end{itemize}
    Thus,
       \begin{align*}
          &\sum_{d=1}^D\int_0^1(\widehat{K}_\phi(s,\frac{\lambda}{M})-K_\phi(s,\frac{\lambda}{M}))_{d',d}g_{1,d}(s)ds\\&=\frac{1}{n}\sum_{i=1}^n\widetilde{y}_{i,d',\lambda_0}\sqrt{M}\underbrace{\Big(\sum_{d=1}^D\sum_{\lambda\in\Lambda_M}\widetilde{y}_{i,d,\lambda}\langle\phi_{\lambda},g_{1,d}\rangle\Big)}_{\langle I_{\Lambda_M}(\textbf{g}_1),\widetilde{\textbf{y}}_i\rangle_{\ell_2}}-\mathbb{E}[\widetilde{y}_{1,d',\lambda_0}\sqrt{M}\underbrace{\Big(\sum_{d=1}^D\sum_{\lambda\in\Lambda_M}\widetilde{y}_{1,d,\lambda}\langle\phi_{\lambda},g_{1,d}\rangle}_{\langle I_{\Lambda_M}(\textbf{g}_1),\widetilde{\textbf{y}}_i\rangle_{\ell_2}}\Big)]\\&=\frac{1}{n}\sum_{i=1}^n\langle I_{\Lambda_M}(\textbf{g}_1),\widetilde{\textbf{y}}_i\rangle_{\ell_2}\sqrt{M}\widetilde{y}_{i,d',\lambda_0}-\mathbb{E}[\langle I_{\Lambda_M}(\textbf{g}_1),\widetilde{\textbf{y}}_1\rangle_{\ell_2}\sqrt{M}\widetilde{y}_{1,d',\lambda_0}].
       \end{align*}
     In the sequel, we compute the moments of $\langle I_{\Lambda_M}(\textbf{g}_1),\widetilde{\textbf{y}}_1\rangle_{\ell_2}\sqrt{M}\widetilde{y}_{1,d',\lambda_0}$. Let $k\in\mathbb{N}^*$ we have:
     $$\mathbb{E}[|\langle I_{\Lambda_M}(\textbf{g}_1),\widetilde{\textbf{y}}_1\rangle_{\ell_2}\sqrt{M}\widetilde{y}_{1,d',\lambda_0}|^k]\leq\sqrt{\mathbb{E}[|\langle I_{\Lambda_M}(\textbf{g}_1),\widetilde{\textbf{y}}_1\rangle_{\ell_2}|^{2k}]\mathbb{E}[|\sqrt{M}\widetilde{y}_{1,d',\lambda_0}|^{2k}]}.$$
     Furthermore, since both variables are Gaussian, we know that:
     
     $$\mathbb{E}[|\langle I_{\Lambda_M}(\textbf{g}_1),\widetilde{\textbf{y}}_1\rangle_{\ell_2}|^{2k}]=\frac{(2k)!}{2^kk!}\mathbb{E}[\langle I_{\Lambda_M}(\textbf{g}_1),\widetilde{\textbf{y}}_1\rangle_{\ell_2}^2]^k,\quad \mathbb{E}[|\widetilde{y}_{1,d',\lambda_0}|^{2k}]=\frac{(2k)!}{2^kk!}\mathbb{E}[|\widetilde{y}_{1,d',\lambda_0}|^{2}]^k.$$
     Finally, since by \eqref{norm of Sigma} we showed that $\|\Sigma\|_2\leq p\widetilde{\mu}_1+\sigma^2$, it implies that:
     \begin{align*}
         \mathbb{E}[\langle I_{\Lambda_M}(\textbf{g}_1),\widetilde{\textbf{y}}_1\rangle_{\ell_2}^2]&\leq \frac{1}{p}\langle\psi I_{\Lambda_M}(\textbf{g}_1)\Sigma,\psi I_{\Lambda_M}(\textbf{g}_1)\rangle_{\ell_2}\\&\leq\frac{\|\Sigma\|_2}{p}\|\psi I_{\Lambda_M}(\textbf{g}_1)\|_{\ell_2}^2\\&\leq(\widetilde{\mu}_1+\frac{\sigma^2}{p})\| \textbf{g}_1\|_\mathbb{H}^2,
     \end{align*}
    and
    \begin{align*}
        \mathbb{E}[\widetilde{y}_{1,d',\lambda_0}^2]&=\frac{1}{p^2}\sum_{h,h'=0}^{p-1}\mathbb{E}[Y_{1,d'}(t_h)Y_{1,d'}(t_{h'})]\phi_{\lambda_0}(t_h)\phi_{\lambda_0}(t_{h'})\\&\leq (\|K\|_\infty+\sigma^2)\frac{M}{p^2}\sum_{h,h'=0}^{p-1}1_{t_{h'}\in I_{\lambda_0}}1_{t_h\in I_{\lambda_0}} \\&=\frac{\|K\|_\infty+\sigma^2}{M}.
    \end{align*}
Thus we have:
\begin{align*}
    \mathbb{E}[|\langle I_{\Lambda_M}(\textbf{g}_1),\widetilde{\textbf{y}}_i\rangle_{\ell_2}\sqrt{M}\widetilde{y}_{1,d}|^k]&\leq \frac{(2k)!}{2^kk!}\Big(\sqrt{\|\textbf{g}_1\|_\mathbb{H}^2(\widetilde{\mu}_1+\frac{\sigma^2}{p})(\|K\|_\infty+\sigma^2)}\Big)^k.
    \end{align*}
    Hence, using the fact that $(2k)!=2^kk!\prod_{j=1}^k(2j-1)\leq 2^k k!\prod_{j=1}^k(2j)=(2^k k!)^2$, we have
    \begin{align*}
    \mathbb{E}[|\langle I_{\Lambda_M}(\textbf{g}_1),\widetilde{\textbf{y}}_i\rangle_{\ell_2}\sqrt{M}\widetilde{y}_{1,d}|^k]&\leq 2^k k!\|\textbf{g}_1\|_\mathbb{H}^k\Big(\sqrt{(\widetilde{\mu}_1+\frac{\sigma^2}{p})(\|K\|_\infty+\sigma^2)}\Big)^k.
    \end{align*}
    Finally,we have.
    \begin{align*}
    \mathbb{E}[|\langle I_{\Lambda_M}(\textbf{g}_1),\widetilde{\textbf{y}}_i\rangle_{\ell_2}\sqrt{M}\widetilde{y}_{1,d}|^k]
    \leq \frac{k!}{2}v^2c^{k-2},
\end{align*}
where $v=4\|\textbf{g}_1\|_\mathbb{H}\sqrt{(\widetilde{\mu}_1+\frac{\sigma^2}{p})(\|K\|_\infty+\sigma^2)}$ and $c=2\|\textbf{g}_1\|_\mathbb{H}\sqrt{(\widetilde{\mu}_1+\frac{\sigma^2}{p})(\|K\|_\infty+\sigma^2)}$.

\item Let $\lambda\in\Lambda_M$ and $d'\in\{1,\dots,D\}$, using Bernstein inequality of Theorem \ref{thm:bernstein} we have for any $t>0$
$$P\left(|\sum_{d=1}^D\int_0^1(\widehat{K}_\phi(s,\frac{\lambda}{M})-K_\phi(s,\frac{\lambda}{M}))_{d',d}g_{1,d}(s)ds|\geq \sqrt{2v^2t}+ct\right)\leq 2\exp(-nt),$$
and using the union bound we get for any $t>0$

$$P\left(\|(\widehat{{\Gamma}}_\phi-{\Gamma}_\phi)(\textbf{g}_1)\|_\infty\geq \sqrt{2v^2t}+ct\right)\leq \sum_{d=1}^D\sum_{\lambda\in\Lambda_M} 2\exp\left(-nt\right)=2MD\exp(-nt).$$
Taking $t=\frac{2\log(MD)}{n}$, we have with probability at least $1-\frac{2}{MD}$:
\begin{align*}
    \|(\widehat{{\Gamma}}_\phi-{\Gamma}_\phi)(\textbf{g}_1)\|_\infty&\leq [\sqrt{\frac{4v^2\log(DM)}{n}}+\frac{2c\log(DM)}{n}]\\&\leq 2c[4\sqrt{\frac{\log(DM)}{n}}+\frac{\log(DM)}{n}]
\end{align*}

\end{itemize}

\begin{definition}\label{def:subgaussian:design}
      We say that a vector $\textbf{Y}\in\mathbb{R}^{pD}$ is sub-Gaussian with parameter $\zeta$ if for all vectors $u\in\mathbb{R}^{pD}$ such that $\|u\|_{\ell^2}=1$, it holds
      
      \begin{equation}\label{eq:subgaussian:design}
          \mathbb{E}[\exp(\frac{\langle u,\textbf{Y}\rangle_{\ell_2}^2}{\zeta^2})]\leq 2. 
      \end{equation}
\end{definition}
\begin{lemma}\label{lemma:10}
Assuming $Y\in\mathbb{R}^{n\times pD}$ is a sub-Gaussian matrix, with parameter $\zeta$. Let $J:=\log(T)$, and 
\begin{align*}
    \lambda_0 = \sqrt{\frac{2\log(2pD)}{n}}
\end{align*}
Then with probability at least $1-2J\exp(-\log(2pD))$, it holds
\begin{align*}
    \forall \theta \text{ such that }, 1\leq\|\theta\|_1\leq T:
    |\theta^T(\frac{1}{n}\sum_{i=1}^nY_i^TY_i-\mathbb{E}[Y^TY]) \theta|\leq \mathcal{Q}(\|\theta\|_1^2,\zeta)\|\theta\|_{\ell_2}^2,
\end{align*}
where $\mathcal{Q}(x,\zeta):=4\times 27\zeta^2\big[3x^2\lambda_0^2+\sqrt{6}x\lambda_0\big]$.
\end{lemma}
This lemma is essentially Lemma 10 in \cite{VAN}, but we use only the first half of the result, i.e., only for $\theta\in\mathbb{R}^{pD},1\leq\|\theta\|_1\leq T$. We recall the statement of Lemma \ref{lemma:most:important}.\\
{\bf Lemma~\ref{lemma:most:important} }{ \it Let $J:=\log(T)$ and $\lambda_0=\sqrt{\frac{2\log(2pD)}{n}}$. Then with probability at least $1-2\frac{J+1}{pD}$, it holds for all $ \textbf{g}\in\mathcal{B}(\eta)  ,\|\textbf{g}\|_1\leq T$:
\begin{equation*}
 |\langle(\widehat{\Gamma}_\phi-\Gamma_\phi)(\textbf{g}),\textbf{g}\rangle_\mathbb{H}|\leq \lambda_1\|\textbf{g}\|_\mathbb{H}+\mathcal{Q}(\|\textbf{g}\|_1^2,3\sqrt{\widetilde{\mu}_1+\frac{\sigma^2}{p}})\|\textbf{g}\|_\mathbb{H}^2
\end{equation*}
}

\paragraph{Proof of the claims}
The proof is based on the application of Lemma \ref{lemma:10} to our situation, to do that we need to establish the following :
\begin{itemize}
\item As defined in Definition \ref{def:subgaussian:design}, \textbf{Y} is said to be sub-Gaussian with parameter $\zeta$ if 

$$\mathbb{E}\Big[\exp\big(\frac{\langle \textbf{Y},a\rangle_{\ell_2}^2}{\zeta^2}\big)\Big]\leq 2.$$
In the sequel, we will show that \textbf{Y} is sub-Gaussian with parameter $\zeta=3\sqrt{\widetilde{\mu}_1p+\sigma^2}$. First note that since $\textbf{Y}_1$ is Gaussian with covariance $\Sigma$ it implies that $\frac{\langle \textbf{Y}_1,a\rangle_{\ell_2}}{\zeta} \sim\mathcal{N}(0,\frac{\langle a,\Sigma a\rangle_{\ell_2}}{\zeta^2})$. We know that for any $\zeta>\frac{1}{2\sqrt{\langle a,\Sigma a\rangle_{\ell_2}}}$:

$$\mathbb{E}\Big[\exp\big(\frac{\langle \textbf{Y}_1,a\rangle_{\ell_2}^2}{\zeta^2}\big)\Big]= \frac{1}{\sqrt{1-2\frac{\sqrt{\langle a,\Sigma a\rangle_{\ell_2}}}{\zeta}}}.$$
As defined note that we already established that $\|\Sigma\|_2\leq p\widetilde{\mu}_1+\sigma^2$ in \eqref{norm of Sigma},
\begin{align*}
   \frac{\sqrt{\langle a,\Sigma a\rangle_{\ell_2}}}{\zeta}&\leq\frac{\sqrt{\|\Sigma\|_2\|a\|_2^2}}{\zeta}\\&\leq\frac{\sqrt{p\widetilde{\mu}_1+\sigma^2}}{\zeta}\\&\leq \frac{1}{3}. 
\end{align*}
Thus :
$$\mathbb{E}\Big[\exp\big(\frac{\langle \textbf{Y}_1,a\rangle_{\ell_2}^2}{\zeta^2}\big)\Big]= \frac{1}{\sqrt{1-2\frac{\sqrt{\langle a,\Sigma a\rangle_{\ell_2}}}{\zeta}}}\leq\sqrt{3}\leq2 .$$

    \item Recall the we are on the subset of $\mathbb{H}$ such that $\|\textbf{g}\|_1\leq 1$. Since we have $|(\widehat{\Gamma}_\phi-\Gamma_\phi)(\textbf{g}),\textbf{g}\rangle_\mathbb{H}|\leq \|(\widehat{\Gamma}_\phi-\Gamma_\phi)(\textbf{g})\|_\infty\|\textbf{g}\|_1$, and since by Lemma \ref{Lemma concentration}, we know that with probability at least $1-\frac{2}{M  D} $ we have $\|(\widehat{\Gamma}_\phi-\Gamma_\phi)(\textbf{g})\|_\infty\leq \lambda_1\|\textbf{g}\|_\mathbb{H}$. Thus:
    
    $$|(\widehat{\Gamma}_\phi-\Gamma_\phi)(\textbf{g}),\textbf{g}\rangle_\mathbb{H}|\leq \lambda_1\|\textbf{g}\|_\mathbb{H}\|\textbf{g}\|_1\leq\lambda_1\|\textbf{g}\|_\mathbb{H}.$$
    
    \item Note that :
    \begin{align*}
    &\langle(\widehat{{\Gamma}}_\phi-{\Gamma}_\phi){\textbf{g}},{\textbf{g}}\rangle_\mathbb{H}\\&=\sum_{d,d'=1}^D\sum_{\lambda,\lambda'\in\lambda_M}\Big(\frac{1}{p^2}\sum_{h,h'=1}^p\phi_\lambda(t_h)\phi_{\lambda'}(t_{h'})\big(\sum_{i=1}^n\frac{Y_{i,d}(t_h)Y_{i,d'}(t_{h'})}{n}-K(t_h,t_{h'})\big)\Big)\\&\times\langle\phi_{\lambda},g_{d}\rangle\langle\phi_{\lambda'},g_{d'}\rangle\\&=\frac{M}{p}\langle\psi^T(\widehat{\Sigma}-\Sigma)\psi a_\textbf{g},a_\textbf{g})\rangle_{\ell_2},
\end{align*}
where $(a_\textbf{g})_{\lambda,d}:=\frac{\langle\phi_\lambda,g_d\rangle}{\sqrt{M}}$. As defined, the vector $a_\textbf{g}$ has a the following properties :
\begin{align*}
    \|a_\textbf{g}\|_1&=\frac{1}{\sqrt{M}}\sum_{d=1}^D\sum_{\lambda\in\Lambda_M}|\langle\phi_\lambda,g_d\rangle|\\&=\frac{1}{\sqrt{M}}\sum_{d=1}^D\sum_{\lambda\in\Lambda_M}\sqrt{M}|\langle 1_{I_\lambda},g_d\rangle|\\&\leq \sum_{d=1}^D\sum_{\lambda\in\Lambda_M}\langle 1_{I_\lambda},|g_d|\rangle=\sum_{d=1}^D\|g_d\|_1=\|\textbf{g}\|_1,
\end{align*}
and
\begin{align*}
     \|a_\textbf{g}\|_{\ell_2}^2&=\frac{1}{M}\sum_{d=1}^D\sum_{\lambda\in\Lambda_M}\langle\phi_\lambda,g_d\rangle^2\\&=\frac{1}{M}\|\Pi_{S^D}(\textbf{g})\|_\mathbb{H}^2\\&\leq\frac{\|\textbf{g}\|_\mathbb{H}^2}{M}.
\end{align*}

Since $\|\textbf{g}\|_1\leq T$ it implies that $\|a_\textbf{g}\|_1\leq T$. Then Lemma \ref{lemma:10} gives the following upper bound. With probabilty at least $1-2\frac{\log(T)}{pD}$,
\begin{align*}
    \frac{M}{p}\langle\psi^T(\widehat{\Sigma}-\Sigma)\psi a_\textbf{g},a_\textbf{g})\rangle_{\ell_2}\leq \frac{M}{p}\mathcal{Q}(\|\psi a_\textbf{g}\|_1^2,3\sqrt{p\widetilde{\mu}_1+\sigma^2})\|\psi a_\textbf{g}\|_{\ell_2}^2.
\end{align*}

Since $\psi^T\psi=I_{MD}$ we know that $\|\psi\|_2=1$. To conclude note that $\|\psi a_\textbf{g}\|_{\ell_2}^2\leq\|\psi\|_2^2\|a_\textbf{g}\|_{\ell_2}^2\leq\|a_\textbf{g}\|_{\ell_2}^2\leq \frac{\|\textbf{g}\|_\mathbb{H}^2}{M}$, $\|\psi a_\textbf{g}\|_1^2\leq\|\psi\|_2^2\|a_\textbf{g}\|_1^2\leq\|a_\textbf{g}\|_1^2\leq\|\textbf{g}\|_1^2$ and the function $\mathcal{Q}(\cdot,\cdot)$ is increasing in its first variable thus 
  \begin{align*}
     |\langle(\widehat{{\Gamma}}_\phi-{\Gamma}_\phi){\textbf{g}},{\textbf{g}}\rangle_\mathbb{H}|&\leq \frac{M}{p}\mathcal{Q}(\|\textbf{g}\|_1^2,3\sqrt{p\widetilde{\mu}_1+\sigma^2})\frac{\|\textbf{g}\|_\mathbb{H}^2}{M}.
  \end{align*}   
  Finally note that $\frac{1}{p}\mathcal{Q}(\|\textbf{g}\|_1^2,3\sqrt{p\widetilde{\mu}_1+\sigma^2})=\mathcal{Q}(\|\textbf{g}\|_1^2,3\sqrt{\widetilde{\mu}_1+\frac{\sigma^2}{p}})$, we have with probability at least $1-2\frac{\log(T)}{pD}$:
       \begin{align*}
     |\langle(\widehat{{\Gamma}}_\phi-{\Gamma}_\phi){\textbf{g}},{\textbf{g}}\rangle_\mathbb{H}|&\leq \mathcal{Q}(\|\textbf{g}\|_1^2,3\sqrt{\widetilde{\mu}_1+\frac{\sigma^2}{p}})\|\textbf{g}\|_\mathbb{H}^2.
  \end{align*}   
\end{itemize}
Thus 
\begin{align*}
    &P(\exists \textbf{g}:\|\textbf{g}\|_1\leq T\quad |\langle(\widehat{{\Gamma}}_\phi-{\Gamma}_\phi){\textbf{g}},{\textbf{g}}\rangle_\mathbb{H}|\geq\lambda_1\|\textbf{g}\|_1+\mathcal{Q}(\|\textbf{g}\|_1^2,3\sqrt{\widetilde{\mu}_1+\frac{\sigma^2}{p}})\|\textbf{g}\|_\mathbb{H}^2)\\&\leq P(\exists \textbf{g}:\|\textbf{g}\|_1\leq 1\quad |\langle(\widehat{{\Gamma}}_\phi-{\Gamma}_\phi){\textbf{g}},{\textbf{g}}\rangle_\mathbb{H}|\geq\lambda_1\|\textbf{g}\|_1)\\&+ P(\exists \textbf{g}:1\leq\|\textbf{g}\|_1\leq T\quad |\langle(\widehat{{\Gamma}}_\phi-{\Gamma}_\phi){\textbf{g}},{\textbf{g}}\rangle_\mathbb{H}|\geq\mathcal{Q}(\|\textbf{g}\|_1^2,3\sqrt{\widetilde{\mu}_1+\frac{\sigma^2}{p}})\|\textbf{g}\|_\mathbb{H}^2)\\&\leq \frac{2\log(T)}{pD}+\frac{2}{MD}\leq 2\frac{\log(T)+1}{MD}
\end{align*}
which ends the proof of Lemma \ref{lemma:most:important}.

\subsection{Proof of lemma \ref{Lemma hermitian}}\label{proof:Lemma hermitian}
    The proof is inspired by Lemma 12.7 from \cite{bookvan}, we show here a functional counterpart of their result.
\paragraph{Proof}Proposition \ref{prop:differential:multi} gives
\begin{align*}
    \ddot{R}(\textbf{g})&=4(-\Gamma +\|\textbf{g}\|_\mathbb{H}^2I+2\textbf{g}\otimes \textbf{g})\\
               &=4(-\sum_{\ell\in\mathbb{N}^*}(\sqrt{\mu}_\ell)^2\textbf{f}_\ell\otimes\textbf{f}_\ell+\|\textbf{g}\|_\mathbb{H}^2\sum_{\ell\in{N}^*}\textbf{f}_\ell\otimes\textbf{f}_\ell+2\textbf{g}\otimes \textbf{g})\\
               &=4[(\|\textbf{g}\|_\mathbb{H}^2-(\sqrt{\mu}_1)^2)\textbf{f}_1\otimes\textbf{f}_1+\sum_{i\geq 2}(\|\textbf{g}\|_\mathbb{H}^2-(\sqrt{\mu}_\ell)^2)\textbf{f}_\ell\otimes\textbf{f}_\ell+2\textbf{g}\otimes \textbf{g}],
\end{align*} 
where $(\textbf{f}_\ell)_{\ell\in\mathbb{N}^*}$ are the eigenfunctions and $(\mu_\ell)_{\ell\in\mathbb{N}^*}$ the eigenvalues of $\Gamma$. Since by assumption $\|\textbf{g}-\textbf{g}_1\|_\mathbb{H} <\eta$ and $\|\textbf{g}_1\|_\mathbb{H}=\sqrt{\mu_1}\geq \rho \geq \frac{\eta}{8}$ it holds that :
$$\|\textbf{g}\|_\mathbb{H}\geq \|\textbf{g}_1\|_\mathbb{H}-\eta = \sqrt{\mu}_1-\eta,$$
it follows that :
$$\|\textbf{g}\|_\mathbb{H}^2\geq \mu_1-2\eta\sqrt{\mu}_1.$$

Hence for all $\ell\geq 2$ for $\rho=\sqrt{\mu_1}-\sqrt{\mu_2}$ we have $\sqrt{\mu}_1\geq \sqrt{\mu}_\ell+\rho$ which implies the following:
\begin{align*}
  \|\textbf{g}\|_\mathbb{H}^2-\mu_\ell&\geq \mu_1-\mu_\ell-2\eta\sqrt{\mu}_1\\&=(\sqrt{\mu_1}-\sqrt{\mu}_\ell)(\sqrt{\mu_1}+\sqrt{\mu_\ell})-2\eta\sqrt{\mu}_1\\&\geq (\rho-2\eta)\sqrt{\mu}_1.  
\end{align*}
Moreover, for all $\textbf{x}\in \mathbb{H}$
\begin{align*}
  \langle \textbf{x},\textbf{g}\otimes\textbf{g}(\textbf{x})\rangle_\mathbb{H}& =\langle \textbf{x},\textbf{g}\rangle_\mathbb{H}^2\\&= (\langle \textbf{x},\textbf{g}-\textbf{g}_1\rangle_\mathbb{H}+\langle \textbf{x},\textbf{g}_1\rangle_\mathbb{H})^2\\
   			&=(\langle \textbf{x},\textbf{g}-\textbf{g}_1\rangle_\mathbb{H}^2+\langle \textbf{x},\textbf{g}_1\rangle_\mathbb{H}^2+2\langle \textbf{x},\textbf{g}_1\rangle_\mathbb{H}\langle \textbf{x},\textbf{g}-\textbf{g}_1\rangle_\mathbb{H})\\&\geq\langle \textbf{x},\textbf{g}_1\rangle_\mathbb{H}^2+2\langle \textbf{x},\textbf{g}_1\rangle_\mathbb{H}\langle \textbf{x},\textbf{g}-\textbf{g}_1\rangle_\mathbb{H}\\&\geq\langle \textbf{x},\textbf{g}_1\rangle_\mathbb{H}^2-2\| \textbf{x}\|_\mathbb{H}\|\textbf{g}_1\|_\mathbb{H} \|\textbf{x}\|_\mathbb{H}\|\textbf{g}-\textbf{g}_1\|_\mathbb{H}\\
   			&\geq \langle \textbf{x},\textbf{g}_1\rangle_\mathbb{H}^2-2\sqrt{\mu}_1\eta\|\textbf{x}\|_\mathbb{H}^2,
\end{align*}
and 

\begin{align*}
    \langle \textbf{x},(\|\textbf{g}\|_\mathbb{H}^2-\mu_1\ )\textbf{f}_1\otimes\textbf{f}_1 (\textbf{x})\rangle_\mathbb{H}&=(\|\textbf{g}\|_\mathbb{H}^2-\mu_1)\langle\textbf{f}_1,\textbf{x}\rangle^2_\mathbb{H}\\&\geq  -2\eta\sqrt{\mu}_1\|\textbf{x}\|_\mathbb{H}^2.
\end{align*}
Since $\sqrt{\mu}_1\geq \rho\geq \rho-\eta$ we thus see that 

\begin{align*}
\langle \textbf{x},\ddot{R}(\textbf{g})\textbf{x}\rangle_\mathbb{H}&=4\langle \textbf{x},[(\|\textbf{g}\|_\mathbb{H}^2-(\sqrt{\mu}_1)^2)\textbf{f}_1\otimes\textbf{f}_1(\textbf{x})\rangle_\mathbb{H}\\&+\sum_{\ell\geq 2}\langle \textbf{x},(\|\textbf{g}\|_\mathbb{H}^2-(\sqrt{\mu}_\ell)^2)\textbf{f}_\ell\otimes\textbf{f}_\ell(\textbf{x})\rangle_\mathbb{H}+2\langle \textbf{g},\textbf{x}\rangle_\mathbb{H}^2]\\&\geq 4[-2\eta\sqrt{\mu}_1\|\textbf{x}\|_\mathbb{H}^2+\sqrt{\mu}_1(\rho-2\eta)\sum_{\ell\geq 2}\langle \textbf{x},\textbf{f}_\ell\otimes\textbf{f}_\ell(\textbf{x})\rangle_\mathbb{H}\\&+2(\langle \textbf{x},\textbf{g}_1\rangle_\mathbb{H})^2-4\sqrt{\mu}_1\eta\|\textbf{x}\|_\mathbb{H}^2].
\end{align*}
Recall that $(\langle \textbf{x},\textbf{g}_1\rangle_\mathbb{H})^2\geq\mu_1 (\langle \textbf{x},\textbf{f}_1\rangle_\mathbb{H})^2\geq \sqrt{\mu_1}(\rho-2\eta) (\langle \textbf{x},\textbf{f}_1\rangle_\mathbb{H})^2$ which implies the following: 
\begin{align*}
\langle \textbf{x},\ddot{R}(\textbf{g})\textbf{x}\rangle_\mathbb{H}&\geq4[-2\eta\sqrt{\mu}_1\|\textbf{x}\|_\mathbb{H}^2+\sqrt{\mu}_1(\rho-2\eta)\sum_{\ell\geq 1}\langle \textbf{x},\textbf{f}_\ell\rangle_\mathbb{H}^2-4\sqrt{\mu}_1\eta\|\textbf{x}\|_\mathbb{H}^2] \\&\geq 4\sqrt{\mu}_1[(\rho-8\eta)\|\textbf{x}\|_\mathbb{H}^2.
\end{align*}

\bibliographystyle{abbrvnat}
\bibliography{biblio}

\end{document}